\documentclass[12pt]{article}
\usepackage{mathrsfs}
\usepackage{amsfonts,amsmath,latexsym,amsthm}
\usepackage{graphicx} 
\usepackage{xr-hyper}
\usepackage{url, hyperref}
\usepackage{setspace}
\usepackage{natbib}
\bibpunct{(}{)}{;}{a}{}{,}
\usepackage{multirow}  

\usepackage{enumitem}

\hypersetup{
    colorlinks,%
    citecolor=blue,%
    filecolor=black,%
    linkcolor=blue,%
    urlcolor=blue
}

\usepackage{setspace}

\newtheorem{theorem}{Theorem}
\newtheorem{corollary}{Corollary}
\newtheorem{lemma}{Lemma}
\newtheorem{proposition}{Proposition}

\long\gdef\blind#1#2{\ifbld{#1}\else{#2}\fi} 

\newif\ifbld \bldfalse

\newcommand{\vect}[1]{\mbox{\boldmath $ #1$}}
\def\bbR{{\mathbb R}}
\def\bbN{{\mathbb N}}
\def\E{{\mathbb E}}

\def\MM{ {\mathcal{M}} }
\def\MT{ {\mathcal{M}_T} }
\def\J{ {\mathcal{J}} }

\def\b{\boldsymbol{\beta}}

\def\e{\boldsymbol{\epsilon}}

\newcommand{\Emptyset}{\text{\o}}
\def\Mnull{{\cal M}_{\Emptyset}}

\def\g{\boldsymbol{\gamma}}
\def\eg{\boldsymbol{\eta}_{\MM}}

\def \s2{\sigma^2}
\def\X{\mathbf{X}}
\def\x{\mathbf{x}}

\def\Y{\mathbf{Y}}

\def\one{\mathbf{1}}
\def\zero{\mathbf{0}}
\newcommand{\I}{\mathbf{I}}
\renewcommand{\P}{\mathcal{P}}
\def\uv{\mathbf{u}}

\def\W{\mathbf{W}}
\def\vv{\mathbf{v}}

\def\N{\textsf{N}}
\def\BF{\textsf{BF}}

\def\diag{\textsf{diag}}
\def\rank{\rho}
\def\convd{\stackrel{D}{\longrightarrow}} 

\newpage

\title{Mixtures of $g$-Priors in Generalized Linear Models} \date{}
\author{Yingbo Li\footnote{Yingbo Li, Department of Mathematical Sciences, Clemson University, Clemson,
      SC 29634 (e-mail: {\tt carolli13@gmail.com})}  and Merlise
    A.~Clyde\footnote{Merlise A. Clyde is Professor and Chair,
      Department of Statistical Science, Duke University, Durham, NC
      27708 (e-mail: {\tt clyde@duke.edu})}  }

\begin{document}

\maketitle


\begin{abstract}

  Mixtures of Zellner's $g$-priors have been studied extensively in
  linear models and have been shown to have numerous desirable
  properties for Bayesian variable selection and model averaging.
  Several extensions of $g$-priors to Generalized Linear Models (GLMs)
  have been proposed in the literature; however, the choice of prior
  distribution of $g$ and resulting properties for inference have
  received considerably less attention. In this paper, we unify
  mixtures of $g$-priors in GLMs by assigning the truncated Compound
  Confluent Hypergeometric (tCCH) distribution to 
  $1/(1+g)$, which encompasses as special cases several mixtures of
  $g$-priors in the literature, such as the hyper-$g$, Beta-prime,
  truncated Gamma, incomplete inverse-Gamma, benchmark, robust,
  hyper-$g/n$, and intrinsic priors.  Through an integrated Laplace
  approximation, the posterior distribution of $1/(1+g)$ is in turn a
  tCCH distribution, and approximate marginal likelihoods are thus
  available analytically, leading to ``Compound Hypergeometric
  Information Criteria'' for model selection.  We discuss the local
  geometric properties of the $g$-prior in GLMs and show how the
  desiderata for model selection proposed by Bayarri et al, such as
  asymptotic model selection consistency, intrinsic consistency, and
  measurement invariance may be used to justify the prior and specific
  choices of the hyper parameters.  We illustrate inference using
  these priors and contrast them to other approaches  via
  simulation and real data examples.  \blind{An {\tt R} package on
    CRAN is available to implement the methodology.}{ The methodology
    is implemented in the {\tt R} package {\tt BAS} and freely
    available on {\tt CRAN}.}

\end{abstract}

\noindent%



\section{Introduction} 
Subjective elicitation of prior distributions for variable selection,
quickly becomes intractable as the number of potential
variables $p$ increases, motivating objective or conventional prior
distributions for default usage \citep{Berger_Pericchi_2001}.  In the
context of linear models, Zellner's $g$-prior and
mixtures of $g$-priors have witnessed widespread use due to
computational tractability, consistency, invariance, and other
desiderata \citep{Liang_etal_2008, Bayarri_etal_2012, Ley_Steel_2012},
leading to the preference of these priors over many other
conventional prior distributions \citep{Forte_et_al_2016}.

\citet{Zellner_1983, Zellner_1986} proposed the $g$-prior as a simple
partially informative distribution in Gaussian regression models
$\Y = \X \b + \e$, $\e \sim \N(0, \sigma^2 \mathbf{I}_n)$.  Through
the use of imaginary responses at the observed design matrix $\X$, he
obtained a conjugate Gaussian prior distribution
$\b \mid \sigma^2 \sim \N(\vect{b}_0, g \sigma^2 (\X^T\X)^{-1})$, with
an informative mean $\vect{b}_0$, but having a covariance matrix that
was a scaled version of the covariance matrix of the maximum
likelihood estimator\footnote{We follow the now standard notation,
  however, in Zellner's papers the prior covariance appears as
  $(\sigma^2/g) (\X^T\X)^{-1}$}, $g \sigma^2 (\X^T\X)^{-1}$.  This
greatly simplified elicitation to two quantities: the prior mean
$\vect{b}_0$ of the regression coefficients, for which practitioners
often had prior beliefs, and the scalar $g$ which controlled both
shrinkage towards the prior mean and dispersion of the posterior
covariance through the shrinkage factor $g/(1+g)$.

Using Zellner's $g$-prior for Bayesian variable selection (BVS) and
Bayesian model averaging (BMA) requires specification of the
hyper parameters $\vect{b}_0$ and $g$ for each of the $2^p$ submodels,
indexed by $\MM \in \{0, 1\}^{p}$,
\begin{equation} \label{eq:normal-model}
  \Y = \one_n \alpha + \X_\MM \b_\MM + \e
\end{equation}
where $\mathbf{1}_n$ is a
column vector of ones of length $n$, $\alpha$ is the intercept,
$\X_\MM$ is a model specific design matrix with $p_\MM$ columns 
assumed to be full rank, and $\b_\MM$ is the vector of length $p_\MM$
of the non-zero regression coefficients in model $\MM$.
The most common formulation of
Zellner's $g$-prior for BMA/BVS \citep{Fernandez_etal_2001,Liang_etal_2008}
assigns an
independent Jeffreys prior to $\alpha$ and $\sigma^2$ \vspace{-12pt}
\begin{align}
p(\alpha) & \propto 1, \label{eq:alpha-prior}\\
p(\sigma^2) & \propto 1/\sigma^2, \label{eq:sigma-prior} 
\intertext{and a $g$-prior using the centered predictors $(\I_n -
              \mathcal{P}_{\one_n})\X_\MM$ of the form }
\b_\MM \mid \sigma^2, g, \MM & \sim 
\text{N}\left( \zero_{p_\MM}, g \sigma^2 (\X_\MM^{T}(\I_n - \mathcal{P}_{\one_n})
                               \X_\MM)^{-1}  \right), \label{eq:gprior}
\end{align}
where $\mathcal{P}_{\one_n} = \one_n(\one_n^T\one_n)^{-1} \one_n^T$ is
the orthogonal projection on the space spanned by the column vector
$\one_n$.  To justify the use of the improper reference priors on the
intercept and variance \eqref{eq:alpha-prior}-\eqref{eq:sigma-prior},
it is often assumed the columns of the design matrix $\X_\MM$ must be
orthogonal to $\one_n$ so that
the expected Fisher Information is block diagonal.  In that case
$\X_\MM^{T}(\I_n - \mathcal{P}_{\one_n}) \X_\MM$ in \eqref{eq:gprior}
reduces to the more familiar
$\X_\MM^T\X_\MM$. \citet{Bayarri_etal_2012}, however, argue that
measurement invariance, which leads to \eqref{eq:alpha-prior} and
\eqref{eq:sigma-prior}, combined with predictive matching (discussed
in more detail later) lead to the form of the $g$-prior above without
the explicit centering of the predictors in the sampling model
\eqref{eq:normal-model}. Both model parameterizations, however, lead
to equivalent posterior distributions through a change of variables.

It is well known that the choice of $g$ affects  shrinkage in
estimation/prediction as well as posterior probabilities of models, 
with various approaches being put forward to determine a $g$
with desirable properties. Independent of \citeauthor{Zellner_1983},
\citet{Copas_1983, Copas_1997} arrived at $g$-priors in linear and
logistic regression by considering shrinkage of maximum likelihood
estimators (MLEs) to improve prediction and estimation, as in
James-Stein estimators, proposing empirical Bayes estimates of the
shrinkage factor to improve frequentist properties of the estimators.
Related to \citeauthor{Copas_1983}, \citet{Foster_George_1994} considered
risk and expected loss in selecting $g$, \citet{George_Foster_2000}
derived global empirical Bayes estimators, while \citet{Hansen_Yu_2003}
derived model specific local empirical Bayes estimates of $g$ from an
information theory perspective.  \citet{Fernandez_etal_2001} studied
consistency of BMA under $g$-priors in linear models, recommending 
$g= \max(p^2, n)$, which lead to Bayes
factors that behave like BIC when $g = n$ or the Risk Inflation
Criterion \citep{Foster_George_1994} when $g=p^2$.

Mixtures of $g$-priors, obtained by specifying a prior distribution on
the hyper parameter $g$ in \eqref{eq:gprior}, include the Cauchy prior
of \citet{Zellner_Siow_1980}, the hyper-$g$ and related hyper-$g/n$
priors \citep{Liang_etal_2008, Cui_George_2008}, the Beta-prime prior
\citep{Maruyama_George_2011}, the robust prior
\citep{Bayarri_etal_2012}, and the intrinsic prior
\citep{Casella_Moreno_2006, Womack_etal_2014}, among others. 
Mixtures of $g$-priors not only inherit the desirable measurement
invariance property from the $g$-prior but under a range of hyper
parameters also resolve the information paradox
\citep{Liang_etal_2008} and Bartlett's paradox \citep{Bartlett_1957,
  Lindley_1968} that occur with a fixed $g$, meanwhile leading to
asymptotic consistency for model selection and estimation and other
attractive theoretical  properties
\citep{Liang_etal_2008, Maruyama_George_2011, Bayarri_etal_2012,
  Feldkircher_Zeugner_2009, Celeux_etal_2012,
  Ley_Steel_2012, Feldkircher_2012, Fouskakis_Ntzoufras_2013}.
Furthermore, by yielding exact or analytic expressions for marginal
likelihoods in tractable forms, these mixtures of $g$-priors enjoy
most of the computational efficiency of the original $g$-prior,
permitting efficient computational algorithms for stochastic search of
the posterior distribution over the model space
\citep{Clyde_etal_2011}.

For generalized linear models (GLMs), many variants of $g$-priors have
been proposed in the literature, including \citet{Copas_1983, Copas_1997,
  Kass_Wasserman_1995, Hansen_Yu_2003, Rathbun_Fei_2006,
  Marin_Robert_2007, Wang_George_2007, Fouskakis_etal_2009,
  Gupta_Ibrahim_2009, Bove_Held_2011, Hanson_etal_2014,
  Perrakis_etal_2015,Held_etal_2015, Fouskakis_etal_2016}, with
current methods favoring adaptive estimates of $g$ via mixtures of
$g$-priors or empirical Bayes estimates of $g$.  While these priors
have a number of desirable properties, no consensus on an objective
prior has emerged for GLMs.  The seminal paper of
\citet{Bayarri_etal_2012} takes an alternative approach and explores
whether a consensus of criteria or desiderata that any objective prior
should satisfy can  be used  to identify an objective
prior, leading to their recommendation of the ``robust'' prior in Gaussian
variable selection problems.  In this article, we view $g$-priors in
GLMs through this lens seeing if the desiderata can essentially
determine an objective prior in GLMs for practical use.
 
The remainder of the article is arranged as follows.
In Section \ref{section:g-prior_GLM}, we begin by reviewing $g$-priors 
in GLMs and corresponding (approximate)
Bayes factors, and the closely
related Bayes factors based on  test statistics 
\citep{Johnson_2005, Johnson_2008, Hu_Johnson_2009,
  Held_etal_2015}. As tractable expressions are generally unavailable
in GLMs, we focus attention on using an integrated Laplace
approximation and show that $g$-priors based on observed
information lead to distributions that are closed under sampling
(conditionally conjugate).  
To unify results with linear models and $g$-priors in GLMs, in Section
\ref{section:mix_g} we introduce the truncated Compound Confluent
Hypergeometric distribution \citep{Gordy_1998a}, a flexible
generalized Beta distribution, which encompasses current mixtures of $g$-priors
as special cases.   This
leads to a new family of ``Compound Hypergeometric Information
Criteria'' or CHIC.
In Section \ref{section:criteria} we review the desiderata for model selection priors
of \citet{Bayarri_etal_2012} and use them to establish theoretical properties of
the CHIC family, which provides general
recommendations for hyper parameters.  In Section \ref{sec:examples}, we
study the BVS and BMA performance of the CHIC $g$-prior with various 
hyper parameters, using
simulation studies and  the GUSTO-I data
\citep{Steyerberg_2009, Held_etal_2015}. Finally in Section
\ref{section:conclusion}, we summarize recommendation and discuss
directions for future research.


\section{$g$-Priors in  Generalized Linear Models}\label{section:g-prior_GLM} 
To begin we define notation and assumptions for the generalized linear
models (GLMs) under consideration.
GLMs arise from distributions within the exponential family
\citep{McCullagh_Nelder_1989}, with density
\begin{equation}\label{eq:likelihood}
p(Y_i) =  
\exp\left\{ \frac{Y_i  \theta_i -b\left(\theta_i \right)}{a(\phi_0)} + c(Y_i, \phi_0) \right\},
\quad  i = 1, \ldots, n,
\end{equation}
where $a(\cdot), b(\cdot)$ and $c(\cdot, \cdot)$ are specific functions that determine 
the distribution.
The mean and variance for each observation $Y_i$ can be written 
as $\mathbb{E}(Y_i) =  b^{\prime}(\theta_i)$ and 
$\mathbb{V}(Y_i) =a(\phi_0) b^{\prime\prime}(\theta_i)$, respectively,
where $b^{\prime}(\cdot)$ and $b^{\prime\prime}(\cdot)$ are the first and second derivatives of $b(\cdot)$.
In \eqref{eq:likelihood},  $Y_1, \ldots, Y_n$ are independent but not
identically distributed,
as  their corresponding canonical parameters $\theta_1, \ldots, \theta_n$ 
are linked with the predictors via $\theta_i = \theta(\eta_{\MM, i})$, 
where $\eta_{\MM, i}$ is the $i$-th entry of the linear predictor 
\begin{equation}
  \label{eq:eta}
  \boldsymbol\eta_\MM =  \mathbf{1}_n \alpha + \mathbf{X}_\MM \boldsymbol\beta_\MM  \end{equation}
under model $\MM$, providing the ``linear model''.  Under this
parameterization, the canonical link corresponds to  the identity function for $\theta(\cdot)$.

To begin, we will assume that the scale parameters are fixed, with
$a(\phi_0) = \phi_0/w_i$ with known $\phi_0$ and $w_i$, a weight that
may vary with the observation.  This includes popular GLMs such as
binary and Binomial regression, Poisson regression, and
heteroscedastic normal linear model with known variances.  Later in
Section \ref{section:mix_g}, we will relax the assumption of known
$\phi_0$ to illustrate the connections between the prior distributions
developed here and existing mixtures of $g$-priors in normal linear
models with unknown precision $\phi_0 = 1/\sigma^2$, and extend
results to consider GLMs with over-dispersion.

Unless specified otherwise, we assume that the design matrix $\X$
under the full model has full column rank $p$ and the column space
$C(\X)$ does not contain $\mathbf{1}_n$.  Furthermore, we assume that
the true model, $\MT$, is included in the $2^p$ models under
consideration.  Under $\MT$, true values of the intercept and
regression coefficients are denoted by $\alpha_{\MT}^*, \b_{\MT}^*$.
For a model $\MM$, if $\X_\MM$ contains all columns of $\X_\MT$
(including the case that $\MM = \MT$), we say $\MM \supset \MT$,
otherwise, $\MM \not \supset \MT$.  The MLEs
$\hat{\alpha}_\MM, \hat{\b}_\MM$ are assumed to exist and are
unique. Under standard regularity conditions provided in the
supplementary materials Appendix \ref{section:assumptions}, MLEs are
consistent and asymptotically normal.  In Section \ref{sec:no_MLE} we
will relax the conditions to consider non-full rank design matrices.

In BVS/BMA, posterior
probabilities of models are critical components for posterior
inference, which  
in the context of $g$-priors, may be expressed as
\[
p(\MM \mid \mathbf{Y}, g) = \frac{ p( \mathbf{Y} \mid \MM, g)~ \pi(\MM)}
{\sum_{\MM^{\prime}}p( \mathbf{Y} \mid \MM^{\prime}, g)~ \pi(\MM^{\prime})},
\]
where $\pi(\MM)$ is the prior probability of model $\MM$, and 
\begin{equation}\label{eq:marlik_def}
p(\mathbf{Y} \mid \MM, g) 
=  \iint  p(\mathbf{Y} \mid \alpha, \boldsymbol\beta_\MM, \MM)
             p(\alpha)p(\boldsymbol\beta_\MM \mid \MM, g) 
             ~d\alpha~d\boldsymbol\beta_\MM
\end{equation}
is the marginal likelihood of model $\MM$.  In normal linear
regression, $g$-priors yield closed form marginal likelihoods, which
permits quick posterior probability computation and efficient model
search, by avoiding the time-consuming procedure to sample $\alpha$
and $\boldsymbol\beta_\MM$.  When the likelihood is non-Gaussian,
normal priors no longer are conjugate, however Laplace approximations
to the likelihood \citep{Tierney_Kadane_1986, Tierney_etal_1989}
combined with normal priors such as $g$-priors may be used to achieve
computational efficiency such as in Integrated Nested Laplace
approximations \citep{Rue_etal_2009, Held_etal_2015}.


\subsection{$g$-Priors in Generalized Linear Models}\label{section:g-prior_literature}

There have been several variants of $g$-priors suggested for GLMs,
starting with \citet{Copas_1983} who proposed a normal prior centered
at zero, with a covariance based on a scaled version of the
inverse expected Fisher information evaluated at the MLE of $\alpha$
and $\b = \zero$.  Under a large sample normal approximation for the
distributions of the MLEs, this leads to conjugate updating and closed
form expressions for Bayes factors.  Unlike Gaussian models, however,
both the observed information $\mathcal{J}_n(\b_\MM)$, which is the
negative Hessian matrix of the log likelihood, and the expected Fisher
information
$\mathcal{I}_n(\b_\MM) = \mathbb{E}[\mathcal{J}_n(\b_\MM)]$, depend on
the parameters $\alpha$ and $\b$, leading to alternative $g$-priors
based on whether the expected information 
\citep{Kass_Wasserman_1995, Hansen_Yu_2003, Marin_Robert_2007,
  Fouskakis_etal_2009, Gupta_Ibrahim_2009, Bove_Held_2011,
  Hanson_etal_2014} or observed information \citep{Wang_George_2007}
is adopted;
they are equal under canonical links when evaluated at the same values.  As
these information matrices depend on $\b_\MM$,  the asymptotic
covariance is typically evaluated at either $\b_\MM = \zero$ or at the model
specific MLE. For expected information,
$\mathcal{I}_n(\b_\MM) = \mathbf{X}_\MM^T
\mathcal{I}_n(\boldsymbol\eta_\MM) \mathbf{X}_\MM$,
with $\mathcal{I}_n(\boldsymbol\eta_\MM)$ a diagonal matrix
whose $i$-th diagonal entry under model $\MM$ is
$\mathcal{I}(\eta_{\MM,i}) = -\mathbb{E}\left[ \partial^2 \log p(Y_i
  \mid \eta_i, \MM)/ \partial \eta_i^2 \right]$,
  for $i = 1, \ldots, n$.  When $\b_\MM = \zero$, all $\eta_{i} = \alpha$
 under all models, and $\mathcal{I}_n(\boldsymbol\eta_\MM)$
 is equal to  $\mathbf{I}_n / c$ where $1/c = \mathcal{I}(\eta) = -\mathbb{E}\left[ \partial^2 \log p(Y
  \mid \eta, \Mnull)/\partial \eta^2 \right]$ is the unit information under the null
model. The resulting 
  $g$-priors have precision matrices that are multiples of
  $\mathbf{X}_\MM^T \mathbf{X}_\MM$ as in the Gaussian case.


  Similar in spirit to Zellner's derivation of the $g$-prior, priors
  based on imaginary data have been developed in the context of GLMs
  by \citet{Bedrick_etal_1996, Chen_Ibrahim_2003, Bove_Held_2011,
    Perrakis_etal_2015, Fouskakis_etal_2016} among others.  In
  general, these do not lead to normal prior distributions and
  typically require MCMC methods to sample both parameters and models
  for BVS and BMA.
  The $g$-prior introduced by \citet{Bove_Held_2011} and later
  modified by \citet{Held_etal_2015} adopts a large sample approximation  to justify a normal density:
\begin{equation}\label{eq:g-prior_SBH2011} 
\b_\MM \mid g, \MM   \sim 
\text{N}\left( \mathbf{0}, ~ g c ( \X_\MM^T (\I_n - \mathcal{P}_{\one_n}) \X_\MM)^{-1} \right)
\end{equation}
where imaginary samples are generated from the null model $\Mnull$ 
and the constant $c$ is 
inverse of the unit information given above evaluated at
the MLE of $\alpha$ under $\Mnull$. 
For the normal linear regression,  $c =
\sigma^2$ recovers the usual $g$-prior.

Under large sample approximations to the likelihood, the $g$-prior in
\eqref{eq:g-prior_SBH2011} permits conjugate updating, however, unlike
the Gaussian case, evaluating the resulting Bayes factors that contain
ratios of information matrix determinants among others can increase
computational complexity, and thus negates some of the advantages that
made the $g$-prior so popular in linear models.  Classic asymptotic
theory suggests that $\mathcal{I}_n(\b_\MM)$ measures the large sample
precision of $\b_\MM$, while $\mathcal{J}_n(\b_\MM)$ is recommended as
a more accurate measurement of the same quantity
\citep{Efron_Hinkley_1978}.  When the true model $\MT \neq \Mnull$,
evaluating information matrices at the MLE $\hat{\b}_\MM$
\citep{Hansen_Yu_2003, Wang_George_2007} may better capture the large
sample covariance structures of $\b_\MM$ and the local geometry under
model $\MM$.  On the other hand, using large sample approximations to
imaginary data generated from $\MM$ leads to a prior distribution for
$\b_\MM$ that is not centered at zero, and therefore will not satisfy
the predictive matching criterion of \citet{Bayarri_etal_2012}.

Next, we propose a $g$-prior that incorporates the local geometry at
the MLE with the objective of providing a prior that satisfies the
model selection desiderata, provides analytic expressions that permit
deeper understanding of their theoretical properties, and leads to
computationally efficient algorithms under large sample approximations
to likelihoods.

\subsection{Local Information Metric $g$-Prior }\label{section:g-prior}

The invariance and predictive matching criteria in
\citet{Bayarri_etal_2012} lead to adoption of \eqref{eq:alpha-prior}-\eqref{eq:sigma-prior}
for location-scale families.  Although the Poisson and Bernoulli families are
not location-scale families, it is desirable that the prior/posterior
distribution for $\eg$ is invariant under any location changes in
the design matrix $\X_\MM$.  In the following proposition, 
we will use the uniform prior in \eqref{eq:alpha-prior} 
and a second order Taylor expansion of the likelihood
as a starting point, for deriving the (approximate)
integrated likelihood for $\b_\MM$  and subsequent prior distribution
for $\b_\MM$.

\begin{proposition}\label{proposition:marlik_beta}
For any model $\MM$, with a uniform prior $p(\alpha) \propto 1$, 
the marginal likelihood of $\b_\MM$ under model $\MM$ is proportional to
\begin{align}
&	p(\Y \mid \b_\MM, \MM) =
	\int p(\Y \mid \alpha, \b_\MM, \MM)p(\alpha) d\alpha \nonumber \\
\propto
~ &
	p\left(\Y \mid \hat{\alpha}_\MM, \hat{\b}_\MM, \MM\right)
	\mathcal{J}_n(\hat{\alpha}_\MM)^{-\frac{1}{2}}
	\exp\left\{ -\frac{1}{2}\left(\b_\MM - \hat{\b}_\MM\right)^T  
	\mathcal{J}_n(\hat{\b}_\MM)
	 \left(\b_\MM - \hat{\b}_\MM\right) 
	\right\}, \label{eq:marlik_beta}
\end{align}
where the approximation \eqref{eq:marlik_beta} is precise $O(n^{-1})$,
and the observed information of $\boldsymbol\eta_\MM$, $\alpha$, and $\b_\MM$ at the MLEs 
$\hat{\eta}_{\MM,i} = \hat{\alpha}_\MM + \mathbf{x}_{\MM, i}^T\hat{\boldsymbol\beta}_\MM$ are
\begin{align}
{\mathcal{J}}_n(\hat{\boldsymbol\eta}_\MM) 
& = 	\emph{\diag} (d_i) \text{ where } d_i = -Y_i \, \theta^{\prime\prime}(\hat{\eta}_{\MM, i}) 
	+ (b \circ \theta)^{\prime\prime}(\hat{\eta}_{\MM, i}) \text{
  for }
	i = 1, \ldots, n,  
	 \label{eq:obs_info}\\
\mathcal{J}_n(\hat{\alpha}_\MM)
&	= \mathbf{1}_n^T \mathcal{J}_n(\hat{\boldsymbol\eta}_\MM) \mathbf{1}_n, \label{eq:J_alpha}\\
\mathcal{J}_n(\hat{\b}_\MM)
&	= \mathbf{X}_\MM^T (\mathbf{I}_n -  \P_{\mathbf{1}_n})^T
	\mathcal{J}_n(\hat{\boldsymbol\eta}_\MM) (\mathbf{I}_n -  \P_{\mathbf{1}_n})\mathbf{X}_\MM 
	\label{eq:J_beta},
\end{align}
respectively, and 
\begin{eqnarray} \label{eq:projection}
\P_{\mathbf{1}_n} = \mathbf{1}_n
\left( \mathbf{1}_n^T \mathcal{J}_n(\hat{\boldsymbol\eta}_{\MM})
\mathbf{1}_n\right)^{-1} \mathbf{1}_n^T
\mathcal{J}_n(\hat{\boldsymbol\eta}_{\MM})
\end{eqnarray}
is the orthogonal projection onto the span $\mathbf{1}_n$ under the information
$\mathcal{J}_n(\hat{\boldsymbol\eta}_\MM)$ inner product,
$\uv^T\mathcal{J}_n(\hat{\boldsymbol\eta}_\MM)\vv$ for $\uv, \vv \in \bbR^n$. 
\end{proposition}
The proof of Proposition \ref{proposition:marlik_beta} is given in the
supplementary material Appendix \ref{PROOFproposition:marlik_beta}.

The approximate marginal likelihood in \eqref{eq:marlik_beta} is
proportional to a normal kernel of $\b_\MM$ with 
a precision (inverse covariance matrix) that is equal to the marginal
observed information $\mathcal{J}_n(\hat{\b}_\MM)$ and is a function of
the ``centered'' predictors, 
\begin{equation}\label{eq:Xc}
\mathbf{X}_{\MM}^c \stackrel{\triangle}{=} (\I_n - \P_{\one_n})\mathbf{X}_\MM,
\end{equation}
where the column means for centering are weighted averages
$\bar{\x}_{\J, j} = \sum_i d_i x_{ij}/\sum_i d_i$, with the weights 
proportional to $d_i$ in \eqref{eq:obs_info}. 
For non-Gaussian GLMs, $d_i$'s are not equal, and hence 
this centering step is different from the conventional 
procedure that uses the column-wise arithmetic average.

This leads to the following  proposal for a $g$-prior
under all models $\MM$ 
\begin{equation}\label{eq:g-prior_GLM1} 
\boldsymbol\beta_\MM \mid \MM, g   \sim 
\text{N}\left( \mathbf{0}, ~ g\cdot \mathcal{J}_n( \hat{\boldsymbol\beta}_\MM)^{-1} \right).
\end{equation}
The advantage of \eqref{eq:g-prior_GLM1} is two-fold: geometric
interpretability through local orthogonality, which will be
illustrated next, and computational efficiency in Bayes factor
approximation (see Section \ref{sec:marlik}). Note that we may
reparameterize the model \eqref{eq:eta}
\begin{equation} \label{eq:eta-centered}
  \eg = \one_n \alpha + \X_{\MM}^c \b_\MM
\end{equation}
where (with apologies for abuse of notation) $\alpha$ is the intercept
in the centered parameterization.  Under this centered parameterization
and with  $p(\alpha) \propto 1$, the observed information
at the MLEs is block diagonal, and leads to the same marginal
likelihood as in \eqref{eq:marlik_beta}. 

In hypothesis testing, where parameter $\b$ is tested against a null
value $\b_0$ with a nuisance parameter $\alpha$, \citet{Jeffreys_1961}
argues that when the Fisher information is block diagonal for all
values of $\b$ and $\alpha$, improper uniform priors on $\alpha$ can
be justified.  This global orthogonality, however, rarely holds
outside of normal models \citep{Cox_Reid_1987}.  Under a local
alternative hypothesis where the true value of $\b$ is in an
$O(n^{-1/2})$ neighborhood of $\b_0$, \citet{Kass_Vaidyanathan_1992}
show that Bayes factors are not sensitive to prior choices on the
nuisance parameter under a weaker condition of null orthogonality,
where $\mathcal{I}_n(\alpha, \b_0)$ is block diagonal for all $\alpha$
under the null hypothesis.  In particular, under null orthogonality,
the logarithm of the Bayes factor under the unit information prior for
$\b$ can be approximated by BIC with an error of $O_p(n^{-1/2})$
\citep{Kass_Wasserman_1995}.  For GLMs, the $g$-prior
\eqref{eq:g-prior_SBH2011} implies null orthogonality under the
centered reparameterization from $\X_\MM$ to
$(\I_n - \P_{\one_n}) \X_\MM$.

For variable selection, if the true value $\b_\MT^*$ does not lie in an
$O(n^{-1/2})$ neighborhood of the
null value, \citet{Kass_Vaidyanathan_1992} point out that the Bayes
factor will likely be decisive and for practical purposes the accuracy
of BIC does not matter. For model averaging, however, we may wish to
have more precise estimates of Bayes Factors in the posterior
probabilities. For estimation, local orthogonality 
at the MLE, as in the $g$-prior in \eqref{eq:g-prior_GLM1},  
captures the large sample geometry of the likelihood
parameters $(\alpha, \b_\MM)$  better than null orthogonality, and as we will
see, greatly simplifies posterior derivations and theoretical
calculations, and reduces computational complexity.

\citet{Bayarri_etal_2012} note that orthogonalization is
not required for adopting a uniform prior on $\alpha$, but instead the
criteria of predictive matching and location invariance are used to
justify the choice. Integration with respect to an improper uniform
measure on $\alpha$ leads to a marginal likelihood involving a
``centered'' $\X$ that is locally orthogonal to the column of ones
under the information inner product and invariant under any location
changes for the columns of $\X$. The uniform prior on the
intercept in either parameterization with the $g$ prior
\eqref{eq:g-prior_GLM1} leads to equivalent posterior
distributions on $\eg$.   For ease of exposition, however, we
will adopt the centered parameterization in \eqref{eq:eta-centered}
for the remainder of the article, and drop the superscript $c$ for
simplification of notation when there is no ambiguity.

\subsection{Posterior Distributions of Parameters}\label{section:shrinkage}

Under the $g$-prior \eqref{eq:g-prior_GLM1} on $\b_\MM$
and a uniform prior \eqref{eq:alpha-prior} on $\alpha$  for the
centered parameterization \eqref{eq:eta-centered},
asymptotic limiting distribution theory \citep[pp.\
287]{Bernardo_Smith_2000} under a Laplace approximation
yields the approximate  posterior distributions conditional on $\MM$ as
\begin{align}
\label{eq:betapost} \boldsymbol\beta_{\MM} \mid \mathbf{Y}, \MM, g &  \convd  
  \text{N}\left( \frac{g}{1+g}\ \hat{\boldsymbol\beta}_{\MM},\ \frac{g}{1+g} 
  \mathcal{J}_n(\hat{\boldsymbol\beta}_{\MM})^{-1} \right),\\ 
\label{eq:alphapost}\alpha \mid \mathbf{Y}, \MM & \convd
  \text{N}\left( 
    \hat{\alpha}_\MM,\    
    \mathcal{J}_n(\hat{\alpha}_\MM)^{-1}\right),
\end{align}
where the symbol $\convd$ indicates convergence in distribution, and
$\hat{\alpha}_\MM$ and $\hat{\boldsymbol\beta}_{\MM}$ are MLEs of
$\alpha$ and $\beta_\MM$ respectively under model $\MM$.  Due to local
orthogonality, the posterior distributions of $\b_\MM$ and $\alpha$
are asymptotically independent. Furthermore, for large $n$, the
asymptotic marginal posterior distribution of $\alpha$ is proper,
although its prior distribution is improper.  Similar results are
obtained by \citet{Held_etal_2015} under the assumption that
$\mathcal{I}_n (\hat{\alpha}_\MM, \hat{\b}_\MM)$ equals the block
diagonal matrix
$\mathcal{I}_n(\alpha, \boldsymbol\beta_\MM=\mathbf{0})$, which
approximates the expected information when $\b_\MM$ is in a
neighborhood of zero.

The conditional posterior mean of $\boldsymbol\beta_\MM$ is shrunk
from the MLE $\hat{\boldsymbol\beta}_\MM$ towards the prior mean
$\mathbf{0}$ by the ratio $g / (1+g) $, which is usually referred to
as the shrinkage factor for $g$-priors in normal linear regression
\citep{Liang_etal_2008}.  As
discussed in \citet{Copas_1983, Copas_1997},  shrinking predicted values toward the center of
responses, or equivalently, shrinking regression coefficients towards
zero, may alleviate over-fitting, and thus yield optimal predictive
performance.  In Section \ref{subsection:gusto-i}, using the GUSTO-I
data  and logistic regression, we find that methods that favor
smaller values of $g$, i.e., smaller shrinkage factors, tend to be
more accurate in out-of-sample prediction.

\subsection{Approximate Bayes Factor}\label{sec:marlik}


In GLMs, normal priors such as \eqref{eq:g-prior_SBH2011} and \eqref{eq:g-prior_GLM1}
yield closed form marginal likelihoods under Laplace approximations
which are precise to $O(n^{-1})$.
Under an integrated Laplace approximation \citep{Wang_George_2007} with the uniform prior on
$\alpha$ and $g$-prior in \eqref{eq:g-prior_GLM1}
for any model $\MM$, the approximate marginal likelihood
for $\MM$ and $g$ in \eqref{eq:marlik_def}  has a closed form
expression 
\begin{align}
p(\mathbf{Y} \mid \MM, g) 
= & ~  \int  p(\mathbf{Y} \mid \boldsymbol\beta_\MM, \MM)
             p(\boldsymbol\beta_\MM \mid \MM, g) 
             ~d\boldsymbol\beta_\MM \nonumber \\       
\propto
& ~
p (\mathbf{Y} \mid \hat{\alpha}_\MM, \hat{\boldsymbol\beta}_\MM, \MM)
  \mathcal{J}_n(\hat{\alpha}_\MM)^{-\frac{1}{2}}
  (1+g)^{-\frac{p_{\MM}}{2}} \exp \left\{ -\frac{Q_{\MM}}{2(1+g)} \right\}, \label{eq:marlik_fixedg}
\end{align}
where the approximation \eqref{eq:marlik_fixedg} is precise to $O(n^{-1})$, 
$p_\MM$ is the column rank of $\mathbf{X}_\MM$, and 
\begin{equation}\label{eq:Q}
Q_{\MM} = \hat{\boldsymbol\beta}_{\MM} ^T 
\mathcal{J}_n( \hat{\b}_{\MM}) \hat{\boldsymbol\beta}_{\MM}
\end{equation}
is the Wald statistic (under observed information). 
For the null model $\Mnull$ where $p_{\Mnull} = 0$,
$Q_{\Mnull} = 0$ so that
\eqref{eq:marlik_fixedg} still holds.
The approximate marginal likelihood \eqref{eq:marlik_fixedg} is a function of MLEs, 
which is fast to compute using existing algorithms such as the iterative weighted least squares
\citep{McCullagh_Nelder_1989}.   

To compare a pair of models $\MM_1$ and $\MM_2$, the Bayes factor
\citep{Kass_Raftery_1995}, defined as
$\BF_{\MM_1: \MM_2} = p(\mathbf{Y} \mid \MM_1, g) / p(\mathbf{Y} \mid
\MM_2, g)$,
is commonly used in Bayesian model selection, assuming the two models
are equally likely {\it a priori}.  If $\BF_{\MM_1: \MM_2}$ is greater
(less) than one, then $\MM_1$ ($\MM_2$) is favored.  When $2^p$ models
are considered simultaneously, under the uniform prior
$\pi(\MM) = 2^{-p}$, comparing their posterior probabilities is
equivalent to comparing their Bayes factors where each model is
compared to a common baseline model,  such as the
null model \citep{Liang_etal_2008}.  With the availability of closed
form approximate marginal likelihoods \eqref{eq:marlik_fixedg}, the
$g$-prior \eqref{eq:g-prior_GLM1} yields closed form Bayes factors
\begin{equation}\label{eq:DBF_g}
\BF_{\MM: \Mnull} = \frac{p(\mathbf{Y} \mid \MM, g)}{p(\mathbf{Y} \mid \Mnull) }
	= \exp \left\{\frac{z_\MM}{2} \right\} \left[ \frac{\mathcal{J}_n(
            \hat{\alpha}_{\Mnull})}{\mathcal{J}_n(
            \hat{\alpha}_{\MM})} \right]^{\frac{1}{2}}        
	(1 + g)^{-\frac{p_\MM}{2}} \exp\left\{ - \frac{Q_\MM}{2(1+g)} \right\},
\end{equation} 
where 
\begin{equation}\label{eq:z}
z_\MM  = 2\log \left\{ \frac{p(\mathbf{Y} \mid \hat{\alpha}_\MM, \hat{\boldsymbol\beta}_\MM, \MM)}
{p(\mathbf{Y} \mid \hat{\alpha}_{\Mnull}, \Mnull)} \right\}
\end{equation}
is the change in deviance or two times the likelihood ratio test
statistic for comparing model $\MM$ to $\Mnull$. For simplicity,
$z_\MM$ will be referred as the deviance statistic for the rest of
this article.  The Bayes factors under the $g$-prior provides an
adjustment to the likelihood ratio test with a penalty that depends
on $g$ and the Wald statistic.

The expression for the Bayes factor in \eqref{eq:DBF_g} is closely related to the
test-based Bayes factors (TBF) of  \citet{Hu_Johnson_2009, Held_etal_2015,
  Held_etal_2016}
\begin{equation}\label{eq:TBF_g}
\textsf{TBF}_{\MM:{\Mnull}} = \frac{ \text{G}\left(z_\MM; \frac{p_\MM}{2}, \frac{1}{2(1+g)} \right)}
  {\text{G}\left(z_\MM; \frac{p_\MM}{2}, \frac{1}{2} \right)}
  = (1+g)^{-\frac{p_\MM}{2}} \exp\left\{\frac{g \, z_\MM}{2(1+g)}\right\},
  \end{equation}
  which is derived from the asymptotic distributions of $z_\MM$ under
  $\MM$ and $\Mnull$; $\text{G}(z_\MM; a, b)$ denotes the density of a
  Gamma distribution with mean $a/b$, evaluated at $z_\MM$.  Under the
  null or a local alternative where $\b_\MM$ is in an 
  $O(n^{-1/2})$ neighborhood of the null, the Wald statistic $Q_\MM$ and deviance
  statistic $z_\MM$ are asymptotically equivalent and the ratio
  $\mathcal{J}_n(\hat{\alpha}_{\Mnull})/\mathcal{J}_n(\hat{\alpha}_{\MM})$
  in \eqref{eq:DBF_g} converges to one in probability, resulting in
  the data-based Bayes factor in \eqref{eq:DBF_g} (or DBF for short)
  with $Q_\MM$ replaced by $z_\MM$ being
  equivalent asymptotically to the TBF.  When the distance between
  $\b_\MM$ and the null does not vanish with $n$, we find that the TBF
  exhibits a small but systematic bias, but leads to little difference
  in inference for large $g = n$, where both are close to BIC. In
  Section \ref{sec:examples}, using simulation and real examples, we
  find that with $g = n$, TBF and the DBF \eqref{eq:DBF_g} have almost
  identical performance in model selection, estimation, and
  prediction. More discussions and an empirical example with TBF are
  available in the supplementary material Appendix \ref{section:TBF}.

\subsection{When MLEs Do Not Exist}\label{sec:no_MLE}
Before turning to the choice of $g$ and other properties,
we investigate the possible use of $g$-priors
\eqref{eq:g-prior_GLM1} when MLEs of $\alpha_\MM$ or $\b_\MM$ do not
exist. Two different cases are considered: data separation in binary
regression,  and non-full rank design matrices for GLMs with known dispersion.
We will return to the case of $g$-priors in linear models with unknown
dispersion in the non-full rank case in Section \ref{section:CHIC_normal}.

For binary regression models with a finite sample size, data
separation problems may lead to MLEs that are not unique nor finite
\citep{Albert_Anderson_1984, Heinze_Schemper_2002, Ghosh_Li_Mitra_2017}.
For $\X_\MM$ of full rank, the data exhibit separation if there exists
a scalar $\gamma_0 \in \mathbb{R}$ and a non-null vector
$\g = ( \gamma_1, \ldots, \gamma_{p_\MM})^T \in \mathbb{R}^{p_\MM}$
such that
\begin{equation}\label{eq:separation}
\gamma_0 +  \mathbf{x}_{\MM, i}^T \g \geq 0 \ \text{ if } Y_i = 1, \quad 
\gamma_0 +  \mathbf{x}_{\MM, i}^T \g  \leq 0 \ \text{ if } Y_i = 0, \quad
\text{for all } i = 1, \ldots, n.
\end{equation} 
In particular, there is complete separation if in
\eqref{eq:separation} strict inequalities hold 
for all observations. In the absence of complete separation, 
there is quasi-complete separation if \eqref{eq:separation} holds with equality
for at least one sample.
 
This implies that
the information metric is no longer a valid inner product and that the
operator in \eqref{eq:projection} is not an orthogonal projection.
While it is possible to define projections in the case  where
$\mathcal{J}_n(\hat{\boldsymbol{\eta}}_\MM)$ is not full rank 
\citep[Chapter~10]{Christensen:2011},
we will restrict attention to the case where $\mathcal{J}_n(\hat{\boldsymbol{\eta}}_\MM)$ is
of full rank and conditions for asymptotic normality hold  to avoid additional technicalities.

Design matrices that are not full rank lead to identifiability
problems with MLEs of $\alpha_\MM$ and $\b_\MM$ in GLMs.
Consider a model $\MM$ where $\textsf{rank}(\X_{\MM}) = \rank_\MM < p_\MM$, and a
full rank design matrix $\X_{\MM'}$ that contains $\rank_\MM$ columns and
spans the same column spaces as $\X_\MM$, i.e.,
$C(\X_\MM) = C(\X_{\MM'})$.  Although the MLE of the coefficients
$\hat{\b}_\MM$ are not all unique, MLEs of the linear predictors
$\hat{\eta}_{\MM, i}$ are unique; in
fact,
\begin{equation}\label{eq:Jeta_non-full_rank}
\hat{\boldsymbol\eta}_\MM = \one_n \hat{\alpha}_{\MM} + \X_{\MM} \hat{\b}_{\MM}
= \one_n \hat{\alpha}_{\MM'} + \X_{\MM'} \hat{\b}_{\MM'}
\end{equation}
and $\mathcal{J}_n (\hat{\boldsymbol\eta}_\MM)$ is unique and positive definite.
The precision matrix of the $g$-prior \eqref{eq:g-prior_GLM1}, 
$\mathcal{J}_n(\hat{\b}_\MM) =
\mathbf{X}_\MM^{cT}\mathcal{J}_n(\hat{\boldsymbol\eta}_\MM)\mathbf{X}_\MM^c$
is well-defined, however, since ${\textsf{rank}}(\X_{\MM}^{c}) = \textsf{rank}(\X_{\MM})
= \rank_\MM < p_\MM$,  it is not invertible.
Note that the null-based $g$-prior \eqref{eq:g-prior_SBH2011} suffers from 
a similar singularity problem. 

We may extend the definition of $g$ priors to include singular
covariance matrices by adopting generalized inverses in defining the $g$-prior.
Because of the invariance of orthogonal projections to choices of
generalized inverse and uniqueness of the MLE of $\eg$,
we have the following proposition regarding the Bayes factors in
models that are rank deficient.
\begin{proposition}\label{proposition:non_full_rank_BF}
Suppose $\emph{\textsf{rank}}(\X_\MM) = \rank_\MM < p_\MM$, then 
\begin{equation}
  \label{eq:BF-singular}
\BF_{\MM: \Mnull} = \frac{p(\mathbf{Y} \mid \MM, g)}{p(\mathbf{Y} \mid \Mnull) }
	=  \exp \left\{ \frac{z_\MM}{2} \right\}  \left[ \frac{\mathcal{J}_n(
            \hat{\alpha}_{\Mnull})}{\mathcal{J}_n(
            \hat{\alpha}_{\MM})} \right]^{\frac{1}{2}}       
	(1 + g)^{-\frac{\rank_\MM}{2}} \exp\left\{ - \frac{Q_\MM}{2(1+g)} \right\}.
\end{equation}
If $\MM'$ is a full rank model whose column space $C(\X_{\MM'}) = C(\X_\MM)$, then
$Q_\MM = Q_{\MM'}$, $z_\MM = z_{\MM'}$, 
and 
$\BF_{\MM: \MM'} = 1$.
\end{proposition}

The proof is available in supplementary material Appendix
\ref{PROOFproposition:non_full_rank_BF}. Here the two models $\MM$ and
$\MM'$ have the same Bayes factor if their design matrices span the
same column space.  This form of invariance is not possible with other
conventional independent prior distributions, such as generalized
ridge regression or independent scale mixtures of normals. While
posterior means of coefficients under BMA will not be well defined,
predictive quantities under model selection or model averaging will 
exist, however, care must be taken in assigning prior probabilities
over equivalent models.

\subsection{Choice of $g$}\label{subsection:EB}

Problems with fixed values of $g$ prompted \citet{Liang_etal_2008} to
study data-dependent or adaptive values for $g$.  This includes the
unit information prior where $g = n$ \citep{Kass_Wasserman_1995},
 and local and global empirical Bayes (EB) estimates of
$g$ \citep{Copas_1983, Copas_1997, Hansen_Yu_2001, Hansen_Yu_2003,
Liang_etal_2008, Held_etal_2015}.

For the local EB, each model $\MM$ has its own optimal value of $g$
that maximizes its marginal likelihood:
\[
\hat{g}_\MM^{\text{LEB}} = \arg\max_{g\geq0}~ p(\mathbf{Y} \mid \MM, g),
\]
and the local EB estimator of the marginal likelihood is obtained by
simply plugging in the estimator:
$p^{\text{LEB}}(\mathbf{Y} \mid \MM) = p(\mathbf{Y} \mid \MM, \hat{g}_\MM^{\text{LEB}})$.

For example, under the $g$-prior \eqref{eq:g-prior_GLM1}, \citet{Hansen_Yu_2003}
derive
\[
\hat{g}^{\text{LEB}}_\MM = \max \left(\frac{Q_\MM}{p_\MM} -1, 0 \right),
\]
which has a similar format to $\hat{g}^{\text{LEB}}_\MM = \max(z_\MM/p_\MM - 1, 0)$,
its counterpart for the test-based marginal likelihood
under the $g$-prior \eqref{eq:g-prior_SBH2011}, derived by \citet{Held_etal_2015}.

The global EB involves only a single estimator of $g$, based
on the marginal likelihood averaged over all  models
$\hat{g}_\MM^{\text{GEB}} = \arg\max_{g\geq0} \sum_\MM p(\MM)p(\mathbf{Y} \mid \MM, g)$.
The global EB estimator may be obtained via an EM algorithm when all
models may be enumerated \citep{Liang_etal_2008}, but is more
difficult to compute for larger problems \citep{Held_etal_2015}.  
For the remainder of the article,  we will restrict attention to the
local EB approach. 

The EB estimates of $g$ do not lead to consistent model
selection  under the null 
model \citep{Liang_etal_2008} although provide consistent estimation. 
Mixtures of $g$-priors provide an alternative 
that propagate uncertainty in $g$ with  other desirable properties.

\section{Mixtures of $g$-Priors} \label{section:mix_g}

\citet{Liang_etal_2008} highlight some of the problems with using a
fixed value of $g$ for model selection or BMA and recommend mixtures
of $g$-priors that lead to closed form expressions or tractable
approximations.  In order to consider the model selection criteria of
\citet{Bayarri_etal_2012}, we propose an extremely flexible
mixture of $g$-priors family that can encompass the majority of the
existing mixtures of $g$-priors as special cases.  Furthermore,
utilizing Laplace approximations to obtain \eqref{eq:marlik_beta}, it
yields  marginal likelihoods and (data-based) Bayes
factors in closed form, for both GLMs \eqref{eq:likelihood}, and
extensions such as normal linear regressions with unknown variances
and over-dispersed GLMs. This tractability permits establishing
properties such as consistency.

\subsection{Compound Confluent Hypergeometric Distributions}\label{section:CHIC}

The parameter $g$ enters into the posterior distribution for $\b_\MM$
and the marginal likelihood \eqref{eq:marlik_fixedg} through the
shrinkage factor $g/(1+g)$ or the complementary shrinkage factor
$u = 1/(1+g)$. Since the approximate marginal likelihood depends on $g$
in the format of $u$,
$p(\mathbf{Y} \mid \MM, u) \propto 
u^{p_\MM/2}\exp(-u Q_\MM
/2)$,
a conjugate prior for $u$ (given $\phi_0$) should contain the
kernel of a truncated Gamma density with the support
$u \in [0, 1]$.  Beta distributions are also natural prior choice for
$u$, such as the hyper-$g$ prior of \citet{Liang_etal_2008}.  Other
mixtures of $g$-priors such as the robust prior
\citep{Bayarri_etal_2012} and the intrinsic prior
\citep{Womack_etal_2014} truncate the support of $g$ away from zero,
so the resulting $u$ has an upper bound strictly smaller than one. 

To incorporate the above choices in one unified family, we adopt a
generalized Beta distribution introduced by \citet{Gordy_1998a} called
the Compound Confluent Hypergeometric distribution, whose density
function contains both Gamma and Beta kernels, and allows truncation on
the support through a straightforward  extension. We say that $u$ has a
truncated Compound Confluent Hypergeometric distribution if
$u \sim \text{tCCH}(t,q,r,s,v,\kappa)$ with density expressed as
\begin{equation}\label{eq:CCH_dist}
p(u \mid t, q, r, s, v, \kappa)
= \frac{v^t \exp(s/v)}{B(t, q)\ \Phi_1(q, r, t+q, s/v, 1-\kappa)}\
\frac{u^{t-1}(1-vu)^{q-1}e^{-su}}{\left[ \kappa + (1-\kappa)vu \right]^r}\ 
\mathbf{1}_{\{0 < u< \frac{1}{v}  \}} 
\end{equation}
where parameters $t>0, q>0, r \in \mathbb{R}, s \in \mathbb{R}, v \geq 1$, and $\kappa >
0$. Here, $B(t,q)$ is the Beta function and
$\Phi_1(\alpha, \beta, \gamma, x, y) = \sum_{m=0}^{\infty} \sum_{n=0}^{\infty}
(\alpha)_{m+n} (\beta)_n x^m y^n /  \left[(\gamma)_{m+n} m! n!\right]$
is the confluent hypergeometric function of two variables or Humbert
series \citep{Humbert_1920}, 
and $(\alpha)_n$ is the Pochammer coefficient or shifted factorial:
$(\alpha)_n = 1$ if $n = 0$ and
$(\alpha)_n = \Gamma(\alpha + n)/\Gamma(\alpha)$ for $n\in \bbN$.
Note that the parameter $v$ controls the support of $u$. 
When $v=1$, the support is $[0, 1]$. 
When $v > 1$, the upper bound of the support is strictly less than one, 
which may accommodate priors with truncated $g$.
This leads to conjugate updating of $u$ as follows:
\begin{proposition}\label{proposition:tCCH_marlik}
Let $u = 1/(1+g)$ have the prior distribution
\begin{equation}\label{eq:tCCH_mixture_of_g}
u \sim {\text{tCCH}}\left(\frac{a}{2}, \frac{b}{2}, r, \frac{s}{2}, v, \kappa\right)
\end{equation}
where $a, b, \kappa>0, r , s \in \mathbb{R}$, and $v \geq 1$, then for
GLMs with a fixed dispersion $\phi_0$, 
integrating the marginal likelihood in \eqref{eq:marlik_fixedg} with
respect to the prior on $u$
yields the  marginal likelihood for $\MM$ which is proportional to
\begin{align}p(\mathbf{Y} \mid \MM)  
 \propto & \ p \left(\mathbf{Y}| \hat{\alpha}_\MM, \hat{\boldsymbol\beta}_\MM, \MM \right)
   \mathcal{J}_n(\hat{\alpha}_\MM)^{-\frac{1}{2}}
   v^{-\frac{p_\MM}{2}} \exp\left\{-\frac{Q_\MM}{2v}\right\} \nonumber
   \\ \label{eq:marlik_tCCHg}
& \cdot \frac{B\left(\frac{a + p_\MM}{2}, \frac{b}{2} \right)
  \Phi_{1}\left( \frac{b}{2}, r, \frac{a + b + p_\MM}{2}, \frac{s+Q_\MM}{2v}, 1-\kappa \right)
  }
  {B\left(\frac{a}{2}, \frac{b}{2} \right) \Phi_{1}\left( \frac{b}{2}, r, \frac{a+b}{2}, \frac{s}{2v}, 1-\kappa \right)},
\end{align}
where $p_{\MM}$ is the rank of $\mathbf{X}_\MM$, and $Q_{\MM}$ is
given in \eqref{eq:Q}.
The posterior distribution of $u$ under model $\MM$ 
is also a tCCH distribution asymptotically
\begin{equation}\label{eq:tCCH_post}
u \mid \mathbf{Y}, \MM \convd  
\text{tCCH}\left(\frac{a + p_\MM}{2}, \frac{b}{2}, r, \frac{s + Q_\MM}{2}, v, \kappa \right)
\end{equation}
allowing conjugate updating under integrated  Laplace approximations.
\end{proposition} 

The proof is available in supplementary material Appendix
\ref{PROOFproposition:tCCH_marlik}. 

\begin{corollary}
The Bayes factor for comparing $\MM$ to $\Mnull$ is  
\[
\BF_{\MM: \Mnull}
= \left[ \frac{\mathcal{J}_n(\hat{\alpha}_{\Mnull})}{\mathcal{J}_n(\hat{\alpha}_\MM)} \right]^{\frac{1}{2}}
v^{-\frac{p_\MM}{2}} \exp\left\{\frac{z_\MM}{2}   - \frac{Q_\MM}{2v}\right\}
 \frac{B\left(\frac{a + p_\MM}{2}, \frac{b}{2} \right)
  \Phi_{1}\left( \frac{b}{2}, r, \frac{a + b + p_\MM}{2}, \frac{s+Q_\MM}{2v}, 1-\kappa \right)
  }
  {B\left(\frac{a}{2}, \frac{b}{2} \right) \Phi_{1}\left( \frac{b}{2}, r, \frac{a+b}{2}, \frac{s}{2v}, 1-\kappa \right)} 
\]
and depends on the data through the deviance $z_\MM$ and the Wald
statistic $Q_\MM$.
\end{corollary}

We refer to the model selection criterion based on the Bayes factor
above as the ``Confluent Hypergeometric Information Criterion'' or
CHIC, as it  involves the confluent hypergeometric function
in two variables and the  $g$-prior is derived using the information 
matrix; the hierarchical prior formed by
\eqref{eq:alpha-prior}, \eqref{eq:g-prior_GLM1} and
\eqref{eq:tCCH_mixture_of_g} will be denoted as the CHIC $g$-prior.


In the conjugate updating scheme \eqref{eq:tCCH_post}, 
the parameter $a$ and $s$ are updated by the model rank $p_\MM$
and the Wald statistic $Q_\MM$, respectively, while none of the remaining 
four parameters are updated by the data.   The
parameters $a/2$ and $b/2$ play a role similar to the shape parameters in Beta
distributions, where small $a$ or large $b$ tends to put more prior weight
on small values of $u$, or equivalently, large values of $g$. We will
show later that $a$ also controls the tail behavior of the marginal
prior on $\b_\MM$.  The parameter $v$ controls the
support, while parameters $r, s$, and $\kappa$ ``squeeze'' the prior
density to left or right \citep{Gordy_1998a}. In particular, large $s$
skews the prior distribution of $u$ towards the left side and in turn
favoring large $g$. Table \ref{tb:tCCHg_parameters} lists special
cases of the CHIC $g$-prior and corresponding hyper parameters that have
appeared in the literature.  The last column indicates  whether the
model selection consistency holds for all models which will be
presented in Section \ref{sec:selection_consistency}.  We provide more
details about these special cases in the next section.

\begin{table}[]
\centering
\caption{Special cases of the CHIC $g$-prior with hyper parameters and
  whether the prior distributions lead to consistency for model
  selection under all models. If no, the models where consistency
  fails are indicated.  }\label{tb:tCCHg_parameters}
\begin{tabular}{|l |cccccc | l | }
  \hline			
				& $a$		& $b$		& $r$		& $s$		& $v$		& $\kappa$	& Consistency		\\	
  \hline\hline
CH			& $a$		& $b$		& $0$
                                                                                & $s$		& $1$		& $1$		& If $b = O(n)$ or $s = O(n)$	\\	\hline
Hyper-$g$	& $1$	& $2$		& $0$		& $0$
                                                                                                &
                                                                                                  $1$		& $1$				& No, $
\Mnull$			\\	\hline
Uniform			& $2$		& $2$		& $0$
                                                                                & $0$		& $1$		& $1$		& No, $\Mnull$			\\	\hline
Jeffreys			& $0$		& $2$		& $0$		& $0$		& $1$		& $1$		& No, $\Mnull$			\\	\hline
Beta-prime		& $\frac{1}{2}$	& $n-p_\MM - 1.5$	& $0$	& $0$	& $1$	& $1$		& Yes			\\	\hline
Benchmark		& $0.02$	& $0.02 \max(n, p^2)$	& $0$	& $0$	& $1$	& $1$		& Yes			\\	\hline
TruncGamma	& $2a_t$	& $2$	& $0$		& $2s_t$	& $1$		& $1$	&	If $s_t = O(n)$ \\	\hline
ZS adapted	& $1$		& $2$	& $0$		& $n + 3$	& $1$		& $1$		& Yes			\\	\hline\hline
Robust  & $1$		& $2$		& $1.5$		& $0$		& $\frac{n + 1}{p_\MM + 1}$	& $1$		& Yes			\\	\hline
Hyper-$g/n$ 	& $1$	& $2$		& $1.5$& $0$	& $1$	& $\frac{1}{n}$	& Yes			\\	\hline
Intrinsic 			& $1$		& $1$		& $1$		& $0$		& $\frac{n + p_\MM + 1}{p_\MM + 1}$	& $\frac{n + p_\MM + 1}{n}$	& Yes			\\	\hline
\end{tabular}
\end{table}

\subsection{Special Cases}\label{section:CH-g}

\begin{description}[leftmargin=*, itemsep=0em]

\item[\emph{Confluent Hypergeometric (CH) prior}]
The Confluent Hypergeometric distribution, proposed by
\citet{Gordy_1998b} is a special case of the CHIC family and
is a generalized Beta distribution with density
\[
 p (u \mid t, q, s) = \frac{u^{t-1} (1-u)^{q-1} \exp(-su)}{B(t, q)\  _{1}F_{1}(t, t+q, -s)} 
\ \mathbf{1}_{\{0 < u< 1  \}}
\]
where  $t > 0, q > 0, s \in \bbR$, 
and 
 $_{1}F_{1}(a, b, s) =\frac{\Gamma(b)}{\Gamma(b-a)\Gamma(a)}\int_0^1 z^{a-1} (1-z)^{b-a-1}
 \exp(sz) dz$ is the Confluent Hypergeometric function \citep{Abramowitz_Stegun_1970}.
Based on this distribution, we propose the CH  prior by letting $u$ have 
the following hyper prior
\begin{equation}\label{eq:CHg-prior}
u \sim \text{CH}\left(\frac{a}{2}, \frac{b}{2}, \frac{s}{2}  \right),
\end{equation}
under which the posterior for $u$ is again in the same family, and $p(\mathbf{Y} \mid \MM)$ has a closed form
\begin{align}
\label{eq:upost} 
u \mid \mathbf{Y}, \MM & \convd  
  \text{CH}\left(\frac{a + p_\MM}{2},  \frac{b}{2}, \frac{s +
                         Q_\MM}{2} \right), \\
p(\mathbf{Y} \mid \MM)  
& \propto  \ p \left(\mathbf{Y}\mid \hat{\alpha}_\MM, \hat{\boldsymbol\beta}_\MM, \MM\right) 
  \mathcal{J}_n(\hat{\alpha}_\MM)^{-\frac{1}{2}} \cdot \frac{B\left( \frac{a+p_\MM}{2}, \frac{b}{2} \right)\  
  _{1}F_{1}\left( \frac{a+p_\MM}{2}, \frac{a+b+p_\MM}{2}, -\frac{s+Q_\MM}{2} \right)}
  {B\left( \frac{a}{2}, \frac{b}{2} \right)\  _{1}F_{1}\left(
           \frac{a}{2}, \frac{a+b}{2}, -\frac{s}{2} \right)}, \nonumber
\end{align}
under the integrated Laplace approximation.

Similar to the CHIC $g$-prior, small $a$, large $b$, or large $s$ favors
small $u$ {\it a priori}, with $a$ controlling the tail behavior.
In model selection, preference for heavy-tailed prior distributions
can be traced back to \citet{Jeffreys_1961}, 
who suggested a Cauchy prior for the normal location parameter
to resolve the information paradox in the simple normal means case.
The following result 
shows that the CH prior  has multivariate Student $t$ tails with
degrees of freedom $a$, and in particular, the choice $a = 1$ leads to
tail behavior like a multivariate Cauchy.

\begin{proposition} \label{proposition:CH-g_tails}
Under the CH prior, the marginal prior distribution $p(\boldsymbol\beta_\MM \mid \MM)$
has tails behaving as multivariate Student distribution with degrees of freedom $a$, i.e., 
\[
\lim_{\| \boldsymbol\beta_\MM \| \rightarrow \infty} p(\boldsymbol\beta_\MM \mid \MM) \propto
\left( \| \boldsymbol\beta_\MM \|_{\mathcal{J}_n}^2 \right) ^{-\frac{a + p_\MM}{2}}
\]
where $\| \boldsymbol\beta_\MM \| =  (\boldsymbol\beta_\MM^T \boldsymbol\beta_\MM)^{\frac{1}{2}}$
and $\| \boldsymbol\beta_\MM \|_{\mathcal{J}_n} 
=  \left[ \boldsymbol\beta_\MM^T \mathcal{J}_n(\hat{\boldsymbol\beta}_\MM) \boldsymbol\beta_\MM \right]
^{\frac{1}{2}}$.
\end{proposition}

A proof is available in supplementary materials Appendix \ref{PROOFproposition:CH-g_tails}. 
While the CH prior has only half of the
number of parameters as the CHIC $g$-prior, it remains a flexible
class of priors for $u \in [0, 1]$.  In particular, when $s=0$,
\eqref{eq:CHg-prior} reduces to a Beta distribution, and when $b=2$,
it reduces to a truncated Gamma distribution.  
For the CH prior, we let parameter $a$ be fixed, and parameters $b$ and $s$ be either fixed, or on the order of $O(n)$.
The CH  prior, and
thus the CHIC $g$-prior, encompass several existing mixtures of
$g$-priors as follows:

\item[\emph{Truncated Gamma prior}]
\citep{Wang_George_2007, Held_etal_2015}
\begin{equation}\label{eq:TG}
u\sim \text{TG}_{(0, 1)}\left(a_t, s_t\right) \Longleftrightarrow
p(u) = \frac{s_t^{a_t}}{\gamma(a_t, s_t)} u^{a_t -1} e^{-s_t u}\ \mathbf{1}_{\left\{0 < u < 1 \right\}}
\end{equation}
with parameters $a_t , s_t > 0$ and support $[0,1]$.
Here $\gamma(a, s) = \int_0^{s} t^{a-1} e^{-t}dt$ 
is the incomplete Gamma function.
This is equivalent to assigning an incomplete inverse-Gamma prior to $g$.
The truncated Gamma prior permits conjugate updating in GLMs:
$u\mid \mathbf{Y}, \MM \sim \text{TG}_{(0, 1)}\left(a_t + p_\MM/2, s_t + Q_\MM/2\right)$.
When $a_t = 1, s_t = 0$, \eqref{eq:TG} reduces to a uniform prior on $u$.
\citet{Held_etal_2015} introduce the \emph{ZS adapted prior} by letting 
$a_t = 1/2, s_t = (n+3)/2$, so that the resulting prior on $g$
matches the prior mode of \citet{Zellner_Siow_1980} prior $g \sim
\text{IG}(1/2, n/2)$.

\item[\emph{Hyper-$g$ prior}] \citep{Liang_etal_2008, Cui_George_2008}
\begin{equation}\label{eq:Hyper-g}
u \sim \text{Beta}\left(\frac{a_h}{2}-1, 1 \right), \text{ where } 2 < a_h \leq 4
\end{equation}
with default value $a_h = 3$.
When $a_h = 4$, \eqref{eq:Hyper-g} reduces to \emph{a uniform prior} on $u$.
The choice $a_h = 2$ corresponds to the \emph{Jeffrey's prior} on $g$,
which is an improper prior and will lead to indeterminate Bayes
factors if the null model is included in the space of
models. \citet{Celeux_etal_2012} avoid this by excluding the null
model from consideration. 
The hyper-$g$ prior \eqref{eq:Hyper-g} can also be expressed as a Gamma distribution
truncated to the interval $[0, 1]$, and hence has conjugate updating in GLMs, 
\begin{equation}
\label{eq:u_post_hyper-g}
u  \sim \text{TG}_{\left(0,1\right)} 
\left(\frac{a_h}{2} -1, 0\right)
\Longrightarrow u \mid \mathbf{Y}, \MM \convd \text{TG}_{\left(0,1\right)} 
\left(\frac{p_\MM + a_h}{2} -1, \frac{Q_\MM}{2}\right).
\end{equation}

\item[\emph{Beta-prime prior}] \citep{Maruyama_George_2011} 
\[
u \sim \text{Beta}\left(\frac{1}{4}, \frac{n - p_\MM - 1.5}{2}\right),
\]
which is equivalent to a Beta-prime prior on $g$.
The second parameter was carefully chosen for normal linear models
to avoid evaluation of the Hypergeometric $_2F_1$ function
\citep[eq 15.3.1]{Abramowitz_Stegun_1970} in marginal likelihoods. 

\item[\emph{Benchmark prior}] \citep{Ley_Steel_2012}
\[
u \sim \text{Beta}\left(c, c \cdot \max(n, p^2)\right),
\]
which induces an approximate prior mean $\mathbb{E}(g) \approx \max(n, p^2)$ 
\citep{Fernandez_etal_2001}.
The recommended parameter value  is $c = 0.01$.


\item[\emph{Robust prior}] \citep{Bayarri_etal_2012} 
is a mixture of $g$-priors with the following  hyper prior
\begin{equation}\label{eq:robust_prior}
p_r(u) = a_r \left[ \rho_r (b_r+n)\right]^{a_r}\ \frac{u^{a_r-1}}{\left[ 1 + (b_r-1)u \right]^{a_r+1}}\ 
\mathbf{1}_{\left\{0 < u < \frac{1}{\rho_r(b_r + n) + (1 -b_r)} \right\}}
\end{equation}
where $a_r > 0, b_r > 0$ and $\rho_r \geq b_r / (b_r + n)$. The robust prior is a
special case in the CHIC family. 
The upper bound of its support $1/[\rho_r(b_r + n) + (1 -b_r)] \leq 1$.
Hence, the robust prior  does not include the CH prior \eqref{eq:CHg-prior}
as a special case, and vice versa. 

In normal linear models, the robust prior yields closed form marginal
likelihoods involving the Appell $F_1$ function
\citep{appell_1925,Weisstein_2009}.
Similarly in GLMs, evaluation of the special function $\Phi_1$ is
required. Based on the various criteria for model selection priors, 
default parameters $a_r = 0.5, b_r = 1$, and $\rho_r = 1/(1 + p_\MM)$
are recommended \citep{Bayarri_etal_2012},
under which the prior \eqref{eq:robust_prior} reduces to a truncated Gamma,
which leads to
\begin{equation}\label{eq:u_post_robust}
u \sim \text{TG}_{\left(0, \frac{p_\MM + 1}{n + 1}\right)}\left(\frac{1}{2},~ 0\right)
\Longrightarrow
u \mid \mathbf{Y}, \MM \convd \text{TG}_{\left(0, \frac{p_\MM + 1}{n + 1}\right)} \left(\frac{p_\MM+1}{2}, \frac{Q_\MM}{2}\right),
\end{equation}
and with marginal likelihood proportional to
\begin{align}\nonumber
p(\mathbf{Y} \mid \MM)  
\propto & \ p \left(\mathbf{Y}| \hat{\alpha}_\MM, \hat{\boldsymbol\beta}_\MM, \MM \right) 
   \mathcal{J}_n(\hat{\alpha}_\MM)^{-\frac{1}{2}}
   \left(\frac{n+1}{ p_\MM+1}\right)^{\frac{1}{2}}\\ \label{marlik_robust}
& ~~\cdot \left( \frac{Q_\MM}{2} \right)^{-\frac{p_\MM+1}{2}} \cdot
   \gamma \left(\frac{p_\MM+1}{2} , \frac{Q_\MM(p_\MM+1)}{2(n+1)} \right). 
\end{align}
Comparing \eqref{eq:u_post_hyper-g} and \eqref{eq:u_post_robust} reveals
an interesting finding: the robust prior can be viewed as a truncated hyper-$g$ prior,
with an upper bound increasing with $p_\MM$ and decreasing with $n$.
In fact, the robust prior includes the hyper-$g$ prior \eqref{eq:Hyper-g},
and  hyper-$g/n$ prior as special cases.

\item[\emph{Hyper-$g/n$ prior}] \citep{Liang_etal_2008}
\[
p(g) = \frac{a_h -2}{2n} \left(\frac{1}{ 1+ g/n}\right)^{a_h /2 }, \text{ where } 2 < a_h \leq 4. 
\]

\item[\emph{Intrinsic prior}] \citep{Berger_Pericchi_1996,
Moreno_etal_1998, Womack_etal_2014}
is another mixture of $g$-priors that truncates
the support of $g$. It has the hyper prior
\[
g = \frac{n}{p_\MM + 1}\cdot \frac{1}{w}, \quad
w  \sim \text{Beta}\left( \frac{1}{2}, \frac{1}{2} \right).
\]
Under the intrinsic prior, the parameter $g$ is truncated to have an lower bound 
$n / (p_\MM + 1)$, which corresponds to an upper bound of $u$ 
to be $(p_\MM + 1) / (n + p_\MM + 1)$.
As shown in Table \ref{tb:tCCHg_parameters}, the intrinsic prior
is also in the CHIC family.

\end{description}

\subsection{Unknown Dispersion}\label{section:CHIC_normal}

For the well studied case of  normal linear regressions with unknown variances, 
special cases of the CHIC $g$-prior, such as the hyper-$g$, hyper-$g/n$,
Beta-prime, benchmark, and robust priors yield closed form Bayes factors, 
although they may require evaluation of special functions such as the Gaussian
Hypergeometric $_2F_1$  or Appell $F_1$ \citep{Liang_etal_2008,Bayarri_etal_2012,SabanesBove_etal_2015}. 
For normal linear regression,
\citet{Liang_etal_2008} show that under the
$g$-prior \eqref{eq:alpha-prior}-\eqref{eq:gprior}, 
the marginal likelihood conditional on $g$ (or $u$) is
\begin{equation}\label{eq:normal_marlik}
p(\mathbf{Y} \mid \MM, g) 
	= \frac{p(\mathbf{Y} \mid \Mnull)~ (1 + g)^{\frac{n-p_\MM -1}{2}}}
	{\left[1 + g(1 - R_{\MM}^2)\right]^{\frac{n-1}{2}}}
\Longleftrightarrow
p(\mathbf{Y} \mid \MM, u) 
	= \frac{p(\mathbf{Y} \mid \Mnull)~u^{\frac{p_\MM}{2}}}
	{\left[(1-R_{\MM}^2) + R_{\MM}^2 u\right]^{\frac{n-1}{2}}}.
\end{equation}
Under the general tCCH prior \eqref{eq:tCCH_mixture_of_g},
the marginal likelihood 
$p(\mathbf{Y} \mid \MM) = \int_0^1 p(\mathbf{Y} \mid u, \MM) p(u) du$
lacks a known closed form expression, however,
it is analytically tractable under  the special cases discussed in Section 
\ref{section:CH-g}.  We present results for the general normal linear
model relaxing the assumption that $\X$ is full rank as suggested by
\citet{Liang_etal_2008}.

\begin{proposition}\label{proposition:normal_marlik}
  Consider a linear model
  $\Y \mid \alpha, \b_\MM, \sigma^2 \sim \N(\one_n \alpha + \X^c_\MM
  \b_\MM, \sigma^2 \W^{-1})$, with $\W$ a fixed $n \times n$ positive
  definite matrix and centered predictors
  $\X^c_\MM = (\I_n - \P_{\one_n})\X_\MM$ where $\P_{\mathbf{1}_n}$ is
  the orthogonal projection onto the column space spanned by $\one_n$
    using $\W$ in place of the
  observed information in \eqref{eq:projection}. Define the
  coefficient of determination as
\begin{equation}
  \label{eq:R2-GLM}
  R^2_\MM  = \frac{\| \P_{\X^c_\MM}\Y \|^2_\W}{\| (\I_n - \P_{\one_n})\Y \|^2_\W}
\end{equation}
where $\|\uv\|^2_\W = \uv^T\W \uv$ for $\uv \in \bbR^n$ and 
$\P_{\X^c_\MM} = \X^c_\MM(\X^{cT}_\MM \W \X^c_\MM)^{-} \X^{cT}_\MM\W$
is the rank $\rank_\MM$ 
orthogonal projection onto the column space spanned by
$\X^c_\MM$ using the information  inner product with $\W$. Under the prior
distributions $p(\alpha, \sigma^2) \propto 1/\sigma^2$, $g$-prior
$\b_\MM \mid \sigma^2, g, \MM \sim N(\zero, g \sigma^2 (\X^{cT}_\MM \W
\X^c_\MM)^{-})$, and the tCCH prior on $1/(1 + g)$, analytic
expressions for marginal likelihoods are available for the following
cases:
\begin{enumerate}[label = (\arabic*)]
\item if $r = 0$ (or equivalently, $\kappa = 1$), then
\begin{equation}\label{eq:normal_marlik_CHg}
  p(\mathbf{Y} \mid \MM, \W) 
= \frac{p(\mathbf{Y} \mid \Mnull, \W)~ B\left( \frac{a + \rank_\MM}{2}, \frac{b}{2} \right)
   \Phi_1\left( \frac{b}{2}, \frac{n-1}{2}, \frac{a+b+\rank_\MM}{2}, \frac{s}{2v}, \frac{R_{\MM}^2}{v-(v-1)R_{\MM}^2} \right)}
   {v^{\frac{\rank_\MM}{2}} \left[ 1-(1-\frac{1}{v})R_{\MM}^2 \right]^{\frac{n-1}{2}}B\left( \frac{a}{2}, \frac{b}{2} \right)
   {_1}F_1\left(\frac{b}{2},  \frac{a+b}{2}, \frac{s}{2v} \right)};
\end{equation}

\item if $s = 0$, then
\begin{align}\label{eq:normal_marlik_robust}
& p(\mathbf{Y} \mid \MM, \W) 
 = \frac{p(\mathbf{Y} \mid \Mnull, \W)~ \kappa^{\frac{a+\rank_\MM - 2r}{2}}
	B\left( \frac{a + \rank_\MM}{2}, \frac{b}{2} \right) } 
   	{v^{\frac{\rank_\MM}{2}} (1-R_\MM^2)^{\frac{n-1}{2}}
	B\left( \frac{a}{2}, \frac{b}{2} \right)
   	{_2}F_1\left(r, \frac{b}{2}; \frac{a+b}{2}, 1-\kappa \right)} \\ \nonumber
& \cdot F_1\left( \frac{a + \rank_\MM}{2}; \frac{a+b+\rank_\MM+1-n-2r}{2}, \frac{n-1}{2}; 
   	\frac{a + b + \rank_\MM}{2}; 1-\kappa, 
	1-\kappa - \frac{R_\MM^2\kappa}{(1-R_\MM^2)v} \right).
\end{align}
\end{enumerate}
Furthermore, if the rank of  $\P_{\X_\MM^{c}}$ is $n-1$, the Bayes factor 
$\BF_{\MM, \Mnull} = 1$.
\end{proposition}
A proof of Proposition \ref{proposition:normal_marlik} is provided in
supplementary material Appendix \ref{PROOFproposition:normal_marlik},
along with a brief summary of relevant special functions in
supplementary material Appendix \ref{subsection:special_functions}.
Note that (1) applies to the CH prior and all its special cases, and
(2) applies to robust, hyper-$g/n$, and intrinsic priors.

Similarly, 
the CHIC $g$-prior also yields tractable marginal likelihoods
for the double exponential family \citep{West_1985, Efron_1986}, 
which permits over-dispersion in GLMs
by introducing an unknown dispersion parameter $\phi$:
\begin{equation}\label{eq:de}
p(Y_i \mid \theta_i, \phi) 
 = \phi^{\frac{1}{2}}
	p(Y_i \mid \theta_i)^{\phi}p(Y_i \mid \theta_i = t_i)^{1 - \phi}, \quad i = 1, \ldots, n,
\end{equation}
where $p(Y_i \mid \theta_i)$ follows the GLM density \eqref{eq:likelihood},
and $t_i = \arg\max_{\theta_i} p(Y_i \mid \theta_i)$ is a constant that depends
on the data.
In this formulation, the MLEs $\alpha_\MM, \b_\MM$
do not depend on $\phi$ and the observed information of $\alpha_{\MM},
\boldsymbol\beta_{\MM}$ is block diagonal 
$\mathcal{J}_{n, \phi}\left( \hat{\alpha}_\MM, \hat{\boldsymbol\beta}_\MM\right) 
= {\diag}
\left\{
\phi\mathcal{J}_n( \hat{\alpha}_\MM), \phi\mathcal{J}_n( \hat{\boldsymbol\beta}_\MM) 
\right\},
$
where $\mathcal{J}_n( \hat{\alpha}_\MM)$ and 
$\mathcal{J}_n( \hat{\boldsymbol\beta}_\MM)$ are the observed information
matrices for standard GLMs as in \eqref{eq:J_alpha} and \eqref{eq:J_beta}.
A CHIC $g$-prior to account for over-dispersion based on the observed information
\[
\boldsymbol\beta_\MM \mid g, \MM   \sim 
\text{N}\left( \mathbf{0}, ~ \frac{g}{\phi}\cdot \mathcal{J}_n( \hat{\boldsymbol\beta}_\MM)^{-1} \right),\quad p(\alpha)\propto 1,\quad p(\phi) \propto \phi^{-1},
\]
provides closed form approximate marginal likelihoods after 
integrating out  $\phi$
\begin{equation}\label{eq:de_marlik}
p (\mathbf{Y} \mid \MM, u) \propto~ \frac{\left[ \mathcal{J}_n(\hat{\alpha}_\MM) \right]^{-\frac{1}{2}} u^{\frac{p_{\MM}}{2}}}
  	{\left\{ u Q_{\MM} + 2\sum_{i=1}^n \left[
	Y_i (t_i - \hat{\theta}_i) - b(t_i) + b(\hat{\theta}_i) 
	\right] \right\}^{\frac{n-1}{2}} }.
\end{equation}
A derivation of \eqref{eq:de_marlik} is provided in 
supplementary material Appendix \ref{PROOFeq:de_marlik}.
Since the kernel function of $u$  \eqref{eq:de_marlik} is of the same form
as \eqref{eq:normal_marlik}, there exists a similar result to Proposition 
\ref{proposition:normal_marlik} for tractable marginal likelihoods
after integrating out $u$ under the CHIC prior.

The CHIC $g$-prior provides a rich and unifying framework that encompasses
several common mixtures of $g$-priors.
However, this full six-parameter family poses an overwhelming range of 
choices to elicit for applied statisticians.
As many of
the parameters are not updated by the data, we appeal to the model
selection criteria or desiderata  proposed by \citet{Bayarri_etal_2012} to help
in recommending priors from this class.

\section{Desiderata for Model Selection Priors}\label{section:criteria}

\citet{Bayarri_etal_2012} establish primary criteria that
priors for model selection or model averaging should ideally satisfy.

\subsection{Basic Criterion} \label{section:basic}

The basic criterion requires the conditional prior distributions 
$p(\boldsymbol\beta_\MM \mid \MM, \alpha)$ to be
proper, so that Bayes factors do not
contain different arbitrary normalizing constants across different subset models
\citep{Kass_Raftery_1995}. 
This criterion does not require specification of a proper prior on $\alpha$, 
nor orthogonalization of $\alpha$ \citep{Bayarri_etal_2012}.
For the $g$-prior \eqref{eq:g-prior_GLM1}, under any model $\MM$,
as long as the observed information $\mathcal{J}(\hat{\boldsymbol\beta}_\MM)$
is positive-definite,
the prior distribution $p(\boldsymbol\beta_\MM \mid g, \MM)$ is a normal distribution, 
and hence the basic criterion holds. It also holds under mixtures of
$g$-priors for any proper prior distribution on $g$.
The basic criterion eliminates the Jeffreys prior on $g$, 
unless the null model is not within consideration.

\subsection{Invariance}\label{section:measurement_invariance}

Measurement invariance suggests that answers should not be affected by
changes of measurement units, i.e., location-scale transformation of
predictors.  Under the $g$-prior
\eqref{eq:g-prior_GLM1}, 
the prior covariance on $\boldsymbol\beta_\MM$ is proportional to
$\left[\mathbf{X}_{\MM}^{cT} \
  \mathcal{J}_n(\hat{\boldsymbol\eta}_\MM)\
  \mathbf{X}_{\MM}^{c}\right]^{-1}$.
If the design matrix is rescaled to $\mathbf{X}_\MM \mathbf{D}$, where
$\mathbf{D}$ is a positive definite diagonal matrix, then the
normalized design $\mathbf{X}_{\MM}^{c}$ becomes
$\mathbf{X}_{\MM}^{c}\mathbf{D}$, and coefficients are rescaled to
$\mathbf{D}^{-1}\boldsymbol\beta_\MM$.  Since the MLE
$\hat{\boldsymbol\eta}_\MM$ remains the same, the prior distribution
on $\boldsymbol\beta_\MM$ is invariant under rescaling.  Furthermore,
the prior on $\boldsymbol\beta_\MM$ is also invariant under
translation, since shifting columns of $\mathbf{X}_\MM$ does not
change $\boldsymbol\beta_\MM$ or $\mathbf{X}_{\MM}^{c}$.  The uniform
prior on $\alpha$ \eqref{eq:alpha-prior} combined with the CHIC $g$-prior
ensures that the prior on $\eg$ is invariant under linear transformations.
For models with unknown variance, the reference prior on $\sigma^2$ in
\eqref{eq:sigma-prior} ensures invariance under scale transformations.

\subsection{Model Selection Consistency}\label{sec:selection_consistency} 
Model selection consistency 
\citep{Fernandez_etal_2001}
has been widely used as a crucial criterion in prior specification.
Based on Bayes rule under the 0-1 loss, a prior distribution is consistent 
for model selection if as $n \rightarrow \infty$, the posterior probability of $\MT$ converges 
in probability to one, or equivalently, the Bayes factor tends to infinity
\[
p(\MT \mid \mathbf{Y})  \stackrel{P}{\longrightarrow} 1 \Longleftrightarrow
\BF_{\MT : \MM}  \stackrel{P}{\longrightarrow}  \infty,
\text{ for all } \MM \neq \MT,
\]
under fixed $p$ and bounded prior odds $p(\MT) / p(\MM)$. 
For normal linear regressions, Zeller-Siow, hyper-$g/n$, and the robust priors
have been shown to be consistent \citep{Liang_etal_2008,
  Bayarri_etal_2012}, while for
GLMs, the Zeller-Siow and hyper-$g/n$ priors based on the null
based $g$-prior in
\eqref{eq:g-prior_SBH2011} have been shown to be consistent
\citep{Wu_etal_2016}. 
We establish consistency for special cases of the CHIC $g$-prior
in Table \ref{tb:tCCHg_parameters}.

\begin{theorem}\label{th:selection_CHg}
When $\MT \neq \Mnull$,
 model selection consistency holds under the robust prior, 
 the intrinsic prior, the CH prior, and the local EB $g$-prior.
When $\MT = \Mnull$,  consistency still holds under 
the robust prior, 
 the intrinsic prior, and the CH prior with $b = O(n)$ or $s = O(n)$,
 but not under the local EB.
\end{theorem}

The proof is available in supplementary materials Appendix
\ref{PROOFth:selection_CHg}.  Note that for the CH priors, the result
also holds if the parameters $a, b, s$ are model specific (for
example, the parameters in the Beta-prime prior depends on
$p_\MM$). As revealed in Table \ref{tb:tCCHg_parameters}, among the
mixtures $g$-priors, model selection consistency holds under all but
the three hyper-$g$ prior variants, where consistency fails under the
null model.  Priors that are globally consistent imply
prior choices of $g = O(n)$, which will be discussed in Section
\ref{section:intrinsic}.  This corresponds to flatter priors on
$\b_\MM$, which imposes enough penalty on model sizes, so that the
selection consistency holds even when $\MT = \Mnull$.

\subsection{Information Consistency}\label{section:small_sample}

In normal linear regression, with a fixed sample size $n > p_\MM + 1$,
the information consistency fails under the $g$-prior \eqref{eq:gprior} with fixed $g$  
\citep{Liang_etal_2008}, in the sense that the Bayes factor $\BF_{\MM:\Mnull}$ 
\eqref{eq:normal_marlik} is bounded when model $\MM$ fits all observations
perfectly, i.e., $R^2 = 1$ or $F \to \infty$,  
although in principle it should favor $\MM$ overwhelmingly over $\Mnull$.
\citet{Bayarri_etal_2012} reformulate the information consistency 
as follows: 
If there exists a sequence of datasets with the same sample size $n$ such that 
the likelihood ratio between $\MM$ and $\Mnull$ goes to infinity,
then their Bayes factor should also go to infinity.

GLMs with categorical responses such as binary and Poisson regressions, 
have likelihood functions based on probability mass functions, which
have a natural  upper bound $1$, so that  even under data separation for
binary data,  the likelihood ratio remains bounded, and hence information
consistency is not an issue for these GLMs for any prior that
satisfies the basic criterion.

\subsection{Intrinsic Consistency}\label{section:intrinsic}

The intrinsic consistency suggests that as $n$ increases, the limit
distribution of the prior 
$p(\boldsymbol\beta_\MM \mid \alpha, \MM)$ should be independent of $n$
and remain proper, instead of degenerating to a point mass \citep{Bayarri_etal_2012}.
By Lemma \ref{lemma:order_Jalpha_Jbeta} in the supplementary materials,
$\mathcal{J}_n(\hat{\boldsymbol\beta}_\MM) = O_P(n)$ if $\MM \supset \MT$, 
so with any fixed value of $g$, the $g$-prior \eqref{eq:g-prior_GLM1} 
depends implicitly on $n$, and reduces to a point mass at zero asymptotically.
Hence in the $g$-prior or mixtures of $g$-priors, 
the choice $g = O(n)$ is essential to prevent the $g$-prior from 
dominating the likelihood. 

The intrinsic consistency is shown to hold under
the robust prior, since the prior density 
of $g/n$ does not depend on $n$ in the limit
\citep{Bayarri_etal_2012}.
In this sense, other existing priors such as the unit information prior
($g$ set to be $n$), 
Zellner-Siow, hyper-$g/n$, and intrinsic priors also satisfy the intrinsic consistency.
On the other hand, for some mixtures of $g$-priors, whose induced 
prior densities $p(g/n)$ lack closed forms, an implicit version of the intrinsic consistency
that states $\mathbb{E}(1/g) = O(1/n)$ can be studied. 
This implicit intrinsic consistency is shown to hold under 
the Beta-prime prior \citep{Maruyama_George_2011}.
We  show that it also holds under the CH prior in the following proposition, 
with certain hyper parameters. 
\begin{proposition} \label{proposition:CH-g_intrinsic}
Under the CH prior, if the parameters $b = O(n)$ or $s = O(n)$, 
then the prior expectation $\mathbb{E}(1/g) = O(1/n)$ as $n$ goes to infinity.
\end{proposition}
The proof is provided in supplementary materials Appendix 
\ref{PROOFproposition:CH-g_intrinsic}.
In contrast, the $g$-prior with fixed $g$, the hyper-$g$ prior and 
its special cases are 
eliminated due to their $g=O(1)$ choices.
Note that for the CHIC family, the intrinsic consistency and the previously discussed
model selection consistency hold under the same conditions.

\subsection{Estimation Consistency}\label{sec:BMA_consistency}

Parameter estimation is an essential part of regression analysis, with
or without model selection.  When $\MT$ is known and
$\MT \neq \Mnull$, one detractor of the $g$-prior with fixed $g$ is
that the approximate posterior mean
$\mathbb{E}[\boldsymbol\beta_\MT \mid \mathbf{Y}, g, \MT] = g/(1+g)
\hat{\boldsymbol\beta}_\MT \stackrel{\text{P}}{\longrightarrow}
g/(1+g) \boldsymbol\beta^*_\MT$
remains biased asymptotically as $n$ tends to infinity.  For mixtures
of $g$-priors, since the distribution of $g$ adapts to the  data, a sufficient
condition to resolve this asymptotic bias is for the posterior
distribution of the shrinkage factor $z = g/(1+g)$ to converge to $1$
in the limit.

\begin{proposition}\label{prop:consistency_shrinkage_factor1}
For the CH, robust, and intrinsic priors, 
when $\MT \neq \Mnull$,  
 the characteristic function of the conditional posterior distribution $z=g/(1+g)$ under $\MT$
converges  in probability to that of a degenerate distribution at $1$, i.e.,
for any $t \in \mathbb{R}$,
$\phi_{z\mid \mathbf{Y}, \MT}(t)  \stackrel{\triangle}{=} 
\mathbb{E}\left( e^{itz} \right)  \stackrel{\text{P}}{\longrightarrow} \exp(it)$.
Therefore, all moments of $p(z\mid \mathbf{Y}, \MT)$ converge
to $1$ in probability. In particular,
the posterior mean $\mathbb{E}(z\mid \mathbf{Y}, \MT) \stackrel{\text{P}}{\longrightarrow} 1$ and  
the posterior variance $\mathbb{V}(z\mid \mathbf{Y}, \MT) \stackrel{\text{P}}{\longrightarrow} 0$.
\end{proposition}

The proof is given in supplementary materials 
Appendix \ref{PROOFprop:consistency_shrinkage_factor1}.


When $\MT$ is unknown, one may prefer Bayesian model averaging (BMA)
estimators to account for model uncertainty. In BMA, $\boldsymbol\beta$
denotes the $p$ dimensional vector of coefficients corresponding to
all potential predictors, while $\boldsymbol\beta_\MM$ is typically
length $p_\MM$ vector of the nonzero coefficients. With a slight
over-use of notation, we let $\b_{\MM}$ denote the length $p$ vector, with
zeros filled for the dimensions not included in $\MM$.  The posterior
 of $\boldsymbol\beta$ under BMA is thus 
\begin{equation}\label{eq:bma}
p(\boldsymbol\beta \mid \mathbf{Y})
= p(\MT\mid \mathbf{Y})\ 
  p(\boldsymbol\beta_{\MT} \mid \mathbf{Y}, \MT) + 
  \sum_{\MM \neq \MT} p(\MM\mid \mathbf{Y})\ 
  p(\boldsymbol\beta_{\MM} \mid \mathbf{Y}, \MM)
\end{equation}
where conditional posterior distributions $p(\boldsymbol\beta_{\MM} \mid \mathbf{Y}, \MM)
 = \int p\left(\boldsymbol\beta_{\MM} \mid \mathbf{Y}, g, \MM\right) 
 p\left(g \mid \mathbf{Y}, \MM \right) d g$ for all subset models $\MM \neq \Mnull$.
When the selection consistency holds, i.e., $p(\MT\mid \mathbf{Y})\stackrel{\text{P}}{\longrightarrow}1$,
the second term in \eqref{eq:bma} vanishes in the limit, 
so we just need to study the posterior distribution of $\boldsymbol\beta_{\MT}$.
When $\MT = \Mnull$, even if the selection consistency fails,
consistency of the MLEs yields the correct estimation of the true
parameter $\b_\MT^* = \zero$, with or without shrinkage. 
\begin{theorem}\label{thm:consistency_post_beta}
For the CH, robust, and intrinsic priors,
the characteristic function of the posterior distribution under BMA $p(\boldsymbol\beta \mid \mathbf{Y})$
converges  in probability to that of a degenerate distribution at $\boldsymbol\beta_\MT^*$; i.e., 
for any $\mathbf{t} \in \mathbb{R}^p$,
$\phi_{\boldsymbol\beta\mid \mathbf{Y}}(\mathbf{t})  \stackrel{\text{P}}{\longrightarrow} 
e^{i\mathbf{t}^T \boldsymbol\beta_{\MT}^*}$.
In particular, the mean and covariance of the
posterior distribution of $\boldsymbol\beta$ under model averaging have limits
$\mathbb{E}(\boldsymbol\beta \mid \mathbf{Y}) \stackrel{\text{P}}{\longrightarrow} \boldsymbol\beta_\MT^*$
and $\mathbb{V}(\boldsymbol\beta \mid \mathbf{Y})  \stackrel{\text{P}}{\longrightarrow} \mathbf{0}$.
\end{theorem}
A proof is given in supplementary materials Appendix \ref{PROOFthm:consistency_post_beta}.
Note, this estimation consistency for $\b$ also implies estimation
consistency for $\boldsymbol{\eta}$ and functions
of $\boldsymbol{\eta}$.

\subsection{Predictive Matching}

Predictive matching is viewed as one of the most crucial aspects for
objective model selection priors as improper scaling of priors may
have critical consequences for comparing models in high dimensional
problems \citep{Bayarri_etal_2012}.  Jeffreys suggests that when
comparing two models with minimal sample sizes where one should not be
able to discriminate between them,  the Bayes factor should
be close to one. In particular, exact predictive matching occurs if
it equals one.  The minimal training sample is defined by
\citet{Bayarri_etal_2012} as the smallest sample size with a finite
nonzero marginal density for the combination of models and priors.
For normal linear models with unknown variance, the minimal sample
size is 2 (or the number of parameters in the null model) and 
exact predictive matching occurs under the CHIC
$g$-priors. For GLMs with known
dispersion, the minimal training sample size would be 1.  The
asymptotic approximations of course do not apply in such a case,
however, for a minimal sample size and a model for which
$\J(\eg) \neq \zero$ but $\b_\MM$ is not identifiable, the results
from Proposition \ref{proposition:non_full_rank_BF} establish
that exact null predictive matching holds under the CHIC $g$-prior.

\section{Examples}\label{sec:examples}

We explore properties of the priors in finite samples for logistic
regression via simulation studies under a range of sparsity scenarios.
Results from Poisson regression reveal similar findings to the
logistic simulation study, and are included in supplementary material
Appendix \ref{section:poisson}. We then turn to a re-analysis of the
GUSTO-I data considered in \citet{Held_etal_2015} to illustrate the
methodology and compare prior distributions for estimation of
posterior inclusion probabilities and out-of-sample predictive
performance. \blind{An R package}{The R package {\tt BAS}}, available on CRAN, is used for all
computations in this section.

\subsection{A Simulation Study}\label{section:simulation}

We conduct a simulation to explore properties of the priors for model
selection and estimation in logistic regression using $p=20$ and
$p=100$ predictors and under different designs for $\X$.  For each
simulated dataset, we take $n = 500$ with the columns of $\X$ drawn
from standard normal distributions, which have pairwise correlation
$\text{cor}(\mathbf{X}_i, \mathbf{X}_j) = r^{|i - j|}$ for
$1\leq i < j \leq p$, with $r = 0$ (independent design) or $r = 0.75$
(correlated design). We consider four different levels of sparsity in the
true model (see Table \ref{tb:true_models}) for $p=20$.  For $p = 100$, we
consider only the sparse scenario where $p_\MT = 5$, with additional
coefficients $\boldsymbol\beta^*_{\MT, 21:100} = \mathbf{0}$.
For $p = 20$, we enumerate among all $2^{20}$ subset models using a
uniform distribution over the model space, 
$p(\MM) = 1/ 2^p$, which assigns every models equal prior weights.
For $p = 100$, we use the MCMC algorithm in \citet{Clyde_etal_2011} with
$2^{17} \approx 131,000$ iterations. 
In addition to the uniform prior, we also consider 
the Beta-Binomial$(1, 1)$ prior over the model space,
$p(\MM) = (p+1)^{-1} {p \choose p_\MM}^{-1}$, which is recommended for
multiplicity adjustment in Bayesian variable selection for large $p$ as
it puts uniform weights on model sizes $0, 1, \ldots, p$  
\citep{Ley_Steel_2009}  and encourages sparsity when $p_{\MT} \ll p/2$.

\begin{table}[ht]
  \caption[Simulation: four scenarios of true models]
{Values of the intercept and coefficients $(\alpha^*_\MT, \boldsymbol\beta^*_\MT)$
in the true models in the logistic regression simulation study 
with $p = 20$, where  $\mathbf{b} = (2, -1, -1, 0.5, -0.5)^T$.}
\label{tb:true_models}
\begin{center}
\begin{tabular}{| l | r | r|rrrr |}
  \hline
Scenario	& $p_{\MT}$	& $\alpha^*_\MT$ 	& $\boldsymbol\beta^*_{\MT, 1:5}$ &$\boldsymbol\beta^*_{\MT, 6:10}$ 	&$\boldsymbol\beta^*_{\MT, 11:15}$ 	&$\boldsymbol\beta^*_{\MT, 16:20}$ 	\\
  \hline
Null 		& $0$		&\multirow{4}{*}{$-0.5$}	&	$\mathbf{0}$			&	$\mathbf{0}$					&	$\mathbf{0}$					&	$\mathbf{0}$					 \\ 
Sparse 		& $5$		&					&$\mathbf{b}$					&	$\mathbf{0}$					&	$\mathbf{0}$					&	$\mathbf{0}$					 \\ 
Medium		& $10$		&					&$\mathbf{b}$					&	$\mathbf{0}$					&	$\mathbf{b}$					&	$\mathbf{0}$					 \\ 
Full 		& $20$		&					&$\mathbf{b}$					&	$\mathbf{b}$					&	$\mathbf{b}$					&	$\mathbf{b}$					 \\ 
   \hline
\end{tabular}
\end{center}
\end{table}

For model selection, we select the model with the highest posterior
probability (or the smallest AIC, BIC) under a 0-1 loss. Table
\ref{tb:selection_logistic_top} displays the number of times $\MT$ is
selected in 100 simulations under each scenario, while Table
\ref{tb:selection_logistic_averagesize} in the supplementary materials
shows the average size of the selected models. 
The fully Bayes methods can be roughly divided into two groups
according to their prior concentration preference: $g = O(n)$ and $g = O(1)$.
The $g=O(n)$ group, including all the special cases
of the CHIC prior that satisfy model selection and intrinsic consistency (see Table \ref{tb:tCCHg_parameters}), 
lead to more parsimonious models, and hence outperform the rest of the methods
in  scenarios where the full model is not true, 
while the $g=O(1)$ group, including the hyper-$g$ prior and its 
special cases, 
are more accurate only when the full model is true.  These result also
confirm the theoretical findings in Section
\ref{sec:selection_consistency}  
and in \citet{Liang_etal_2008}, that the priors on $g$ independent of $n$
are not consistent for model selection\footnote{
Since the Jeffreys prior is improper, when implementing it,
the null model is always excluded.} when $\MT = \Mnull$.
Interestingly, the hyper-$g/n$ prior, although in the $g = O(n)$ group,
performs closer to the hyper-$g$ prior variants,
when the full model is true, or when $p = 100$.
The results under the unit information prior, i.e., the $g$-prior with $g=n$, 
DBF and TBF yield almost identical results, which is also noted by
\citet{Held_etal_2015} and provide results that are intermediate. 
Both can outperform mixtures of $g$-priors in the $g=O(n)$ group 
when the true model is sparse, but may not perform as
well as them when $\MT$ is the null model or the full model. 

Among non-fully Bayesian methods, the local EB tends to favor large
models, which is also noted in \citet{Hansen_Yu_2003}.  When
$\MT = \Mnull$, it never selects the correct model but surprisingly
almost always selects the full model (average model size is 19).
Between AIC and BIC, the former favors larger models while the latter
favors smaller ones.  BIC performs comparably to priors in the
$g=O(n)$ group as long as $\MT$ is not the full model.

The prior distribution over the model space also leads to significant
difference.  When 
$p = 100$ and $p_\MT = 5$, under most $g$-priors and mixtures of
$g$-priors, the Beta-Binomial$(1,1)$ prior favors sparser models than
the uniform prior, leading to more accurate model selection results.
However, it is the opposite case with the hyper-$g/n$ prior, the three
hyper-$g$ variants, and the local EB, for which the average model
sizes are large (around 70) under the uniform prior, but even larger
under the Beta-Binomial prior (close to 100).  This phenomenon can be
explained by the symmetric U-shaped density curve of the Beta-Binomial prior
\citep[Fig 1]{Scott_Berger_2010} --- where the null model and the full
model have the highest prior probabilities, among all individual models.  
For methods that lead to
marginal likelihoods that favor
model sizes larger than $p/2$, the Beta-Binomial(1,1) prior 
does not necessarily promote sparsity and may encourage selection of
the full model.

\setlength{\tabcolsep}{5pt}

\begin{table}[ht]
\centering
\caption[Logistic regression: model selection accuracy under 0-1 loss.]
{
Logistic regression simulation example:
number of times the true model is selected out of 100 realizations.
Column-wise maximum is in bold type.} \label{tb:selection_logistic_top}
\begin{tabular}{| l | rr|rr|rr|rr | rr | rr|}
  \hline
$p$ 			& \multicolumn{8}{c|}{20}	& \multicolumn{4}{c|}{100}	 \\ \hline
$p(\MM)$		&  \multicolumn{8}{c|}{Uniform}	&  \multicolumn{2}{c|}{Uniform}	 &  \multicolumn{2}{c|}{BB$(1,1)$}	\\ \hline
$p_\MT$ 		& \multicolumn{2}{c|}{0}	& \multicolumn{2}{c|}{5}	& \multicolumn{2}{c|}{10}	& \multicolumn{2}{c|}{20}	& \multicolumn{2}{c|}{5}	& \multicolumn{2}{c|}{5} \\ 
$r$ 				& 0		& 0.75	& 0		& 0.75	& 0		& 0.75	& 0		& 0.75	& 0		& 0.75		& 0		& 0.75		 \\ 
  \hline
CH$(a=1/2,b=n)$ & {\bf 92} & 88 & 61 & 29 & {\bf 38} & 8 & 6 & 0 & 11 & 11 & 61 & 6 \\ 
  CH$(a=1,b=n)$ & 85 & 82 & 60 & 30 & 37 & 8 & 6 & 0 & 15 & 9 & 61 & 6 \\ 
  CH$(a=1/2,b=n/2)$ & 86 & 84 & 46 & 28 & 30 & {\bf 12} & 8 & 0 & 3 & 2 & 62 & 6 \\ 
  CH$(a=1,b=n/2)$ & 70 & 73 & 45 & 30 & 30 & 11 & 8 & 0 & 8 & 4 & {\bf 63} & 6 \\ 
  Beta-prime & {\bf 92} & 88 & 61 & 29 & {\bf 38} & 8 & 7 & 0 & 11 & 6 & 61 & 6 \\ 
  ZS adapted & 85 & 82 & 60 & 30 & 37 & 8 & 6 & 0 & 8 & 11 & 61 & 6 \\ 
  Benchmark & 91 & {\bf 93} & 28 & {\bf 31} & 19 & 8 & 16 & 0 & 6 & 3 & 62 & 6 \\ 
  Robust & 86 & 83 & 41 & 29 & 29 & 10 & 8 & 0 & 4 & 1 & 52 & 5 \\ 
  Intrinsic & 76 & 77 & 40 & 29 & 26 & 10 & 8 & 0 & 2 & 3 & 56 & 5 \\ 
  Hyper-$g/n$ & 77 & 73 & 37 & {\bf 31} & 23 & 7 & 16 & 0 & 0 & 0 & 1 & 0 \\ 
  DBF, $g=n$ & 73 & 79 & {\bf 67} & 29 & 31 & 2 & 0 & 0 &  {\bf 68} & 26 & 55 & 3 \\ 
  TBF, $g=n$ & 73 & 79 & {\bf 67} & 29 & 31 & 2 & 0 & 0 &  {\bf 68} & 27 & 55 & 3 \\ 
  Jeffreys & NA & NA & 28 & 28 & 17 & 7 & 16 & 0 & 0 & 0 & 1 & 0 \\ 
  Hyper-$g$ & 6 & 9 & 25 & 29 & 15 & 8 & 16 & {\bf 1} & 0 & 0 & 0 & 1 \\ 
  Uniform & 2 & 5 & 23 & 24 & 14 & 6 & {\bf 18} & {\bf 1} & 0 & 0 & 0 & 0 \\ 
  Local EB & 0 & 0 & 25 & 29 & 15 & 7 & 16 & {\bf 1} & 0 & 0 & 0 & 0 \\ 
  AIC & 3 & 7 & 5 & 9 & 13 & 5 & 12 & 0 & 1 & 2 & {\bf 63} & {\bf 15} \\ 
  BIC & 73 & 79 & {\bf 67} & 29 & 31 & 2 & 0 & 0 & 67 &  {\bf 28} & 55 & 3 \\ 
   \hline
\end{tabular}
\end{table}

\begin{table}[ht]
\centering
\caption[Logistic regression: model selection accuracy under 0-1 loss.]
{
Logistic regression simulation example: 
$100$ times the average SSE
$=\sum_{j=0}^p (\tilde{\beta}_j - \beta_{j, \MT}^*)^2$ of 100 realizations.
Column-wise minimum is in bold type.
} \label{tb:selection_logistic_bma}
\begin{tabular}{| l | rr|rr|rr|rr | rr | rr|}
  \hline
$p$ 			& \multicolumn{8}{c|}{20}	& \multicolumn{4}{c|}{100}	 \\ \hline
$p(\MM)$		&  \multicolumn{8}{c|}{Uniform}	&  \multicolumn{2}{c|}{Uniform}	 &  \multicolumn{2}{c|}{BB$(1,1)$}	\\ \hline
$p_\MT$ 		& \multicolumn{2}{c|}{0}	& \multicolumn{2}{c|}{5}	& \multicolumn{2}{c|}{10}	& \multicolumn{2}{c|}{20}	& \multicolumn{2}{c|}{5}	& \multicolumn{2}{c|}{5} \\ 
$r$ 				& 0		& 0.75	& 0		& 0.75	& 0		& 0.75	& 0		& 0.75	& 0		& 0.75		& 0		& 0.75		 \\ 
  \hline
  CH$(a=1/2,b=n)$ & 3 & 3 & 21 & 44 & 51 & 96 & 94 & 184 & 109 & 135 & {\bf 26} & 78 \\ 
  CH$(a=1,b=n)$ & 3 & 4 & 21 & {\bf 43} & 51 & 96 & 94 & 183 & 119 & 139 & {\bf 26} & 77 \\ 
  CH$(a=1/2,b=n/2)$ & 4 & 5 & 22 & {\bf 43} & 50 & 92 & 87 & 172 & 158 & 182 & {\bf 26} & 75 \\ 
  CH$(a=1,b=n/2)$ & 4 & 5 & 22 & {\bf 43} & 50 & 92 & 86 & 172 & 160 & 189 & 27 & 74 \\ 
  Beta-prime & 3 & 3 & 21 & 44 & 51 & 96 & 94 & 183 & 123 & 142 & {\bf 26} & 78 \\ 
  ZS adapted & 3 & 4 & 21 & {\bf 43} & 51 & 96 & 94 & 183 & 121 & 144 & {\bf 26} & 77 \\ 
  Benchmark & 4 & 7 & 21 & 44 & 49 & 89 & 73 & {\bf 158} & 169 & 195 & {\bf 26} & 75 \\ 
  Robust & 4 & 5 & 23 & 44 & 52 & 91 & 90 & 165 & 252 & 292 & 193 & 139 \\ 
  Intrinsic & 4 & 6 & 23 & 44 & 52 & 91 & 90 & 165 & 239 & 284 & 143 & 90 \\ 
  Hyper-$g/n$ & 3 & 4 & 21 & {\bf 43} & {\bf 48} & {\bf 88} & {\bf 72} & {\bf 158} & 197 & 226 & 441 & 326 \\ 
  DBF, $g=n$ & 3 & 3 & {\bf 20} & 47 & 54 & 117 & 113 & 244 & {\bf 42} & {\bf 65} & 27 & 82 \\ 
  TBF, $g=n$ & 3 & 3 & {\bf 20} & 47 & 54 & 117 & 113 & 245 & {\bf 42} & {\bf 65} & 27 & 83 \\ 
  Jeffreys & 2 & 3 & 22 & 45 & 50 & 89 & 74 & 159 & 212 & 231 & 444 & 387 \\ 
  Hyper-$g$ & 2 & 3 & 22 & 45 & 51 & 90 & 76 & 160 & 219 & 233 & 451 & 396 \\ 
  Uniform & 2 & 2 & 22 & 46 & 52 & 91 & 78 & 161 & 230 & 236 & 459 & 411 \\ 
  Local EB & {\bf 1} & {\bf 1} & 22 & 45 & 50 & 89 & 74 & {\bf 158} & 245 & 236 & 608 & 434 \\ 
  AIC & 8 & 15 & 29 & 51 & 59 & 93 & 103 & {\bf 158} & 287 & 353 & 39 & {\bf 71} \\ 
  BIC & 3 & 3 & 21 & 47 & 55 & 117 & 113 & 245 & {\bf 42} & {\bf 65} & 27 & 82 \\ 
     \hline
\end{tabular}
\end{table}

Estimation and prediction are often more important than identifying
the true model, particularly for large $p$. To evaluate the
performance for parameter estimation, 
we report SSE$(\boldsymbol\beta) = \sum_{j=0}^p (\tilde{\beta}_j - \beta_{j, \MT}^*)^2$
in Table \ref{tb:selection_logistic_bma} where
$\tilde{\beta}_j$ represents the posterior mean estimates under BMA 
(here $\beta_0$ corresponds to the intercept $\alpha$);
while for AIC and BIC, this is the MLE under the selected model.
An overall trend is that  
the methods perform better in model selection generally yield smaller
estimation errors. One exception is the $g=O(1)$ priors and the local EB, which have
small SSE under the null despite their poor model selection
performance.  

We also examined the out-of-sample classification error for logistic
regression which revealed almost no difference across methods.

\setlength{\tabcolsep}{4pt}


\subsection{GUSTO-I Study}\label{subsection:gusto-i}

We use a publicly available subset of the GUSTO-I data\footnote{This
  dataset is available on the book website
  \href{http://www.clinicalpredictionmodels.org}{http://www.clinicalpredictionmodels.org}}
\citep{Steyerberg_2009, Held_etal_2015}, containing $n=2188$ patients 
to illustrate the methodology for predicting a
binary endpoint of 30 day survival for myocardial infarction. We use
the same $p=17$ predictors as in \citet{Held_etal_2015}, labeled in
the same order. 
 
\begin{figure}
\begin{center}
\includegraphics[width = 0.8\textwidth]{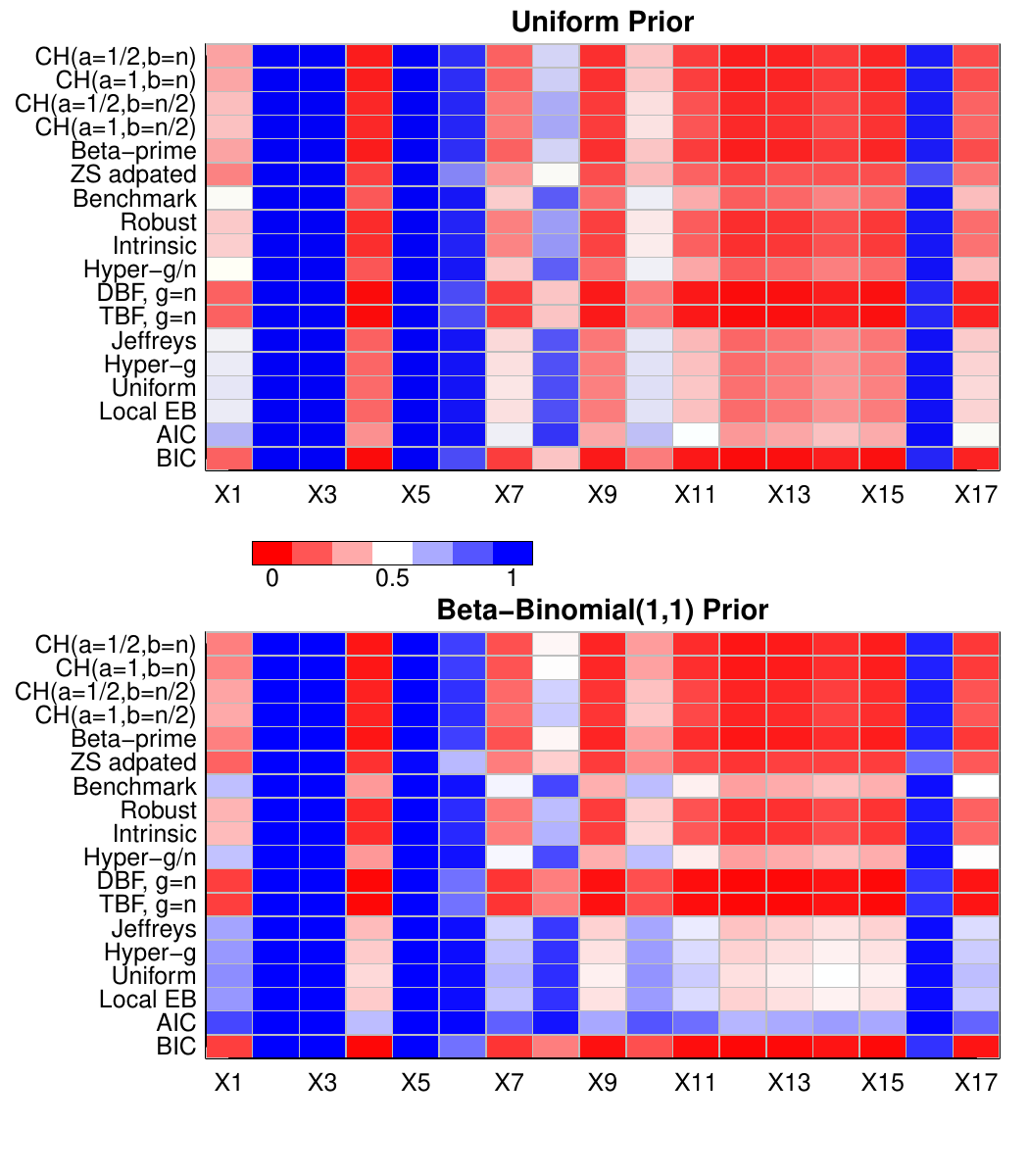}
\caption{Marginal posterior inclusion probabilities for the GUSTO-I
  data. The colors are related to the magnitude of the inclusion probability
with darkest blue corresponding to one and red to zero, while 0.5 is shown as white.}\label{figure:incprob_GUSTO-I}  
\end{center}
\end{figure}


Figure \ref{figure:incprob_GUSTO-I} illustrates heatmaps of the
marginal posterior inclusion probabilities (pip) for each of the 17
predictors under enumeration of all $2^{17}$ possible models in the
model space using a range of priors on $g$ and the uniform and
Beta-Binomial(1,1) prior distributions on the model space. For AIC and
BIC we use $\exp(-\text{AIC}/2)$ and $\exp(-\text{BIC}/2)$,
respectively, as in \citet{Burnham:Anderson:2004,Raft:1996} to approximate
posterior model probabilities.  

Figure \ref{figure:incprob_GUSTO-I} shows that the predictors $X_2,
X_3, X_5, X_6, X_{16}$ have high inclusion probabilities under all methods, 
reinforcing the findings in \citet{Held_etal_2015}.
Comparison across different methods reveals the same trend as
supported by theory 
and in the simulation studies: the $g=O(n)$ group and BIC
lead to sparser models than the $g=O(1)$ group, local EB, 
and AIC. Within the $g=O(n)$ group, 
the unit information prior, under either DBF or TBF, 
yields the most parsimonious model, while 
the benchmark and hyper-$g/n$ priors tend to select more predictors,
leading to results that are more similar to the $g=O(1)$ group.
As with the simulation study, the Beta-Binomial$(1,1)$  does not
automatically favor sparser models where inclusion probabilities are
higher  for a number of variables even in the $g=O(n)$ group compared
to the uniform prior.

To explore out-of-sample predictive performance, 
we use bootstrap cross-validation \citep{Fu_etal_2005}
to evaluate predictions under BMA.
For each of the 1000 bootstrap datasets, 
it is obtained via sampling with replacement,
with the same sample size $n=2188$. 
We fit the models on the bootstrap samples, 
 and then study prediction using the left out samples,
whose sample size is about one-third of $n$. 
As in \citet{Held_etal_2015}, we
summarize performance using the area under ROC curve (AUC),
calibration slope (CS), and logarithmic score (LS), and also include
the Brier score, i.e., the average squared difference between
$\hat{\mu}$ and $Y$.  Among these measurements, AUC and CS closer to
one indicate better discrimination and calibration, respectively, while
smaller LS suggests better discrimination and calibration, and
smaller Brier score indicates more accurate predictions.  Table \ref{tb:GUSTO-I_prediction} shows that overall the
methods perform similarly, with methods that prefer denser models in
selection, such as the benchmark, hyper-$g/n$, hyper-$g$, local EB,
and AIC, slightly outperforming the others. In particular, the uniform prior on $u$
(a special case  of the hyper-$g$ prior)
yields the most accurate prediction under all four summaries. 
Over the model space, the uniform prior slightly outperforms the
Beta-Binomial$(1,1)$, in terms of AUC, CS, and LS.

\begin{table}[ht]
\centering
\caption{Prediction accuracy for the GUSTO-I data, aggregated from 1000 bootstrap cross validation sets.
Bold font marks the largest AUC, the CS closest to one,
and the smallest LS and Brier score. }\label{tb:GUSTO-I_prediction}
\begin{tabular}{| l | cc | cc | cc | cc |}
  \hline
 &  \multicolumn{2}{c|}{AUC}  & \multicolumn{2}{c|}{CS} & \multicolumn{2}{c|}{LS} & \multicolumn{2}{c|}{Brier} \\ 
  \hline
$p(\MM)$		&  Unif	&	BB$(1,1)$  &  Unif	&	BB$(1,1)$&  Unif	&	BB$(1,1)$&  Unif	&	BB$(1,1)$	\\
  \hline
  CH$(a=1/2,b=n)$ & 0.8346 & 0.8338 & 0.9055 & 0.9065 & 0.1848 & 0.1851 & 0.0497 & 0.0497 \\ 
  CH$(a=1,b=n)$ & 0.8347 & 0.8339 & 0.9054 & 0.9063 & 0.1848 & 0.1851 & 0.0497 & 0.0497 \\ 
  CH$(a=1/2,b=n/2)$ & 0.8349 & 0.8343 & 0.9054 & 0.9049 & 0.1846 & 0.1849 & 0.0496 & 0.0497 \\ 
  CH$(a=1,b=n/2)$ & 0.8349 & 0.8343 & 0.9054 & 0.9048 & 0.1846 & 0.1849 & 0.0496 & 0.0497 \\ 
  Beta-prime & 0.8346 & 0.8338 & 0.9055 & 0.9065 & 0.1848 & 0.1851 & 0.0497 & 0.0497 \\ 
  ZS adapted & 0.8345 & 0.8329 & 0.9338 & 0.9382 & 0.1846 & 0.1854 & 0.0496 & 0.0498 \\ 
  Benchmark & {\bf 0.8352} & {\bf 0.8347} & 0.9292 & 0.9251 & 0.1841 & 0.1842 & {\bf 0.0495} & {\bf 0.0495}  \\ 
  Robust & 0.8349 & 0.8344 & 0.9012 & 0.8998 & 0.1847 & 0.1849 & 0.0496 & 0.0497 \\ 
  Intrinsic & 0.8350 & 0.8344 & 0.9010 & 0.8993 & 0.1846 & 0.1849 & 0.0496 & 0.0497 \\ 
  Hyper-$g/n$ & {\bf 0.8352} & 0.8346 & 0.9287 & 0.9265 & 0.1841 & 0.1842 & {\bf 0.0495}  & {\bf 0.0495}  \\ 
  DBF, $g=n$ & 0.8338 & 0.8325 & 0.9100 & 0.9126 & 0.1852 & 0.1857 & 0.0498 & 0.0499 \\ 
  TBF, $g=n$ & 0.8338 & 0.8325 & 0.9101 & 0.9126 & 0.1852 &  0.1857 & 0.0498 & 0.0499 \\ 
  Jeffreys & {\bf 0.8352} & 0.8346 & 0.9392 & 0.9373 & 0.1840 & 0.1841 & {\bf 0.0495}  & {\bf 0.0495}  \\ 
  Hyper-$g$ & {\bf 0.8352} & 0.8346 & 0.9446 & 0.9429 & {\bf 0.1839} & {\bf 0.1840} & {\bf 0.0495}  & {\bf 0.0495}  \\ 
  Uniform & {\bf 0.8352} & 0.8346 & {\bf 0.9502} & {\bf 0.9485} & {\bf 0.1839} & {\bf 0.1840} & {\bf 0.0495}  & {\bf 0.0495}  \\ 
  Local EB & {\bf 0.8352} & 0.8346 & 0.9391 & 0.9373 & 0.1840 & 0.1841 & {\bf 0.0495}  & {\bf 0.0495}  \\ 
  AIC & 0.8351 & 0.8344 & 0.8813 & 0.8645 & 0.1846 & 0.1850 & {\bf 0.0495}  & 0.0496 \\ 
  BIC & 0.8338 & 0.8325 & 0.9096 & 0.9122 & 0.1852 &  0.1857 & 0.0498 & 0.0499 \\ 
   \hline
\end{tabular}
\end{table}

One potential explanation for the better performance 
of the $g = O(1)$ and the local EB
is that shrinkage is better calibrated to the data
by avoiding over-fitting \citep{Copas_1983}.
As the shrinkage factor $g/(1+g)$ increases with $g$, 
the $g = O(1)$ priors and the local EB tend to impose 
stronger shrinkage than the $g=O(n)$ priors. 
For the GUSTIO-I dataset, 
the BMA posterior estimate of $g$ 
is $14.7$ for the uniform prior on $u$, $16.5$ for hyper-$g$,
$18.4$ for local EB, $24.0$ for benchmark, $25.4$ for hyper-$g/n$, 
$50.0$ for ZS adapted, 
$286.5$ for intrinsic, $298.1$ for CH$(a=1, b = n, s = 0)$, 
$319.6$ for Beta-prime, and $321.2$ for robust prior\footnote{
For all special cases of the CHIC $g$-prior, 
the posterior estimates of $g$ are 
converted from the approximate conditional 
posterior means of $u = 1/(1+g)$, which have closed form expressions.
These estimates of $g$ are computed under the uniform prior 
on models $p(\MM) = 1/2^p$.}. 
Comparing these estimates with the data likelihood of $g$
marginalized over the model space
$p(\Y \mid g) = \sum_{\MM} p(\Y \mid \MM, g) p(\MM \mid g)$, 
we find that estimates of $g$ from the $g = O(1)$ priors, local EB, 
benchmark, and hyper-$g/n$ priors are closer to the peak $g \approx 20$ 
of the marginal likelihood (see Figure \ref{figure:marlik_g_GUSTO-I}).
On the other hand, as noted by \citet{Ley_Steel_2012}, 
the robust and intrinsic priors,
which truncate the range of $g$ above $(n-p_\MM)/(p_\MM + 1)\geq 120.6$
and $n/(p_\MM+1)\geq 121.6$, respectively, may not be well supported by the data,
when $n$ is large and $p$ is small like the GUSTO-I data. 

\begin{figure}
\begin{center}
\includegraphics[width = 0.6\textwidth]{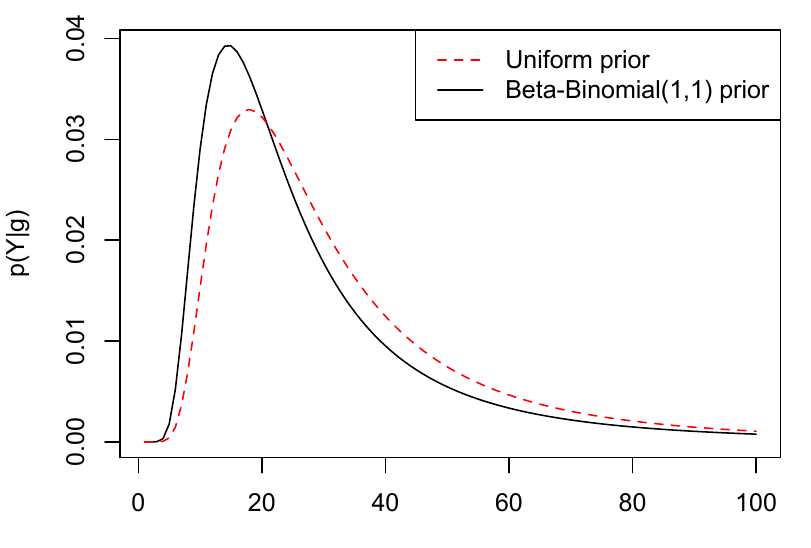}
\caption{Marginal likelihood of $g$ for the GUSTO-I data ($n = 2188$ and $p = 17$).}\label{figure:marlik_g_GUSTO-I}
\end{center}
\end{figure}

\section{Conclusion}\label{section:conclusion}

The analytic tractability of asymptotic CHIC posterior distributions
allows insight into their theoretical properties and fast computation,
serving as a robust conventional prior for most instances.
The CHIC family encompasses
the majority of mixtures of $g$-priors used in practice.
Under a wide range of hyper parameter choices, CHIC $g$-priors
satisfy various desiderata proposed by \citet{Bayarri_etal_2012} for
BVS. Based on both theoretical and empirical studies, we recommend
priors with the choice $g = O(n)$, such as the CH prior with
$b = O(n)$ or $s=O(n)$, hyper-$g/n$, Beta-prime, ZS adapted, benchmark, robust,
intrinsic, and unit information priors. For prediction, all methods
yield similar accuracy and are asymptotically consistent, with the
local EB, hyper-$g$, benchmark, and hyper-$g/n$ priors which favor
larger models slightly outperforming the rest of the $g = O(n)$ group.
Because
model selection and prediction are two unaligned goals with different
objective functions \citep{Copas_1983}, it is not surprising that no
single prior overwhelmingly outperform others for both goals.  Similar
to the findings of \citet{Ley_Steel_2012} in linear models, we also
recommend the benchmark and hyper-$g/n$ priors for general
practitioners, due to their balanced performance in selection and
prediction.

The approximate marginal likelihoods under the CHIC $g$-prior require
only simple summaries from GLMs, hence the CHIC $g$-prior has the same
computational complexity as model fitting for GLMs, leading to
efficient algorithms for variable selection and model averaging under
enumeration.  As $p$ increases (e.g., larger than 35) and enumerating
the entire model space becomes impractical, stochastic search
algorithms (see \citet{Clyde_etal_2011,
  Garcia-Donato_Martinez-Beneito_2013} and the references therein) can
be employed, while avoiding computationally expensive model search
alternatives such as the reversible jump MCMC \citep{Green_1995}, as
Bayes factors can be computed directly without sampling the model
specific parameters.  All of the methods used in the examples and
simulation studies within this article are implemented in the R
package {\tt BAS} \citep{Clyde_2016} available on CRAN.  Additional
refinements for calculating marginal likelihoods for GLMs with
canonical links can be obtained using a correction factor based on a
sixth-order Laplace approximation at little increase in computational
cost \citep{Raudenbush_etal_2000,Bove_Held_2011}.

Propositions \ref{proposition:non_full_rank_BF} and
\ref{proposition:normal_marlik} establish  that
$g$-priors are well defined in the case of non-full rank designs
(including the case $p_\MM > n$). Furthermore under normality,  Proposition
\ref{proposition:normal_marlik} shows that Bayes factors  for models
with $p_\MM > n$ compared to the null are equal to one. In these cases
the prior on the model space plays a critical role in model
averaging as well as stochastic search algorithms.  Prior
distributions, such as the sparsity priors of
\cite{yang:wainwright:jordan:2016} or truncated Poisson distributions,
that place zero prior probability on models with $p_\MM > n$, are
potentially useful in controlling the model size for models that are
not identifiable from the data.  However,
as MCMC algorithms may explore predominantly lower dimensional models
that are full rank, $g$-priors and mixtures may have a computational
advantage over the independent ``spike and slab'' priors, as
inversion of the full $p$ dimensional matrices can be avoided under
the $g$-priors.

Proposition \ref{proposition:normal_marlik} can be used to extend the
results of \citet{SabanesBove_etal_2015} who adopted the hyper-$g/n$
prior to generalized additive models using splines in exponential
families after marginalizing over the coefficients for the spline
basis.  This leads to a general linear model with a $\W$ that depends
on unknown variance components.  Rather than placing a prior directly on
the variance components, they exploited the one-to-one correspondence
between the degrees of freedom $d$ and the variance components and
considered a discrete set of values for  $d$ with the objective of
learning models indexed by $d$.  Proposition
\ref{proposition:normal_marlik} also suggests new directions for
BVS/BMA of linear predictors in spatial-temporal models or other
covariance structures that depend on a low dimensional number of
parameters.

\section*{Supplementary Materials}
\begin{description}[leftmargin=*, itemsep=0em]
\item[Appendix \ref{appendix:proofs}:] a list of assumptions, all the proofs, and some additional theoretical results.
\item[Appendix \ref{section:TBF}:] discussion and an empirical example on the test-based Bayes factor.
\item[Appendix \ref{section:poisson}:] a Poisson regression simulation example, and additional results on the logistic regression simulation example.
\item[CHIC\_examples.zip:] R scripts for producing the simulation and real data results.

\end{description}

 
\blind{}{
\section*{Acknowledgement}
The authors thank the Editor, Associate Editor and
the two reviewers for suggestions that led to a greatly improved paper.
The authors also thank Drs.~James O. Berger, D. Andrew Brown, Michael
B. Gordy, and Yuzo Maruyama for helpful
discussions.   This material is based on work supported, in part, by
the National Science Foundation under grant DMS--11060891 and the
National Institute of Health  grant 1--R21--ES020796-01. Any opinions, findings, and conclusions or recommendations expressed in this material 
are those of the author(s) and do not necessarily reflect the views of
the NSF or NIH.
}

{\small
\singlespacing
\bibliographystyle{asa}
\bibliography{GLMchg}
}

\newpage
\begin{center}
{\LARGE
\emph{
Supplementary Materials: Appendices
}
}
\end{center}

\renewcommand{\thesection}{\Alph{section}}
\setcounter{section}{0}

\section{Assumptions, Theoretical Results, and Proofs}\label{appendix:proofs}
\subsection{Assumptions and Regularity Conditions}\label{section:assumptions}


The following assumptions and standard regularity conditions 
are used throughout the paper unless specified otherwise.
\begin{description}

\item[For functions $b(\cdot)$ and $\theta(\cdot)$] in the GLM density \eqref{eq:likelihood},
 their third derivatives
exist and are continuous on $\mathbb{R}$. The composite function 
$b^{\prime} \circ \theta(\cdot)$, which links $\E(Y)$ and $\eta$, 
is strictly monotonic. 
The variance function $b^{\prime\prime} \circ \theta(\cdot) \geq 0$, 
and the equality can only occur on the boundary $\pm \infty$.

\item[Finite MLEs $\hat{\alpha}_\MM, \hat{\b}_\MM$] exist and are unique,
under all subset models $\MM$.

\item[The design matrix $\X$] under the full model is known and has a
  full column rank $p$. Here, $p$ is fixed. The column space $C(\X)$ does not contain
  $\mathbf{1}_n$.  When studying asymptotics, we assume that for $i = 1, \ldots, n$ the
  norm of the $i$th row $\| \mathbf{x}_i\|_2$ is bounded by a 
  constant, and for all $n$,
  the smallest eigenvalue of $\X^T \X / n$ is bounded from below by a
  positive constant.  These conditions assure weak consistency
  (convergence in probability) and asymptotic normality for MLEs
  \citep{Fahrmeir_Kaufmann_1985}.

\item[The true model $\MT$] 
is among the $2^p$ subset models to be selected under
consideration.
In $\MT$, true values of the intercept and
regression coefficients are denoted by $\alpha_{\MT}^*$ and $\b_{\MT}^*$, 
respectively.

\end{description}


\subsection{Proof of Proposition \ref{proposition:marlik_beta}}\label{PROOFproposition:marlik_beta}
\begin{proof}
We first approximate the likelihood by a second order Taylor expansion at the MLE,
\begin{align*}
&	p(\Y \mid \alpha, \b_\MM, \MM)\\
\approx~ & 	 
	p(\Y \mid \hat{\alpha}_\MM, \hat{\b}_\MM, \MM) \\
&	\cdot		\exp\left\{-\frac{1}{2}
	\left[ 
	\begin{array}{c}
	\alpha - \hat{\alpha}_\MM \\
	\b_\MM - \hat{\b}_\MM
	\end{array} \right]^T
	\left[ 
	\begin{array}{cc}
	\mathbf{1}_n^T \mathcal{J}_n(\hat{\boldsymbol\eta}_\MM) \mathbf{1}_n 
	& \mathbf{1}_n^T \mathcal{J}_n(\hat{\boldsymbol\eta}_\MM) \mathbf{X}_\MM\\
	\mathbf{X}_\MM^T \mathcal{J}_n(\hat{\boldsymbol\eta}_\MM) \mathbf{1}_n 
	& \mathbf{X}_\MM^T \mathcal{J}_n(\hat{\boldsymbol\eta}_\MM) \mathbf{X}_\MM
	\end{array} 
	\right]
	\left[ 
	\begin{array}{c}
	\alpha - \hat{\alpha}_\MM \\
	\b_\MM - \hat{\b}_\MM
	\end{array} \right]
	\right\}	\\
= ~&	
	p(\Y \mid \hat{\alpha}_\MM, \hat{\b}_\MM, \MM)
	 \exp\left\{-\frac{1}{2}
	\left( \alpha - \hat{\alpha}_\MM +\mathbf{m} \right)^T 
	\left( \mathbf{1}_n^T \mathcal{J}_n(\hat{\boldsymbol\eta}_\MM) \mathbf{1}_n \right)
	\left( \alpha - \hat{\alpha}_\MM +\mathbf{m}  \right) \right.\\
&	~~~~~~~~~~~~~~~~~~~~~~~~~~~~~~~~ \left.
	-\frac{1}{2}\left(\b_\MM - \hat{\b}_\MM\right)^T  
	\boldsymbol\Phi
	\left(\b_\MM - \hat{\b}_\MM\right)	 \right\},
\end{align*}
where the above approximation is precise up to a multiplicative term $\left[1 + O(n^{-1})\right]$,
$\mathbf{m} = \left(\mathbf{1}_n^T \mathcal{J}_n(\hat{\boldsymbol\eta}_\MM) \mathbf{1}_n \right)^{-1}
	\left(\mathbf{1}_n^T \mathcal{J}_n(\hat{\boldsymbol\eta}_\MM) \mathbf{X}_\MM \right)
	\left(\b_\MM - \hat{\b}_\MM\right)$, and
\[
\boldsymbol\Phi = \mathbf{X}_\MM^T \mathcal{J}_n(\hat{\boldsymbol\eta}_\MM) \mathbf{X}_\MM 
	- \left(\mathbf{X}_\MM^T \mathcal{J}_n(\hat{\boldsymbol\eta}_\MM)  \mathbf{1}_n \right)
	\left(\mathbf{1}_n^T \mathcal{J}_n(\hat{\boldsymbol\eta}_\MM) \mathbf{1}_n \right)^{-1}
	\left(\mathbf{1}_n^T \mathcal{J}_n(\hat{\boldsymbol\eta}_\MM) \mathbf{X}_\MM \right).
\]
In the above approximate likelihood, 
the matrix $\boldsymbol\Phi$ acts like a precision matrix of $\b_\MM$.
By using the orthogonal projection $\hat{\mathcal{P}}_{\mathbf{1}_n}
= \mathbf{1}_n \left(\mathbf{1}^T \mathcal{J}_n(\hat{\boldsymbol\eta}_\MM) \mathbf{1} \right)^{-1}
\mathbf{1}^T \mathcal{J}_n(\hat{\boldsymbol\eta}_\MM)$, 
we can rewrite it as
\begin{align*}
\boldsymbol\Phi 
& =\mathbf{X}_\MM^T \mathcal{J}_n(\hat{\boldsymbol\eta}_\MM) \mathbf{X}_\MM 
	- \mathbf{X}_\MM^T \hat{\mathcal{P}}_{\mathbf{1}_n}^T
	\mathcal{J}_n(\hat{\boldsymbol\eta}_\MM) \hat{\mathcal{P}}_{\mathbf{1}_n} \mathbf{X}_\MM \\
& = \mathbf{X}_\MM^T (\mathbf{I}_n -  \hat{\mathcal{P}}_{\mathbf{1}_n})^T
	\mathcal{J}_n(\hat{\boldsymbol\eta}_\MM) (\mathbf{I}_n -  \hat{\mathcal{P}}_{\mathbf{1}_n})\mathbf{X}_\MM
	= \mathcal{J}_n(\hat{\b}_\MM). 
\end{align*}
	
Under the flat prior $p(\alpha) \propto 1$, an integrated Laplace approximation yields
the marginal likelihood density conditional on $\b_\MM$:
\begin{align*}
&	p(\Y \mid \b_\MM, \MM) =
	\int p(\Y \mid \alpha, \b_\MM, \MM)p(\alpha) d\alpha\\
\propto 
~ &
	p(\Y \mid \hat{\alpha}_\MM, \hat{\b}_\MM, \MM)
	\exp\left\{ -\frac{1}{2}\left(\b_\MM - \hat{\b}_\MM\right)^T  
	\mathcal{J}_n(\hat{\b}_\MM) \left(\b_\MM - \hat{\b}_\MM\right)
	\right\}\\
&	\cdot \int \exp\left\{-\frac{1}{2}\left( \alpha - \hat{\alpha}_\MM +\mathbf{m} \right)^T 
	\left( \mathbf{1}_n^T \mathcal{J}_n(\hat{\boldsymbol\eta}_\MM) \mathbf{1}_n \right)
	\left( \alpha - \hat{\alpha}_\MM +\mathbf{m}  \right) \right\}  d\alpha\\
\propto~ &
	p(\Y \mid \hat{\alpha}_\MM, \hat{\b}_\MM, \MM)
	\left[ \mathbf{1}_n^T \mathcal{J}_n(\hat{\boldsymbol\eta}_\MM) \mathbf{1}_n \right]^{-\frac{1}{2}}
	\exp\left\{ -\frac{1}{2}\left(\b_\MM - \hat{\b}_\MM\right)^T  
	\mathcal{J}_n(\hat{\b}_\MM) \left(\b_\MM - \hat{\b}_\MM\right)
	\right\}.
\end{align*}
\end{proof}

\subsection{Asymptotic Behavior of the Observed Information}

\begin{lemma}\label{lemma:order_Jalpha_Jbeta}
For any subset model $\MM$, 
\begin{enumerate}[label = (\arabic*)]
\item if $\MM \supset \MT$, then
$\mathcal{J}_n(\hat{\alpha}_{\MM}) = O_P(n)$ and $\mathcal{J}_n(\hat{\b}_{\MM}) = O_P(n)$.
More specifically,
$\mathcal{J}_n(\hat{\alpha}_{\MM})/n -\mathcal{I}_n(\hat{\alpha}_{\MM})/n \stackrel{P}{\longrightarrow} 0$,
and
$\mathcal{J}_n(\hat{\b}_{\MM})/n - \mathcal{I}_n(\hat{\b}_{\MM})/n \stackrel{P}{\longrightarrow} \mathbf{0}$.
\item if $\MM \not\supset \MT$, then
$\mathcal{J}_n(\hat{\alpha}_{\MM}) = O_P(n^{\tau_\MM})$ and $\mathcal{J}_n(\hat{\b}_{\MM}) = O_P(n^{\tau_\MM})$,
where $0\leq \tau_\MM \leq 1$.
\end{enumerate}
\end{lemma}

\begin{proof} 
First, we study the asymptotic of MLEs.
The assumptions on the design matrix of the full model $\X$ remain to hold
for the design matrix $\X_\MM$ under all subset models, i.e., 
$\mathbf{x}_{\MM, i}$ are bounded for all $i = 1, \ldots, n$, 
and as $n$ tends to infinity, the smallest eigenvalue of $\X_\MM ^T \X_\MM / n$ 
is bounded from below by a positive constant.
Since these are stronger than the condition $R_c$ in \citet[pp.\ 355]{Fahrmeir_Kaufmann_1985},
we have weak consistency and asymptotic normality for
MLEs under any $\MM \supset \MT$, i.e., as $n \rightarrow \infty$,
\begin{equation}\label{eq:asym_MLE_MT}
\left(\hat{\alpha}_\MM, \hat{\b}_\MM\right) \stackrel{P}{\longrightarrow} \left(\alpha_\MM^*, \b_\MM^*\right), \quad
\mathcal{I}_n(\b_\MM^*)^{\frac{1}{2}}\left( \hat{\b}_\MM - \b_\MT^* \right)
\stackrel{d}{\longrightarrow} \text{N}(0, \mathbf{I}_{p_\MM}).
\end{equation}
Here, $\alpha_\MM^* = \alpha_\MT^*$, and $\b_\MM^* = \b_\MT^*$ in the
sense that all entries in $\b_\MM^*$ that correspond to predictors not
in $\MT$ are filled with zero.  Therefore, if $\MM \supset \MT$, then
$\eta_{\MM,i}^* = \alpha_\MM^* + \mathbf{x}_{\MM, i}^T \b_\MM^* =
\alpha_\MT^* + \mathbf{x}_{\MT, i}^T \b_\MT^*= \eta_{\MT,i}^*$, for
all $i = 1, \ldots, n$.  On the other hand, if $\MM \not\supset \MT$,
\citet{Self_Mauritsen_1988} and \citet[pp.\ 45, Theorem
5.7]{vanderVaart_2000} suggest that the limits of MLEs still exist,
i.e.,
$\left(\hat{\alpha}_\MM, \hat{\b}_\MM\right)
\stackrel{P}{\longrightarrow} \left(\alpha_\MM^*, \b_\MM^*\right)$,
but the linear predictors in the limit
$\eta_{\MM,i}^* \neq \eta_{\MT,i}^*$.

Under non-canonical links, observed information matrices are functions
of $\mathbf{Y}$, therefore we need a weak law of large numbers for
independently but non-identically distributed random variables.  In
\citet[pp.\ 205]{Resnick_1999}, by Theorem 7.2.1 and the proof of
special case (a), we have that for a sequence of independent random
variables $Y_1, \ldots, Y_n$, if their variances are bounded, then as
$n \rightarrow \infty$,
\begin{equation}\label{eq:weak_LLN}
\frac{1}{n}\sum_{i=1}^n Y_i - \frac{1}{n}\sum_{i=1}^n \mathbb{E}(Y_i)
\stackrel{P}{\longrightarrow} 0.
\end{equation}

Next we show asymptotic results for $\mathcal{J}_n(\hat{\alpha}_{\MM})$.
In \eqref{eq:obs_info}, for $i = 1, \ldots, n$, the $i$th diagonal entry of 
$\mathcal{J}_n(\hat{\boldsymbol\eta}_\MM)$ can be rewritten as
$d_i = b'' \circ\theta (\hat{\eta}_{\MM, i}) \left[ \theta' (\hat{\eta}_{\MM, i})  \right]^2
	+ \left[ b' \circ \theta (\hat{\eta}_{\MM, i}) - Y_i \right] \theta^{\prime\prime}(\hat{\eta}_{\MM, i})$.
Hence, for any model $\MM$, 
\begin{align}\nonumber
&	\frac{1}{n}\mathcal{J}_n(\hat{\alpha}_{\MM}) 
	= \frac{1}{n}\mathbf{1}_n \mathcal{J}(\hat{\boldsymbol\eta}_{\MM}) \mathbf{1}_n 
	= \frac{1}{n} \sum_{i=1}^n d_i\\ \nonumber
=~&	  \frac{1}{n}\left\{\sum_{i=1}^n b'' \circ\theta (\hat{\eta}_{\MM, i}) \left[ \theta' (\hat{\eta}_{\MM, i})  \right]^2
	+ \left[ b' \circ \theta (\hat{\eta}_{\MM, i}) - Y_i \right] \theta^{\prime\prime}(\hat{\eta}_{\MM, i})\right\}\\ \nonumber
\stackrel{P}{\longrightarrow}~	&	
	 \frac{1}{n}\left\{\sum_{i=1}^n b'' \circ\theta (\hat{\eta}_{\MM, i}) \left[ \theta' (\hat{\eta}_{\MM, i})  \right]^2
	+ \left[ b' \circ \theta (\hat{\eta}_{\MM, i}) - b' \circ \theta (\eta_{\MT, i}^*) \right] 
	\theta^{\prime\prime}(\hat{\eta}_{\MM, i})\right\}\\ \label{eq:Jeta_asymptotics1}
\stackrel{P}{\longrightarrow}~	&	
	 \frac{1}{n}\left\{\sum_{i=1}^n b'' \circ\theta (\eta_{\MM, i}^*) \left[ \theta' (\eta_{\MM, i}^*)  \right]^2
	+ \left[ b' \circ \theta (\eta_{\MM, i}^*) - b' \circ \theta (\eta_{\MT, i}^*) \right] 
	\theta^{\prime\prime}(\eta_{\MM, i}^*)\right\},
\end{align}
where the second last line is given by \eqref{eq:weak_LLN} and the 
fact $\mathbb{E}(Y_i) = b' \circ \theta (\eta_{\MT, i}^*)$, for all $i = 1, \ldots, n$, 
and the last line is given by the continuous mapping theorem.
Since for all $i = 1, \ldots, n$, $\mathbf{x}_i$ is bounded,  $\eta_{\MM, i}^*$ 
and $\eta_{\MT, i}^*$ are also bounded. 
For each term in the summation of \eqref{eq:Jeta_asymptotics1},
it is bounded due to the continuity assumptions on the third derivatives of $b(\cdot)$ and $\theta(\cdot)$.
Therefore, $\mathcal{J}_n(\hat{\alpha}_{\MM})/n$ is bounded in probability.

If $\MM \supset \MT$, \eqref{eq:Jeta_asymptotics1} becomes
\begin{equation}\label{eq:Jeta_asymptotics2}
\frac{1}{n} \mathcal{J}_n(\hat{\alpha}_{\MM}) \stackrel{P}{\longrightarrow}
\frac{1}{n} \sum_{i=1}^n b'' \circ\theta (\eta_{\MT, i}^*) \left[ \theta' (\eta_{\MT, i}^*)  \right]^2
= \frac{1}{n}\mathcal{I}_n(\alpha_{\MM}^*),
\end{equation}
which is also the limit of $\mathcal{I}_n(\hat{\alpha}_{\MM}) / n$.
Because we assume that $b' \circ \theta(\cdot)$ is strictly monotonic,
 $\theta(\cdot)$ is also strictly monotonic. 
For each term in the summation of \eqref{eq:Jeta_asymptotics2},
it is positive because $\theta'(\cdot) \neq 0$ and
$b'' \circ\theta (\eta)$ is positive for finite $\eta$. 
Therefore by  \eqref{eq:Jeta_asymptotics2}, if $\MM \supset \MT$, then
$\mathcal{J}_n(\hat{\alpha}_{\MM})/n$ is positive and bounded in probability, 
i.e., $\mathcal{J}_n(\hat{\alpha}_{\MM}) = O_P(n)$.
On the other hand, if $\MM \not\supset \MT$, then only \eqref{eq:Jeta_asymptotics1} holds
but not \eqref{eq:Jeta_asymptotics2}. Each term
in the summation of \eqref{eq:Jeta_asymptotics1} can be
either positive, zero, or negative. In this case, by \eqref{eq:Jeta_asymptotics1}, 
$\mathcal{J}_n(\alpha_{\MM})/n$ is bounded in probability, and it may equal to zero. 
Therefore, $\mathcal{J}_n(\hat{\alpha}_{\MM})$ is on the order of $O(n^{\tau_\MM})$, where $\tau_n \leq 1$,
so that it tends to $\infty$ at a rate no faster than $O_P(n)$.

Last, we show asymptotic results regarding the matrix 
\begin{align*}
\mathcal{J}_n(\hat{\b}_{\MM}) 
&	= \mathbf{\X}_\MM^{cT} \mathcal{J}(\hat{\boldsymbol\eta}_{\MM}) \mathbf{\X}_\MM^c
	= \mathbf{X}_\MM^T(\mathbf{I}_n - \hat{\mathcal{P}}_{\mathbf{1}_n})^T
	\mathcal{J}(\hat{\boldsymbol\eta}_{\MM}) 
	(\mathbf{I}_n - \hat{\mathcal{P}}_{\mathbf{1}_n})\mathbf{X}_\MM \\
& 	= \mathbf{\X}_\MM^{T}
	\left[ \mathcal{J}(\hat{\boldsymbol\eta}_{\MM}) - 
	\mathcal{J}(\hat{\boldsymbol\eta}_{\MM})
	\mathbf{1}_n \left( \mathbf{1}_n^T \mathcal{J}_n(\hat{\boldsymbol\eta}_{\MM}) \mathbf{1}_n\right)^{-1}
  	\mathbf{1}_n^T \mathcal{J}(\hat{\boldsymbol\eta}_{\MM}) \right]
	\mathbf{\X}_\MM.
\end{align*}
For the $(j,k)$th entry, $1\leq j < k \leq p_\MM$, 
\begin{align*}\nonumber
	\frac{1}{n}\left[\mathcal{J}_n(\hat{\b}_{\MM}) \right]_{j,k}
&	= \frac{1}{n} \sum_{i=1}^n d_i x_{i,j} x_{i,k} 
	- \frac{1}{n}\left(\sum_{i=1}^n d_i x_{i, j}\right) \left(\sum_{i=1}^n d_i \right)^{-1} 
	\left(\sum_{i=1}^n d_i x_{i, k}\right)\\
&	= \frac{1}{n} \sum_{i=1}^n d_i x_{i,j} x_{i,k} 
	- \left(\frac{1}{n}\sum_{i=1}^n d_i x_{i, j}\right) \left(\frac{1}{n}\sum_{i=1}^n d_i \right)^{-1} 
	\left(\frac{1}{n}\sum_{i=1}^n d_i x_{i, k}\right)
\end{align*}
is bounded since all $\mathbf{x}_i$ are bounded.
Therefore, $\mathcal{J}_n(\hat{\b}_{\MM}) / n$ is bounded in probability.

To show that for any $\MM \supset \MT$,
$\mathcal{J}_n(\hat{\b}_{\MM}) / n$ does not reduce to zero, we will
show that it is a positive definite matrix.  For any given non-zero
vector $\mathbf{a} \in \mathbb{R}^{p_\MM}$, we denote
$\X_\MM \mathbf{a} = (t_1, \ldots, t_n)^T$, whose entries are all
bounded.  When $\MM \supset \MT$, by \eqref{eq:Jeta_asymptotics2}, all
$d_i$'s have a positive lower bound, hence simple calculation gives
\[
\frac{1}{n}\mathbf{a}^T \mathcal{J}_n(\hat{\b}_{\MM}) \mathbf{a}
	= \frac{1}{n} \sum_{i=1}^n d_i t_i^2 
	- \left(\frac{1}{n}\sum_{i=1}^n d_i \right)^{-1} 
	\left(\frac{1}{n}\sum_{i=1}^n d_i t_i\right)^2 \geq 0.
\]
Here the quality only holds if all $t_i$'s are equal for $i = 1, \ldots, n$, 
which is impossible here because of the assumption $\mathbf{1}_n \not\in C(\X_\MM)$.
For large $n$, the assumption that the smallest eigenvalue of 
$\X^T \X / n$ being bounded from below by a positive constant suggests that
$\X_\MM^T \X_\MM / n$ is positive definite, so 
$\mathbf{a}^T \mathcal{J}_n(\hat{\b}_{\MM}) \mathbf{a} /n \not\longrightarrow 0$.

Furthermore, arguing similarly to \eqref{eq:Jeta_asymptotics2}, we also have
\[
\frac{1}{n} \sum_{i=1}^n d_i t_i^k  \stackrel{P}{\longrightarrow}
\frac{1}{n} \sum_{i=1}^n b'' \circ\theta (\eta_{\MT, i}^*) \left[ \theta' (\eta_{\MT, i}^*)  \right]^2 t_i^k,
\]
for $k = 0, 1, 2$.
Therefore, for any vector $\mathbf{a}$, if $\MM \supset \MT$, then 
\[
\frac{1}{n}\mathbf{a}^T \mathcal{J}_n(\hat{\b}_{\MM}) \mathbf{a} - 
\frac{1}{n}\mathbf{a}^T \mathcal{I}_n(\hat{\b}_{\MM}) \mathbf{a} 
\stackrel{P}{\longrightarrow} 0.
\]
i.e., $\mathcal{J}_n(\hat{\b}_{\MM})/n$ and $\mathcal{I}_n(\hat{\b}_{\MM})/n$
are asymptotically the same.
\end{proof}

\subsection{Proof of Proposition \ref{proposition:non_full_rank_BF}}\label{PROOFproposition:non_full_rank_BF}

\begin{proof}
We first use proof by contradiction to show that for $\MM$, the MLE of the intercept is unique. 
If both $(\hat{\alpha}_1, \hat{\b}_1)$ and $(\hat{\alpha}_2, \hat{\b}_2)$
maximize the likelihood for model $\MM$, where $\hat{\alpha}_1 \neq \hat{\alpha}_2$, then
\[
\hat{\alpha}_1 \mathbf{1}_n + \X_\MM \hat{\b}_1
= \hat{\alpha}_2 \mathbf{1}_n + \X_\MM \hat{\b}_2
\Longrightarrow
(\hat{\alpha}_1 - \hat{\alpha}_2) \mathbf{1}_n =  \X_\MM (\hat{\b}_2 - \hat{\b}_1),
\]
which is contradicted with $\mathbf{1}_n \not\in C(\X_\MM)$.
Similarly, we can show this MLE is the same as the one for model $\MM'$, i.e.,
$\hat{\alpha}_\MM = \hat{\alpha}_{\MM'}$.

By \eqref{eq:Jeta_non-full_rank}, between the two models $\MM$ and $\MM'$,
\[
\mathcal{J}_n(\hat{\alpha}_\MM) = \mathcal{J}_n(\hat{\alpha}_{\MM'}), \quad
z_\MM = z_{\MM'}.
\]
So we just need to show $Q_\MM = Q_{\MM'}$. Since $\hat{\alpha}_\MM = \hat{\alpha}_{\MM'}$, 
 \eqref{eq:Jeta_non-full_rank} suggests that 
\[
\mathbf{x}_{\MM, i}^T \hat{\b}_\MM
= \mathbf{x}_{\MM', i}^T \hat{\b}_{\MM'},
\quad i = 1, \ldots, n.
\]
Hence,
\begin{align*}
\mathbf{X}_\MM^c  \hat{\boldsymbol\beta}_{\MM} 
&=	\mathbf{X}_\MM  \hat{\boldsymbol\beta}_{\MM} - 
	\left(\sum_{i=1}^n w_i \mathbf{x}_{\MM, i}^T \hat{\b}_\MM \right) \mathbf{1}_n\\
&=	\mathbf{X}_{\MM'}  \hat{\boldsymbol\beta}_{\MM'} - 
	\left(\sum_{i=1}^n w_i \mathbf{x}_{\MM', i}^T \hat{\b}_{\MM'} \right) \mathbf{1}_n	
	= \mathbf{X}_{\MM'}^c  \hat{\boldsymbol\beta}_{\MM'},
\end{align*}
where $w_i = d_i / (\sum_{r=1}^n d_r)$. Therefore, we have
\[
Q_{\MM} 
	 = \left[ \mathbf{X}_\MM^c \hat{\boldsymbol\beta}_{\MM} \right]^T 
	\mathcal{J}_n( \hat{\boldsymbol\eta}_{\MM}) 
	\left[ \mathbf{X}_\MM^c  \hat{\boldsymbol\beta}_{\MM}\right]
	= \left[ \mathbf{X}_{\MM'}^c \hat{\boldsymbol\beta}_{\MM'} \right]^T 
	\mathcal{J}_n( \hat{\boldsymbol\eta}_{\MM'}) 
	\left[ \mathbf{X}_{\MM'}^c  \hat{\boldsymbol\beta}_{\MM'}\right]
	= Q_{\MM'}.
\]

\end{proof}

\subsection{Proof of Proposition \ref{proposition:tCCH_marlik}}\label{PROOFproposition:tCCH_marlik}
\begin{proof}
The marginal likelihood of the mixture of $g$-priors is obtained by integrating out 
$g$ from the marginal likelihood of the $g$-prior, i.e.,
\[
p(\mathbf{Y} \mid \MM)  
= \int_0^{\infty} p(\mathbf{Y} \mid \MM,g)p(g) dg 
\]
Here $p(\mathbf{Y} \mid \MM,g)$ is obtained under the integrated Laplace approximation 
as in \eqref{eq:marlik_fixedg}. Because of
the one-to-one mapping between $g$ and $u$, we rewrite this integral in terms of $u$.
\begin{align*}
p(\mathbf{Y} \mid \MM)  = &\int_0^{1} p(\mathbf{Y} \mid \MM,u)p(u) du\\
\propto 
&  \int_0^{1}
  p(\mathbf{Y} \mid \hat{\alpha}_\MM, \hat{\boldsymbol\beta}_\MM, \MM)
  \mathcal{J}_n(\hat{\alpha}_\MM)^{-\frac{1}{2}}
    u^{\frac{p_{\MM}}{2}} e^{ -\frac{Q_{\MM}}{2} u}\\
&\cdot
  \frac{v^\frac{a}{2} \exp\left(\frac{s}{2v}\right)}{B\left(\frac{a}{2}, \frac{b}{2}\right)\ 
  \Phi_1\left(\frac{b}{2}, r, \frac{a+b}{2}, \frac{s}{2v}, 1-\kappa\right)}\
  \frac{u^{\frac{a}{2}-1}(1-vu)^{\frac{b}{2}-1}e^{-\frac{s}{2}u}}{\left[ \kappa + (1-\kappa)vu \right]^r}\ 
  \mathbf{1}_{\{0 < u< \frac{1}{v}  \}}~du\\
=  &  
  ~ p(\mathbf{Y} \mid \hat{\alpha}_\MM, \hat{\boldsymbol\beta}_\MM, \MM)
  \mathcal{J}_n(\hat{\alpha}_\MM)^{-\frac{1}{2}}
  \frac{v^\frac{a}{2} \exp\left(\frac{s}{2v}\right)}{B\left(\frac{a}{2}, \frac{b}{2}\right)\ 
  \Phi_1\left(\frac{b}{2}, r, \frac{a+b}{2}, \frac{s}{2v}, 1-\kappa\right)}\\
&\cdot  \int_0^{1} \frac{u^{\frac{a+ p_{\MM}}{2}-1}(1-vu)^{\frac{b}{2}-1}
  e^{-\frac{s+Q_{\MM}}{2}u}}{\left[ \kappa + (1-\kappa)vu \right]^r}\ 
  \mathbf{1}_{\{0 < u< \frac{1}{v}  \}}~du.
\end{align*}
Since the above integrand is proportional to a tCCH density \eqref{eq:CCH_dist}
with updated parameters, the above integral equals $B\left(\frac{a+ p_\MM}{2}, \frac{b}{2}\right)\ 
  \Phi_1\left(\frac{b}{2}, r, \frac{a+b+ p_\MM}{2}, \frac{s+ Q_\MM}{2v}, 1-\kappa\right)
  v^{-\frac{a + p_\MM}{2}} \exp\left(-\frac{s + Q_\MM}{2v}\right)$. 
\end{proof}

\subsection{Proof of Proposition \ref{proposition:CH-g_tails}}\label{PROOFproposition:CH-g_tails}

\begin{proof}
The marginal prior on $\boldsymbol\beta_\MM$ after integrating $g$ out is
\begin{equation}\label{eq:marginal_prior}
p(\boldsymbol\beta_\MM \mid \MM) \propto \int_0^{\infty} g^{-\frac{p_\MM}{2}} \exp
\left[-\frac{\| \boldsymbol\beta_\MM \|_{\mathcal{J}_n}^2 }{2g} \right]
g^{\frac{b}{2}-1}\left( \frac{1}{1+g} \right)^{\frac{a+b}{2}} \exp\left[ \frac{sg}{2(1+g)} \right] dg
\end{equation}
We will show that  as  $\| \boldsymbol\beta_\MM \|_{\mathcal{J}_n} \rightarrow \infty$,
both a lower bound and an upper bound of \eqref{eq:marginal_prior} are proportional to
$\left( \| \boldsymbol\beta_\MM \|_{\mathcal{J}_n}^2 \right) ^{-\frac{a + p_\MM}{2}}$.
Since $s\geq 0$, a lower bound of of the right side of \eqref{eq:marginal_prior} is
\[
  \int_0^{\infty} g^{-\frac{p_\MM}{2}} e^{-\frac{\| \boldsymbol\beta_\MM \|_{\mathcal{J}_n}^2 }{2g} }
  g^{\frac{b}{2}-1}\left( \frac{1}{1+g} \right)^{\frac{a+b}{2}} dg
= \int_0^{\infty} \left(\frac{g}{1+g}\right)^{\frac{a+b}{2}} \left(\frac{1}{g}\right)^{\frac{a+p_\MM -2}{2}}
  e^{-\frac{\| \boldsymbol\beta_\MM \|_{\mathcal{J}_n}^2 }{2g} } d\left( \frac{1}{g}\right).
\]
Then according to the Watson's Lemma \citep[pp.\ 71]{Olver_1997}, 
as $\| \boldsymbol\beta_\MM \|_{\mathcal{J}_n} \rightarrow \infty$, the limit of this lower bound is
proportional to $\left( \| \boldsymbol\beta_\MM \|_{\mathcal{J}_n}^2 \right) ^{-\frac{a + p_\MM}{2}}$.
Next we find an upper bound of the right side of \eqref{eq:marginal_prior} as 
\begin{align*}
& \int_0^{\infty} g^{-\frac{p_\MM}{2}} \exp \left[-\frac{\| \boldsymbol\beta_\MM \|_{\mathcal{J}_n}^2 }
  											  {2(1+g)} \right]
  g^{\frac{b}{2}-1}\left( \frac{1}{1+g} \right)^{\frac{a+b}{2}} \exp\left[  \frac{sg}{2(1+g)}\right] dg\\
 = & ~ e^{-\frac{\| \boldsymbol\beta_\MM \|_{\mathcal{J}_n}^2}{2} }
  B\left(\frac{b - p_\MM}{2}, \frac{a+p_\MM}{2} \right) 
  \ _{1}F_1\left( \frac{b-p_\MM}{2}, \frac{a+b}{2},
  \frac{s+ \| \boldsymbol\beta_\MM \|_{\mathcal{J}_n}^2}{2} \right).
\end{align*}
According to \citet{Abramowitz_Stegun_1970} formula (13.1.4),
\begin{equation}\label{eq:1F1_limit_positive_s}
 _{1} F_1 (a, b, s) = \frac{\Gamma(b)}{\Gamma(a)} \exp(s) s^{a-b} [1 + O(|s|^{-1})], \textrm{ when } \text{Real}(s)> 0,
\end{equation}
hence as $\| \boldsymbol\beta_\MM \|_{\mathcal{J}_n} \rightarrow \infty$, the limit of the above upper bound
converges to
\[
\exp\left[ -\frac{\| \boldsymbol\beta_\MM \|_{\mathcal{J}_n}^2}{2} \right] \Gamma\left( \frac{a+p_\MM}{2} \right)
  \exp\left[\frac{s+ \| \boldsymbol\beta_\MM \|_{\mathcal{J}_n}^2}{2} \right]
  \cdot \left(\frac{s+ \| \boldsymbol\beta_\MM \|_{\mathcal{J}_n}^2}{2} \right)^{-\frac{a+p_\MM}{2}} 
\propto  \left( \| \boldsymbol\beta_\MM \|_{\mathcal{J}_n}^2 \right) ^{-\frac{a + p_\MM}{2}}.
\]
Therefore, as $\| \boldsymbol\beta_\MM \|_{\mathcal{J}_n}$ increases, 
or equivalently, as $\| \boldsymbol\beta_\MM \|$ increases, both the lower bound and upper bound of
$p(\boldsymbol\beta_\MM \mid \MM)$ are proportional to 
$\left( \| \boldsymbol\beta_\MM \|_{\mathcal{J}_n}^2 \right) ^{-\frac{a + p_\MM}{2}}$.
\end{proof}

\subsection{Special Functions: Definition and Useful Properties} \label{subsection:special_functions}

We first review a list of special functions, including their definitions and relevant properties,
that will be needed in the proof of Proposition \ref{proposition:normal_marlik}.
\begin{itemize}[leftmargin=*]
\item Confluent hypergeometric function \citep[eq 13.2.1]{Abramowitz_Stegun_1970}: for $\gamma > \alpha > 0$,
\[
_{1}F_{1}(\alpha, \gamma, x) 
=	\frac{1}{\text{B}(\gamma - \alpha, \alpha)}\int_0^1 u^{\alpha-1} 
	(1-u)^{\gamma - \alpha - 1} e^{xu} ~du.
 \]
  \begin{itemize}
  \item By \citep[eq 13.2.27]{Abramowitz_Stegun_1970}: $_{1}F_{1}(\alpha, \gamma, x) = e^x \cdot {_1}F_{1}(\gamma- \alpha, \gamma, -x)$.
  \item By \citep[eq 6.5.12]{Abramowitz_Stegun_1970}, the incomplete Gamma function: 
  \[\gamma(a, s) = \int_0^{s} t^{a-1} e^{-t}dt =  {_1}F_{1}(a, a+1, -s) \frac{s^a}{a}.\]
  \item ${_1}F_{1}(\alpha, \gamma, 0) = 1$.
  \end{itemize}

\item Confluent hypergeometric function of two variables \citep{Gordy_1998a}\footnote
{{Note: the definition in \citet{Gordy_1998a} is slightly different from that in 
\citet{Gradshteyn_Ryzhik_2007}.
}}: for $\gamma > \alpha > 0$ and $y < 1$,
\[
\Phi_1(\alpha, \beta,\gamma, x, y) 
= 	\frac{1}{\text{B}(\gamma - \alpha, \alpha)}
	\int_0^1 u^{\alpha -1} (1-u)^{\gamma - \alpha - 1} (1 - yu)^{-\beta} e^{xu} ~du,
\]
Special cases:
  \begin{itemize}
  \item If $x = 0$, then $\Phi_1(\alpha, \beta,\gamma, 0, y) =~ _{2}F_1(\beta, \alpha; \gamma; y)$.
  \item If $\beta = 0$ or $y = 0$, then $\Phi_1(\alpha, 0,\gamma, x, y) 
  = \Phi_1(\alpha, \beta,\gamma, x, 0) = \Phi_1(\alpha, 0,\gamma, x, 0) =~ _{1}F_1(\alpha, \gamma, x)$.
  \item If $x = 0$ and $y=0$, then $\Phi_1(\alpha, \beta,\gamma, 0, 0) = 1$.
  \end{itemize}

\item Hypergeometric function \citep[eq 15.3.1]{Abramowitz_Stegun_1970}: for $\gamma > \alpha > 0$
\[
_{2}F_1(\beta, \alpha; \gamma; x) = \frac{1}{\text{B}(\gamma - \alpha, \alpha)}
	\int_0^1 u^{\alpha -1} (1-u)^{\gamma - \alpha - 1} (1 - xu)^{-\beta}~ du.
\]
\begin{itemize}
\item By \citep[eq 15.3.3]{Abramowitz_Stegun_1970}: in the definition
of $_{2} F_1$ function above, let $w = \frac{1-u}{1-xu}$, then
\begin{equation*}\label{eq:2F1_trans1}
_{2}F_1(\beta, \alpha; \gamma; x) = (1-x)^{\gamma - \beta - \alpha}~ _{2}F_1(\gamma-\beta, \gamma -\alpha; \gamma; x)
\end{equation*}
\item ${_2}F_1(0, \alpha; \gamma, x) = {_2}F_1(\beta, \alpha; \gamma, 0) = 1$
\item ${_2}F_1(\beta, 1; \beta, x) = (1-x)^{-1}{_2}F_1(0, \beta - 1; \beta, x) = (1-x)^{-1}$
\item By \citep[eq 15.3.4]{Abramowitz_Stegun_1970}: $_{2}F_1(\beta, \alpha; \gamma; x) = (1-x)^{-\beta} {_2}F_1\left(\beta, \gamma - \alpha; \gamma, \frac{x}{x-1}\right)$
\item By \citep[eq 15.3.5]{Abramowitz_Stegun_1970}: ${_2}F_1(\beta, \alpha; \gamma; x) = (1-x)^{-\alpha} {_2}F_1\left(\alpha, \gamma - \beta; \gamma, \frac{x}{x-1}\right)$
\end{itemize}

\item Hypergeometric function of two variables (Appell function) \citep{Weisstein_2009}: for $\gamma > \alpha > 0$,
\[
F_1(\alpha; \beta, \beta'; \gamma; x, y) = \frac{1}{\text{B}(\gamma - \alpha, \alpha)}
	\int_0^1 u^{\alpha -1} (1-u)^{\gamma - \alpha - 1} (1 - xu)^{-\beta}(1 - yu)^{-\beta'} ~ du.
\]

\end{itemize}

\subsection{Proof of Proposition \ref{proposition:normal_marlik}} \label{PROOFproposition:normal_marlik}

\begin{proof}
To begin we establish that the marginal likelihood conditional on $g$
is well defined under the $g$-prior when the design matrix is not full
rank for a general linear model.  We will 
assume the inner product space defined by the vector space $\bbR^n$
equipped with inner product $\uv^T\W\vv$ for two vectors
$\uv, \vv \in \bbR^n$ where $\W$ is a real, $n \times n$ symmetric
positive definite matrix. Similarly, $\|\uv\|^2_\W \equiv \uv^T\W \uv$.

For the model
  $$\Y = \one_n \beta_0 + \X_\MM\b_\MM + \e, \quad \text{with } 
   \e \mid \phi \sim \N(\zero_n, \phi^{-1}\W^{-1}),$$
  let $\P_\one = \one_n(\one_n^T\W \one_n)^{-1} \one_n^T\W$ denote the
  orthogonal projection onto the column space of  $\one_n$
  and without loss of generality
  reparameterize the model
  $$\Y = \one_n \alpha + \X^c_\MM\b_\MM + \e $$
  where $\X^c_\MM = (\I_n - \P_{\one_n}) \X_\MM$ and $\alpha \equiv \beta_0 -
  (\one_n^T\W \one_n)^{-1} \one_n^T\W \X_\MM \b_\MM$.   Adopting the $g$-prior of the form
  $$\b_\MM \mid \alpha,\phi, g \sim \N\left(0, \frac{g}{\phi}(\X^{cT}_\MM\W
  \X_\MM^c)^{-}\right),$$ where $(\X_\MM^{cT}\W
  \X_\MM^c)^{-}$ is any generalized inverse, standard normal theory for the 
 linear  combination   $\X_\MM^c\b_\MM + \e$ can be used to show that
 $\Y$ is equal in distribution 
\begin{equation} \label{eq:marg-GLM-ginv}
  \Y \mid \alpha, \phi, g, \MM \sim \N\left(\one_n \alpha, \phi^{-1}
  (\I_n + g \P_{\X_\MM^c}) \W^{-1}\right)
 \end{equation}
  where $\P_{\X_\MM^c} = \X_\MM^c(\X_\MM^{cT}\W  \X_\MM^c)^{-}\X_\MM^{cT} \W$ is the
  $\rank_\MM \le p_\MM$
  orthogonal projection onto the column space $\X_\MM^c$ in the inner
  product space.  
  As the projection $\P_{\X^c_\MM}$ does not depend on
  the choice of 
  generalized inverse, this establishes that the marginal likelihood
  for the model will not depend on the choice of generalized inverse
  employed in defining the $g$-prior.

  Continuing with integration with respect to $\alpha, \phi$ under the independent
  Jeffreys prior $p(\alpha, \phi) \propto \phi^{-1}$,
  \begin{equation} \label{eq:marg-general-linear-model}
    p(\Y \mid g, \MM)  = \iint (2 \pi)^{-\frac{n}{2}} |\I_n + g
    \P_{\X^c_\MM}|^{-\frac{1}{2}} |\W|^{\frac{1}{2}}\phi^{\frac{n}{2} -1} e^{-\frac{\phi}{2}\left\{
    ( \Y - \one_n\alpha)^T\W(\I_n - \frac{g}{1 + g} \P_{\X^c_\MM}) (\Y - \one_n\alpha)
                        \right\} } d\alpha \, d\phi
\end{equation}                        
rearrangement of terms can be used to show that
\[
  p(\Y \mid g, \MM)   =  p(\Y \mid \Mnull) (1 + g)^{\frac{n-\rank_\MM-1}{2}}
 \left\{
        1 + g (1 - R^2_\MM)
                     \right\}^{-\frac{n-1}{2}}  
\]
where $R^2_\MM$ is defined in \eqref{eq:R2-GLM} and 
\[
p(\Y \mid \Mnull)  = (2 \pi)^{-\frac{n-1}{2}} \Gamma\left(
                   \frac{n-1}{2} \right)
                   |\W|^{\frac{1}{2}} (\one_n^T\W \one_n)^{-\frac{1}{2}} 
                   \left[ \frac{\|( \I_n -
                   \P_{\one_n}) \Y \|^2_\W}{2}\right]^{-\frac{n-1}{2}}
\]
is the marginal under the null model.  
Note that in \eqref{eq:marg-general-linear-model}, the determinant
$|\I_n + g\P_{\X^c_\MM}| = (1 + g)^{\rank_\MM}$, because the eigenvalues
of the orthogonal projection $\P_{\X^c_\MM}$ are one with a multiplicity of $\rank_\MM$ and 
zero with a multiplicity of $p_\MM - \rank_\MM$.
The Bayes Factor for comparing $\MM$ to $\Mnull$ is thus
$$
\BF[\MM, \Mnull] = (1 + g)^{\frac{n-\rank_\MM-1}{2}}
 \left\{
        1 + g (1 - R^2_\MM)
                     \right\}^{-\frac{n-1}{2}},    
$$
which will be one for any model $\MM$ where $R^2_\MM = 1$ and
$\rank_\MM= n-1$.

For simplicity in the rest of  proof,  we omit the subscript  $\MM$
when there is no ambiguity. 
We now show part (1). 
In the tCCH distribution, if $r=0$ or $\kappa = 1$, then 
\[
\Phi_1\left(\frac{b}{2}, r,  \frac{a+b}{2}, \frac{s}{2v}, 1 - \kappa\right) 
= \Phi_1\left(\frac{b}{2}, 0, \frac{a+b}{2}, \frac{s}{2v}, 0\right) 
= \ _{1}F_1\left(\frac{b}{2}, \frac{a+b}{2}, \frac{s}{2v} \right).
\]
Then the marginal likelihood becomes
\begin{align*}
&p(\mathbf{Y} \mid \MM) 
	= \frac{p(\mathbf{Y} \mid \Mnull)~ v^{\frac{a}{2}} \exp\left( \frac{s}{2v} \right)}
	{B\left(\frac{a}{2}, \frac{b}{2}\right)
	\ _{1}F_1\left(\frac{b}{2}, \frac{a+b}{2}, \frac{s}{2v} \right)} 
	\int_0^{1/v} 
	\frac{u^{\frac{a + \rank}{2}-1}(1-vu)^{\frac{b}{2}-1}e^{-\frac{su}{2}}}
	{\left[(1-R^2) + R^2 u\right]^{\frac{n-1}{2}}} ~du \\
=~& \frac{p(\mathbf{Y} \mid \Mnull) ~ v^{\frac{a}{2}} \exp\left( \frac{s}{2v} \right)}
	{B\left(\frac{a}{2}, \frac{b}{2}\right)
	\ _{1}F_1\left(\frac{b}{2}, \frac{a+b}{2}, \frac{s}{2v} \right)} 
	\int_0^{1/v} 
	\frac{u^{\frac{a + \rank}{2}-1}(1-vu)^{\frac{b}{2}-1}e^{-\frac{su}{2}}}
	{\left\{ \left[ 1 -\left( 1 - \frac{1}{v}\right) R^2 \right]
	\left[ \frac{1-R^2}{1 -\left( 1 - \frac{1}{v}\right) R^2} + \frac{R^2/v}
	{1 -\left( 1 - \frac{1}{v}\right) R^2} \cdot (vu)
	 \right]\right\}^{\frac{n-1}{2}}} ~du \\
=~& \frac{p(\mathbf{Y} \mid \Mnull) ~ v^{\frac{a}{2}} \exp\left( \frac{s}{2v} \right)}
	{B\left(\frac{a}{2}, \frac{b}{2}\right)
	\ _{1}F_1\left(\frac{b}{2}, \frac{a+b}{2}, \frac{s}{2v} \right)} 
	\cdot
	\frac{\text{B}\left(\frac{a + \rank}{2}, \frac{b}{2}\right)\ 
	\Phi_1\left(\frac{b}{2}, \frac{n-1}{2}, \frac{a + b + \rank}{2}, \frac{s}{2v}, 
	\frac{R^2/v}{1 -\left( 1 - \frac{1}{v}\right) R^2}\right)}
	{v^{\frac{a + \rank}{2}} \exp\left( \frac{s}{2v} \right)
	\left[ 1 -\left( 1 - \frac{1}{v}\right) R^2 \right]^{\frac{n-1}{2}}} \\
=~& p(\mathbf{Y} \mid \Mnull) \cdot 
	\frac{B\left(\frac{a + \rank}{2}, \frac{b}{2}\right)\ 
	\Phi_1\left(\frac{b}{2}, \frac{n-1}{2}, \frac{a + b + \rank}{2}, \frac{s}{2v}, 
	\frac{R^2/v}{1 -\left( 1 - \frac{1}{v}\right) R^2}\right)}
	{v^{\frac{\rank}{2}} \left[ 1 -\left( 1 - \frac{1}{v}\right) R^2 \right]^{\frac{n-1}{2}}
	B\left(\frac{a}{2}, \frac{b}{2}\right)
	\ _{1}F_1\left(\frac{b}{2}, \frac{a+b}{2}, \frac{s}{2v} \right)}.
\end{align*}
Here the second last equality is given by the propriety of the tCCH density function \eqref{eq:CCH_dist}.

Then we show part (2). In the tCCH distribution, when $s=0$, then
\[
\Phi_1\left(\frac{b}{2}, r, \frac{a + b}{2}, 0, 1 - \kappa\right) 
= \ _{2}F_1\left(r, \frac{b}{2}; \frac{a + b}{2}; 1-\kappa\right).
\]
Hence, the marginal likelihood becomes
\begin{equation}\label{eq:normal_marlik_robust2}
p(\mathbf{Y} \mid \MM) 
	= \frac{p(\mathbf{Y} \mid \Mnull)~ v^{\frac{a}{2}} }
	{B\left(\frac{a}{2}, \frac{b}{2}\right)\ 
	_{2}F_1\left(r, \frac{b}{2}; \frac{a + b}{2}; 1-\kappa\right)} 
	\int_0^{1/v} \frac{u^{\frac{a + \rank}{2}-1}(1-vu)^{\frac{b}{2}-1}}
	{\left[(1-R^2) + R^2 u\right]^{\frac{n-1}{2}}\left[ \kappa +
            (1-\kappa) v u \right]^r} ~du 
\end{equation}
For simplification, we denote $x = 1 - 1/\kappa$
and $w = 1 - (1 - vu)/(1 - xvu)$. 
By change of variable,
\[
 u = \frac{w}{v(1-x+xw)}, \quad
\frac{du}{dw} = \frac{1-x}{v (1-x+xw)^2},
\] 
and the integral in \eqref{eq:normal_marlik_robust2} is
\begin{align*}
&	\int_0^{1/v} \frac{u^{\frac{a + \rank}{2}-1}(1-vu)^{\frac{b}{2}-1}}
	{\left[(1-R^2) + R^2 u\right]^{\frac{n-1}{2}}\left[ \kappa + (1-\kappa)vu \right]^r} ~du \\
=~& \int_0^{1} \frac{\left[ \frac{w}{v(1-x+xw)} \right]^{\frac{a + \rank}{2}-1}
	\left[ \frac{(1-x)(1-w)}{1-x+xw} \right]^{\frac{b}{2}-1}\frac{1-x}{v (1-x+xw)^2}}
	{\left\{\frac{(1-R^2)v(1-x) + [(1-R^2)vx + R^2]w}{v(1-x+xw)}\right\}^{\frac{n-1}{2}}
	\left( \frac{1}{1-x+xw} \right)^r} ~dw \\
=~& \frac{(1-x)^\frac{b}{2} ~ v^{\frac{n-1-a -\rank}{2}}}
	{\left[ (1-R^2)v(1-x) \right]^{\frac{n-1}{2}} ~ (1-x)^{\frac{a + b + \rank + 1 - n-2r}{2}}}
	\int_0^{1} \frac{w^{\frac{a + \rank}{2}-1} (1-w)^{\frac{b}{2}-1}}
	{\left[ 1 - \frac{(1-R^2)vx + R^2}{(1-R^2)v(x-1)}w \right]^{\frac{n-1}{2}} 
	\left(1 - \frac{x}{x-1}w \right)^{\frac{a + b + \rank + 1 - n-2r}{2}} }~dw\\
=~&  \frac{\kappa^{\frac{a + \rank - 2r}{2}}~ v^{-\frac{a + \rank}{2}}}{(1-R^2)^{\frac{n-1}{2}} }
	B\left(\frac{a + \rank}{2}, \frac{b}{2}\right) \cdot\\
&	F_1\left( \frac{a + \rank}{2}; \frac{a + b + \rank + 1 - n-2r}{2}, \frac{n-1}{2}; \frac{a + b+ \rank}{2}; 
	1-\kappa, \frac{(1-R^2)v(1-\kappa) - R^2\kappa}{(1-R^2)v} \right).
\end{align*}

\end{proof}

\subsection{Derivation of \eqref{eq:de_marlik}}\label{PROOFeq:de_marlik}

\begin{proof}
Similar to \eqref{eq:marlik_fixedg}, we apply integrated Laplace approximation  to obtain 
$p (\mathbf{Y} \mid \phi, \MM, g)$, then marginalize $\phi$ out as follows.
\begin{align*}
&p (\mathbf{Y} \mid \MM, g) 
 	= \int_0^{\infty } p (\mathbf{Y} \mid \phi, \MM, g) p(\phi) d\phi\\
\propto~&  \int_0^{\infty } p (\mathbf{Y} \mid \hat{\alpha}_\MM, \hat{\boldsymbol\beta}_\MM, \phi, \MM)
  	\left[ \phi\mathcal{J}_n(\hat{\alpha}_\MM) \right]^{-\frac{1}{2}}
  	(1+g)^{-\frac{p_{\MM}}{2}} e^{ -\frac{\phi Q_{\MM}}{2(1+g)} } \phi^{-1}d\phi\\
\propto~& \left[ \mathcal{J}_n(\hat{\alpha}_\MM) \right]^{-\frac{1}{2}} (1+g)^{-\frac{p_{\MM}}{2}}
  	\int_0^{\infty } \phi^{\frac{n-1}{2}-1} 
   	e^{ \phi \left\{ -\frac{ Q_{\MM}}{2(1+g)} + \sum_{i=1}^n \left[
	Y_i (\hat{\theta}_i - t_i) - b(\hat{\theta}_i) + b(t_i)
	\right] \right\} }d\phi\\
\propto~& \left[ \mathcal{J}_n(\hat{\alpha}_\MM) \right]^{-\frac{1}{2}} (1+g)^{-\frac{p_{\MM}}{2}}
  	\left\{ \frac{ Q_{\MM}}{2(1+g)} - \sum_{i=1}^n \left[
	Y_i (\hat{\theta}_i - t_i) - b(\hat{\theta}_i) + b(t_i)
	\right] \right\}^{-\frac{n-1}{2}} \\
\propto~& \frac{\left[ \mathcal{J}_n(\hat{\alpha}_\MM) \right]^{-\frac{1}{2}} u^{\frac{p_{\MM}}{2}}}
  	{\left\{ u Q_{\MM} + 2\sum_{i=1}^n \left[
	Y_i (t_i - \hat{\theta}_i) - b(t_i) + b(\hat{\theta}_i) 
	\right] \right\}^{\frac{n-1}{2}} }.
\end{align*}
Here, the last step replaces $g$ with $u = 1/(1+g)$.
\end{proof}

\subsection{Proof of Model Selection Consistency}\label{PROOFth:selection_CHg}

We first show a lemma about a non-central $\chi^2$ distribution, 
which is useful to prove some of the following lemmas and theorems.
Here the symbol $\chi^2_k(m)$ denotes a non-central $\chi^2$ distribution 
with degrees of freedom $k$ and non-centrality parameter $m$.

\begin{lemma}\label{lemma:non-central_chi-squared}
If a sequence of random variables $\{X_n: n = 1, 2, \ldots\}$  have independent
non-central $\chi^2$ distributions:
$X_n \sim \chi^2_k(nA_n)$, where random variables 
$A_n \stackrel{\text{D}}{\longrightarrow} a_0\in \mathbb{R}^+\cup\{0\}$, 
then as $n\longrightarrow\infty$,
$X_n/n \stackrel{\text{P}}{\longrightarrow} a_0$.
\end{lemma}

\begin{proof}
For any $n \in \mathbb{N}$,
the characteristic function of $X_n / n$ evaluated at $t \in \mathbb{R}$ is
\begin{align*}
\phi_{X_n/n}(t) & = \mathbb{E}\left( e^{itX_n / n} \right)
	= \mathbb{E}_{A_n}\left[ \mathbb{E}\left( e^{itX_n / n} \mid A_n \right) \right]\\
& =	\mathbb{E}_{A_n}\left[ \exp\left(\frac{itA_n}{1 - 2it/n}\right)(1 - 2it/n)^{-\frac{k}{2}} \right]
   = (1 - 2it/n)^{-\frac{k}{2}} \cdot \mathbb{E}_{A_n}\left[ \exp\left(\frac{itA_n}{1 - 2it/n}\right) \right].
\end{align*}
Denote a complex valued random variable $B_n = A_n / (1 - 2it/n)$.
Since the limit of $A_n$ is a constant, for the series $\{A_n: n \in \mathbb{N}\}$,
convergence in distribution is equivalent to convergence in probability. 
Because of the continuous mapping theorem,
$B_n \stackrel{\text{P}}{\longrightarrow} a_0$, or equivalently, convergence in distribution.
Denote the bounded and continuous function $h(B_n) = \exp\left(itB_n\right)$,
then according to Portmanteau lemma,
$\mathbb{E}\left[ h(B_n) \right] {\longrightarrow} \mathbb{E}\left[ h(a_0) \right] = h(a_0)$.
So for any $t\in \mathbb{R}$,
\[
\lim_{n\rightarrow \infty} \phi_{X_n/n}(t) = \lim_{n\rightarrow \infty}(1 - 2it/n)^{-k/2} \cdot
 \lim_{n\rightarrow \infty}\mathbb{E}\left[ h(B_n) \right] = h(a_0) = \exp\left(ita_0\right),
\]
where the limit is the characteristic function of a degenerated distribution at $a_0$.
Therefore, $X_n/n$ converge in distribution to a constant $a_0$,
which implies convergence in probability.

\end{proof}

In order to show the asymptotic performance of the Bayes factor $\BF_{\MT:\MM}$,
we first study asymptotic behaviors of the terms in the Bayes factors in the following lemmas. 
When testing nested models, 
the log likelihood ratio between $\MT$ and $\MM$ converges in distribution to 
a central (non-central) $\chi^2$ distribution, when the smaller (larger) model is true. 
The following lemma studies asymptotic behaviors of the likelihood ratio,
which does not require 
models $\MM$ and $\MT$ to be nested.

\begin{lemma}
\label{lemma:Lambda}
Denote the the likelihood ratio by
\begin{equation}\label{eq:Lambda}
\Lambda_{\MT: \MM}
  \stackrel{\triangle}{=} 
  \frac{p(\mathbf{Y}| \hat{\alpha}_{\MT}, \hat{\boldsymbol\beta}_{\MT}, \MT)}
              {p(\mathbf{Y}| \hat{\alpha}_{\MM}, \hat{\boldsymbol\beta}_{\MM}, \MM)}
              = \exp\left( \frac{z_\MT - z_\MM}{2} \right)
\end{equation}
As the sample size $n$ increases,
\begin{enumerate}[label = \arabic*)]
\item if $\MT \subset \MM$, then $\Lambda_{\MT: \MM} = O_P(1)$.
\item if $\MT \not\subset \MM$, then $\Lambda_{\MT: \MM} = O_P\left(e^{c_\MM n}\right)$, 
where $c_\MM$ is a positive constant.
\end{enumerate}
\end{lemma}
\begin{proof}

In the first case where $\MM \supset \MT$, from the well-known results of likelihood ratio test,
$z_\MM - z_\MT$ has a central chi-square distribution 
$\chi^2_{p_{\MM} - p_\MT}$.
Therefore, the limiting distribution of the log-likelihood ratio does not depend on $n$, i.e.,
$\Lambda_{\MT: \MM} = O_P(1)$.

In the second case where $\MM \not\supset \MT$, 
we first examine the sub-case where $\MM \subset \MT$.
According to the power calculation results for GLM in \citet{Self_etal_1992}
and \citet{Shieh_2000},
when testing nested models, if the larger model is true, 
then we have that $z_\MT - z_\MM$ converges in distribution 
to a non-central $\chi^2$ of degrees of freedom $p_\MT - p_{\MM}$.
The non-centrality parameter $\Psi$ is approximately
\[
\Psi  \approx \sum_{i=1}^n  b^{\prime} (\theta_{\MT, i}^*)
  \left( \theta_{i, \MT}^* - \theta_{i, \MM}^* \right) - \left[ b(\theta_{i, \MT}^*) - b(\theta_{i, \MM}^*) \right],
\]
where $\theta_{i, \MM}^* = \theta(\eta_{i, \MM}^*)$, for $i = 1, \ldots, n$.
By a Taylor expansion, there exist a $\tilde{\theta}_i$ between $\theta_{\MT, i}^*$ and $\theta_{\MM, i}^*$, such that
$b(\theta_{i, \MM}^*) = b(\theta_{i, \MT}^*) + b^{\prime} (\theta_{\MT, i}^*)
  \left( \theta_{i, \MT}^* - \theta_{i, \MM}^* \right) + b''(\tilde{\theta}_i)\left( \theta_{i, \MT}^* - \theta_{i, \MM}^* \right)^2/2$.
This combined with the assumption $b''(\cdot) > 0$ gives that 
$\lim_{n \rightarrow \infty}\Psi/n$ converges to a positive constant $c_\MM$.
Then by Lemma \ref{lemma:non-central_chi-squared},
$(z_\MT - z_\MM)/n\stackrel{P}{\longrightarrow} c_\MM$,
and hence $\Lambda_{\MT: \MM} = O_P(e^{c_\MM n})$.

In the case where $\MM$ and $\MT$ are not nested, 
we introduce a third model $\MM^{\prime}$ which includes all the predictors in both $\MM$ and $\MT$.
Using a similar method as in \citet{Self_etal_1992}, we can treat $\MM^{\prime}$ also as the true model (although with some redundant predictors) when comparing with $\MM$ and easily show that 
$ \Lambda_{\MM^{\prime} : \MM}$ also has a non-central $\chi^2$ distribution.
Hence we decompose $\Lambda_{\MT : \MM} 
= \Lambda_{\MT : \MM^{\prime}} \cdot \Lambda_{\MM^{\prime} : \MM}$.
Since both pairs $(\MT, \MM^{\prime})$ and $(\MM^{\prime} : \MM)$ 
are nested models, we can apply the previous results twice:
$ \Lambda_{\MT : \MM^{\prime}} = O_P(1)$ and 
$ \Lambda_{\MM^{\prime} : \MM}  = O_P(e^{c_\MM n})$.
Therefore, we can conclude that $\Lambda_{\MT : \MM} 
= O_P(1) \cdot O_P(e^{c_\MM n})= O_P(e^{c_\MM n})$. 
\end{proof}

The Bayes factors contain 
the Wald statistics $Q_\MT$ and $Q_\MM$.
We next study their asymptotic behaviors.


\begin{lemma}
\label{lemma:Q}  The Wald statistic 
$Q_{\MM}  = O_P(n^{\xi_\MM})$,
where $0 \leq \xi_\MM \leq 1$. In particular,
\begin{enumerate}[label = \arabic*)]
\item If $\MT \neq \Mnull$, then for any $\MM \supset \MT$, $\xi_\MM = 1$.
\item if $\MT = \Mnull$, then for any model $\MM$, $\xi_\MM = 0$.
\end{enumerate}
\end{lemma}
\begin{proof}
For any $\MM \supset \MT$, we have shown in the proof of 
Lemma \ref{lemma:order_Jalpha_Jbeta} 
 that the MLE $\hat{\boldsymbol\beta}_\MM$ converges in probability  
to the true value $\b_\MM^*$, and 
$\mathcal{J}_n(\hat{\b}_\MM) / n$ is a finite positive definite matrix 
and converges to $\mathcal{I}_n(\b_\MM^*) / n$ in probability.
By Lemma \ref{lemma:order_Jalpha_Jbeta} and 
Slutsky's theorem, we can rewrite the asymptotic normality \eqref{eq:asym_MLE_MT} as
\[
\mathcal{J}_n(\hat{\b}_\MM)^{\frac{1}{2}}\left( \hat{\b}_\MM - \b_\MM^* \right)
\stackrel{d}{\longrightarrow} \text{N}(0, \mathbf{I}_{p_\MM}).
\]
Therefore, $Q_\MM = \hat{\b}_\MM^T \mathcal{J}_n(\hat{\b}_\MM) \hat{\b}_\MM$ converges 
in distribution to a non-central $\chi^2$ random variable 
with degrees of freedom $p_\MM$ and non-centrality parameter
$\b_\MM^{*T} \mathcal{I}_n(\b_\MM^*)\b_\MM^*$, which is $O(n)$ 
if $\b_\MM^* \neq \mathbf{0}$, and zero otherwise. 
Since $\b_\MM^* = \b_\MT^*$ in the sense that all entries in 
$\b_\MM^*$ that correspond to predictors not in $\MT$ are filled with zero, 
$\b_\MM^* = \mathbf{0}$ is equivalent to $\MT = \Mnull$.
Therefore, by Lemma \ref{lemma:non-central_chi-squared},
if $\MT \neq \Mnull$, then $Q_\MM = O_P(n)$; 
if $\MT = \Mnull$, then $Q_\MM = O_P(1)$.

For any $\MM \not\supset \MT$, 
since convergence in probability is
preserved under addition and multiplication \citep[pp.\ 175]{Resnick_1999},
we have
$Q_\MM - \b_\MM^{*T} \mathcal{J}_n(\hat{\b}_\MM)\b_\MM^*
\stackrel{P}{\longrightarrow} 0$, i.e., $Q_\MM$ is at most
on the same order of $\mathcal{J}_n(\hat{\b}_\MM)$.
By Lemma \ref{lemma:order_Jalpha_Jbeta}, 
we have $\xi_\MM = \tau_\MM$ if $\b_\MM^* \neq \mathbf{0}$, 
and $\xi_\MM = 0$ if $\b_\MM^* = \mathbf{0}$.
\end{proof}

Based on the results of Lemma \ref{lemma:Q}, 
the next lemma discusses the asymptotic properties of 
$\Omega_{\MT: \MM}^{\text{CH}}$, a term that appears in the Bayes factor
under the CH prior.


\begin{lemma}
\label{lemma:Omega_CHg}  
Under the CH prior, denote the term in $\BF_{\MT:\MM}$: 
\begin{equation}\label{eq:Omega_CHg}
\Omega_{\MT: \MM}^{\text{CH}} \stackrel{\triangle}{=}
  \frac{B\left( \frac{a+p_{\MT}}{2}, \frac{b}{2} \right)
  \  _{1}F_{1}\left( \frac{a+p_{\MT}}{2}, \frac{a+b+p_{\MT}}{2}, -\frac{s + Q_{\MT}}{2} \right)}
  {B\left( \frac{a+p_{\MM}}{2}, \frac{b}{2} \right)
  \  _{1}F_{1}\left( \frac{a+p_{\MM}}{2}, \frac{a+b+p_{\MM}}{2}, -\frac{s + Q_{\MM}}{2} \right)}.
\end{equation}

\begin{enumerate}[label = \arabic*)]
\item If $\MT \neq \Mnull$, 
then as $n$ increases,
\[
\Omega_{\MT: \MM}^{\text{CH}} = 
\begin{cases}
O_P\left(n^{\frac{\xi_\MM p_\MM - p_\MT - a(1-\xi_\MM)}{2}}\right)
&	\text{ if } b \text{ is fixed, and } s  \text{ is fixed}\\
O_P\left(n^{\frac{p_\MM - p_\MT}{2}}\right)
&	\text{ if } b  = O(n) \text{, or } s = O(n)\\
\end{cases}
\]
In particular, if $\MM \supset \MT$, 
then $\Omega_{\MT: \MM}^{\text{CH}} = O_P\left(n^{\frac{p_\MM - p_\MT}{2}}\right)$ for all $b$ and $s$.

\item If $\MT = \Mnull$, then as $n$ increases,
\[
\Omega_{\MT: \MM}^{\text{CH}} = 
\begin{cases}
O_P\left( 1 \right)
&	\text{ if } b \text{ is fixed, and } s  \text{ is fixed}\\
O_P\left(n^{\frac{p_\MM - p_\MT}{2}}\right)
&	\text{ if } b  = O(n) \text{, or } s = O(n)\\
\end{cases}
\]
\end{enumerate}

\end{lemma}
\begin{proof}
We first show Case 1) where $\MT \neq \Mnull$,  by Lemma \ref{lemma:Q}, 
$\xi_{\MT} = 1$.
We consider the following three scenarios about 
parameters $b$ and $s$ being fixed or $O(n)$.

\emph{Scenario 1: Both $b, s$ are fixed.}
By \citet{Abramowitz_Stegun_1970} formula (13.1.5), 
\begin{equation}\label{eq:1F1_large_negative_x}
 _{1} F_1 (a, b, s) = \frac{\Gamma(b)}{\Gamma(b-a)} (-s)^{-a} [1 + O(|s|^{-1})], \textrm{ when } \text{Real}(s)< 0.
\end{equation}
Continuous mapping theorem suggests that
for any model $\MM$ whose $Q_{\MM} = O_P(n^{\xi_\MM})$, 
\begin{equation}\label{eq:Omega_CHg2}
\Omega^{\text{CH}}_{\MT : \MM}  
 \approx
	\frac{ \Gamma\left( \frac{a+p_\MT}{2} \right) \left(\frac{s+Q_{\MT}}{2} \right)^{-\frac{a+p_\MT}{2}}}
 	{  \Gamma\left( \frac{a+p_{\MM}}{2} \right) \left(\frac{s+Q_{\MM}}{2}\right)^{-\frac{a+p_\MM}{2}} } 
 \propto  
	\frac{\left(s+Q_{\MT}\right)^{-\frac{a+p_\MT}{2}}}
	{\left(s+Q_{\MM}\right)^{-\frac{a+p_\MM}{2}}}
= O_P\left(n^{\frac{\xi_\MM p_\MM - p_\MT - a(1-\xi_\MM)}{2}}\right).
\end{equation}

\emph{Scenario 2: $b$ is fixed, and $s = O(n)$.}
Since $s + Q_\MT = O(n)$ and $s + Q_\MM = O(n)$, then by \eqref{eq:Omega_CHg2}, 
$\Omega^{\text{CH}}_{\MT : \MM} = O_P\left(n^{\frac{p_\MM - p_\MT}{2}}\right)$.

\emph{Scenario 3: $b = O(n)$.}
Lemma \ref{lemma:Q} indicates that $Q_{\MM}$ is between $O_P(1)$ and $O_P(n)$.
By \citet {Slater_1960}  formula (4.3.3): 
if $b$ is large, and $a, s$ are bounded, then 
\begin{equation}\label{eq:1F1_limit_large_b}
 _{1} F_1 (a, b, s) = 1 + O(|b|^{-1})  \text{ is bounded};
\end{equation}
and by \citet {Slater_1960} formulas (4.3.7): 
if $b$ is large, $s = by$, and $a, y$ are bounded, then
\begin{equation}\label{eq:1F1_limit_large_b_s}
 _{1} F_1 (a, b, s) = (1 - y)^{-a} \left[ 1 - \frac{a(a+1)}{2b}\left( \frac{y}{1-y} \right)^2 + O(|b|^{-2}) \right]
 \text{ is also bounded}.
\end{equation}
Therefore, under the CH prior when parameter $b=O(n)$, 
\begin{align*}
\Omega^{\text{CH}}_{\MT : \MM} & = \frac{B\left( \frac{a+p_\MT}{2}, \frac{b}{2} \right)
  \  _{1}F_{1}\left( \frac{a+p_\MT}{2}, \frac{a+b+p_\MT}{2}, -\frac{s+Q_{\MT}}{2} \right)}
  {B\left( \frac{a+p_{\MM}}{2}, \frac{b}{2} \right)\  _{1}F_{1}
  \left( \frac{a+p_{\MM}}{2}, \frac{a+b+p_{\MM}}{2}, -\frac{s+Q_{\MM}}{2} \right)}
\stackrel{\text{P}}{\longrightarrow} C \cdot \frac{B\left( \frac{a+p_\MT}{2}, \frac{b}{2} \right)}
  {B\left( \frac{a+p_{\MM}}{2}, \frac{b}{2} \right)}.
\end{align*}
According to the Stirling's Formula
$\Gamma(n) = e^{-n}n^{n-\frac{1}{2}}(2\pi)^{\frac{1}{2}}(1 + O(n^{-1}))$,
the above ratio becomes $O_P\left(n^{\frac{p_{\MM} - p_\MT}{2}}\right)$.

Next we examine Case 2) where $\MT = \Mnull$. 
In this case, Lemma \ref{lemma:Q} suggests that
both $Q_\MT$ and $Q_\MM$ are on the same order $O_P(1)$.
Hence in Scenario 1, where both $b$ and $s$ are fixed,
$\Omega^{\text{CH}}_{\MT:\MM} = O_P(1)$;
In Scenario 2, since both $s + Q_\MT$ and $s + Q_\MM$ are on the order of $O_P(n)$, 
the same deviation and result as in Case 1) Scenario 2 apply.
In Scenario 3,  both $s + Q_\MT$ and $s + Q_{\MM}$ 
are $O_P(1)$ if $s$ is fixed, and $O_P(n)$ if $s = O(n)$, so the same derivation and result
as in Case 1) Scenario 3 apply.
\end{proof}


\begin{lemma}
\label{lemma:Omega_Robust}  
Under the robust prior, denote the term in $\BF_{\MT:\MM}$:
\begin{equation}\label{eq:Omega_Robust}
\Omega_{\MT: \MM}^{\text{R}} \stackrel{\triangle}{=} 
\left(\frac{ p_\MM+1}{ p_\MT+1}\right)^{\frac{1}{2}}
  \cdot  \frac{Q_\MT^{-\frac{p_\MT+1}{2}}}{Q_\MM^{-\frac{p_\MM+1}{2}}} 
  \cdot \frac{\gamma \left(\frac{p_\MT+1}{2} , \frac{Q_\MT(p_\MT+1)}{2(n+1)} \right)}
  {\gamma \left(\frac{p_\MM+1}{2} , \frac{Q_\MM(p_\MM+1)}{2(n+1)} \right)}.
\end{equation}
As the sample size $n$ increases, $\Omega_{\MT: \MM}^{\text{R}}  = O_P\left(n^{\frac{p_\MM - p_\MT}{2}}\right)$.

\end{lemma}
\begin{proof}
By \citet{Abramowitz_Stegun_1970} formula (6.5.12),
the incomplete Gamma function $\gamma(a, s) = \int_0^{s} t^{a-1} e^{-t}dt$ 
can be expressed using the $_{1}F_1$ function
\begin{equation}\label{eq:incomplete_Gamma}
\gamma(a, s) =  {_1}F_{1}(a, a+1, -s) \frac{s^a}{a}.
\end{equation}
Therefore, \eqref{eq:Omega_Robust} becomes
\[
\left(\frac{ p_\MM+1}{ p_\MT+1}\right)^{\frac{1}{2}}
  \cdot  \frac{Q_\MT^{-\frac{p_\MT+1}{2}}}{Q_\MM^{-\frac{p_\MM+1}{2}}} 
  \cdot \frac{\left(\frac{p_\MT+1}{2}\right)^{-1}  \left(\frac{Q_\MT(p_\MT+1)}{2(n+1)}\right)^{\frac{p_\MT+1}{2}}
  \ _{1}F_{1} \left( \frac{p_\MT+1}{2}, \frac{p_\MT+3}{2}, -\frac{Q_\MT(p_\MT+1)}{2(n+1)} \right)}
  {\left(\frac{p_\MM+1}{2}\right)^{-1}  \left(\frac{Q_\MM(p_\MM+1)}{2(n+1)}\right)^{\frac{p_\MM+1}{2}}
  \ _{1}F_{1} \left( \frac{p_\MM+1}{2}, \frac{p_\MM+3}{2}, -\frac{Q_\MM(p_\MM+1)}{2(n+1)} \right)} 
\]
Since $_1F_1(a, b, 0) = 1$, and both $Q_\MT/n, Q_\MM/n$ are bounded, the ratio between the $_1F_1$ functions is bounded
as $n$ increases. Therefore we further simplify
$\Omega_{\MT: \MM}^{\text{R}}  \propto (n+1)^{\frac{p_\MM - p_\MT}{2}} 
	= O_P\left(n^{\frac{p_\MM - p_\MT}{2}}\right)$.
This result holds no matter whether $\MT=\Mnull$ or not.
\end{proof}

\begin{lemma}
\label{lemma:Omega_Intrinsic}  
Under the intrinsic prior, denote the term in $\BF_{\MT:\MM}$:
\begin{equation*}\label{eq:Omega_Intrinsic}
\Omega_{\MT: \MM}^{\text{I}} \stackrel{\triangle}{=} 
\frac{ \left( \frac{n+ p_\MM + 1}{p_\MM + 1} \right)^{\frac{p_\MM}{2}} 
	~ e^{\frac{Q_\MM\left(p_\MM + 1\right)}{2\left(n + p_\MM + 1\right)}} 
	~ B\left(\frac{p_\MT + 1}{2}, \frac{1}{2}  \right)
	~ \Phi_1\left( \frac{1}{2}, 1, \frac{p_\MT + 2}{2}, \frac{Q_\MT(p_\MT + 1)}{2(n + p_\MT + 1)}, -\frac{p_\MT + 1}{n} \right) }
{ \left( \frac{n+ p_\MT + 1}{p_\MT + 1} \right)^{\frac{p_\MT}{2}} 
	~ e^{\frac{Q_\MT\left(p_\MT + 1\right)}{2\left(n + p_\MT + 1\right)}}
	~ B\left(\frac{p_\MM + 1}{2}, \frac{1}{2}  \right)
	~ \Phi_1\left( \frac{1}{2}, 1, \frac{p_\MM + 2}{2}, \frac{Q_\MM(p_\MM + 1)}{2(n + p_\MM + 1)}, -\frac{p_\MM + 1}{n} \right) }
\end{equation*}
As the sample size $n$ increases, $\Omega_{\MT: \MM}^{\text{I}}  = O_P\left(n^{\frac{p_\MM - p_\MT}{2}}\right)$.
\end{lemma}
\begin{proof}
Since $p_\MT, p_\MM$ are bounded,
and $Q_\MT/n, Q_\MM/n$ are bounded in probability, as $n \rightarrow \infty$,
\[
\Omega_{\MT: \MM}^{\text{I}} 
\stackrel{\text{P}}{\longrightarrow} C \cdot \frac{ \left( \frac{n+ p_\MM + 1}{p_\MM + 1} \right)^{\frac{p_\MM}{2}} }
  {\left( \frac{n+ p_\MT + 1}{p_\MT + 1} \right)^{\frac{p_\MT}{2}} } 
  = O_P\left(n^{\frac{p_\MM - p_\MT}{2}}\right).
\]
\end{proof}

\begin{lemma}
\label{lemma:Omega_LEB}  
Under the local EB, denote the term in $\BF_{\MT:\MM}$: 
\begin{equation}\label{eq:Omega_LEB}
\Omega_{\MT: \MM}^{\text{LEB}} \stackrel{\triangle}{=}
  \frac{\max\left\{ \exp\left(-\frac{Q_\MT}{2}\right), 
  \left(\frac{Q_\MT}{p_\MT}\right)^{-\frac{p_\MT}{2}}\exp\left(-\frac{p_\MT}{2}\right)\right\}}
  {\max\left\{ \exp\left(-\frac{Q_\MM}{2}\right), 
  \left(\frac{Q_\MM}{p_\MM}\right)^{-\frac{p_\MM}{2}}\exp\left(-\frac{p_\MM}{2}\right)\right\}}.
\end{equation}

\begin{enumerate}[label = \arabic*)]
\item If $\MT \neq \Mnull$, 
then as $n$ increases,
$\Omega_{\MT: \MM}^{\text{LEB}} = O_P\left( n^{\frac{\xi_\MM p_\MM - p_\MT}{2} }\right)$.
In particular, if $\MM \supset \MT$, 
then $\Omega_{\MT: \MM}^{\text{LEB}} = O_P\left(n^{\frac{p_\MM - p_\MT}{2}}\right)$.

\item If $\MT = \Mnull$, then as $n$ increases,
$\Omega_{\MT: \MM}^{\text{LEB}} = O_P\left( 1 \right)$.
\end{enumerate}
\end{lemma}
\begin{proof}
Case 2) is straightforward, because when $\MT = \Mnull$,
$Q_\MT = O_P(1)$ and $Q_\MM = O_P(1)$. 
Now let us focus on Case 1). 
In \eqref{eq:Omega_LEB}, the numerator equals $\exp(-Q_\MT/2)$
if and only if $Q_\MT \leq p_\MT$, and 
the denominator follows the same rule when we replacing 
$\MT$ with $\MM$. 
Since $\MT \neq \Mnull$, $Q_\MM = O_P(n)$ is greater
than $p_\MM$ for large $n$.
Hence the numerator of \eqref{eq:Omega_LEB}
is proportional to $\left(Q_\MT/p_\MT\right)^{-\frac{p_\MT}{2}}\exp\left(-p_\MT / 2\right)
= O_P(n^{-\frac{p_\MT}{2}})$.
For model $\MM$ whose $Q_\MM = O_P(n^{\xi_\MM})$, 
if $\xi_\MM > 0$, then when $n$ is large enough, 
$Q_\MM > p_\MM$, so the denominator is $O_P(n^{-\frac{\xi_\MM p_\MM}{2}})$.
If $\xi_\MM = 0$, then the denominator is $O_P(1)$,
which can also be written as $O_P(n^{-\frac{\xi_\MM p_\MM}{2}})$.
\end{proof}

We now examine the model selection consistency.

\noindent {\bf Proof of Theorem \ref{th:selection_CHg}}
\begin{proof}
By Lemma \ref{lemma:order_Jalpha_Jbeta}, 
$\mathcal{J}_n(\hat{\alpha}_{\MM}) = O_P(n^{\tau_\MM})$,
where $0 \leq \tau_\MM \leq 1$, and $\tau_\MM =1$ if $\MM \supset \MT$.
Hence,
\[
\left[ \frac{ \mathcal{J}_n(\hat{\alpha}_{\MT}) }
              { \mathcal{J}_n(\hat{\alpha}_{\MM}) }
\right]^{-\frac{1}{2}} = O_P\left( n^{-\frac{1 - \tau_\MM}{2}} \right).
\]

For the CH prior, 
\begin{equation}\label{eq:BF_chg}
\BF_{\MT : \MM} = \left[ \frac{ \mathcal{J}_n(\hat{\alpha}_{\MT}) }
              { \mathcal{J}_n(\hat{\alpha}_{\MM}) }\right]^{-\frac{1}{2}} \cdot
\Lambda_{\MT: \MM} \cdot \Omega_{\MT: \MM}^{\text{CH}}\cdot [1 + O_P(1/n)].
\end{equation}
We first consider the case where both $b$ and $s$ are fixed,
by using the results in Lemma \ref{lemma:Lambda} and \ref{lemma:Omega_CHg}.
In the case where $\MT \neq \Mnull$, for any non-true model $\MM \supset\MT$, 
then $p_\MM > p_\MT$, $\tau_\MM = 1$, and $\xi_\MM = 1$, hence
\[
\BF_{\MT : \MM} 
= O_P(1)\cdot O_P(1)\cdot O_P\left(n^{\frac{p_\MM - p_\MT}{2}}\right)
\cdot [1 + O_P(1/n)] \stackrel{\text{P}}{\longrightarrow} \infty.
\]
On the other hand, if $\MM \not\supset \MT$, then 
\[
\BF_{\MT : \MM} 
=	 O_P\left( n^{-\frac{1 - \tau_\MM}{2}} \right) \cdot 
	O_P\left(e^{c_\MM n}\right)\cdot 
	O_P\left(n^{\frac{\xi_\MM p_\MM - p_\MT - a(1-\xi_\MM)}{2}}\right) 
	\cdot [1 + O_P(1/n)]  \stackrel{\text{P}}{\longrightarrow} \infty.
\]
In contrast, if $\MT = \Mnull$, then for any model $\MM$, since $\MM \supset \MT$,
$\tau_\MM = 1$.
So the Bayes factor
\[
\BF_{\MT : \MM}
= O_P\left(1\right)\cdot  O_P\left(1\right)\cdot O_P\left(1\right)\cdot [1 + O_P(1/n)]
\] is bounded,
which suggests the selection consistency does not hold when $\MT = \Mnull$.

Next consider the case where $b = O(n)$ or $s = O(n)$.
For any model $\MM  \not\supset \MT$, the proof is similar as above. 
If $\MM \supset \MT$, then $\tau_\MM = 1$ and $p_\MM > p_\MT$, so 
\[
\BF_{\MT : \MM} 
= 	O_P(1)\cdot O_P(1)\cdot 
	O_P\left(n^{\frac{p_\MM - p_\MT}{2}}\right)\cdot [1 + O_P(1/n)] \stackrel{\text{P}}{\longrightarrow} \infty,
\]
which holds even when $\MT = \Mnull$.

For the robust prior, the intrinsic prior, and local EB, their Bayes factor are given by \eqref{eq:BF_chg}, with 
$\Omega_{\MT: \MM}^{\text{CH}}$ replaced by $\Omega_{\MT: \MM}^{\text{R}}$,
$\Omega_{\MT: \MM}^{\text{I}}$, and $\Omega_{\MT: \MM}^{\text{LEB}}$, respectively.
By Lemma \ref{lemma:Omega_Robust}, \ref{lemma:Omega_Intrinsic}, and \ref{lemma:Omega_LEB}, 
the proofs are similar
to the CH prior, hence omitted.
\end{proof}

\subsection{Proof to Proposition \ref{proposition:CH-g_intrinsic}}\label{PROOFproposition:CH-g_intrinsic}
\begin{proof}
If $b = O(n)$\, then by \eqref{eq:1F1_limit_large_b} or \eqref{eq:1F1_limit_large_b_s},
\begin{align}\label{eq:prior_mean_1/g}
\mathbb{E}(1/g)  
& =   \frac{B\left(  \frac{a}{2} + 1, \frac{b}{2}-1 \right)
  \  _{1}F_{1}\left( \frac{a}{2} + 1, \frac{a+b}{2}, -\frac{s}{2} \right)}
  {B\left( \frac{a}{2}, \frac{b}{2} \right)
  \  _{1}F_{1}\left( \frac{a}{2}, \frac{a+b}{2}, -\frac{s}{2} \right)} \\ \nonumber
& \propto   \frac{B\left(  \frac{a}{2} + 1, \frac{b}{2}-1 \right)
  }
  {B\left( \frac{a}{2}, \frac{b}{2} \right)
  } 
  \longrightarrow  \frac{a}{b-2} = O(1/n).
\end{align}
If $b$ is fixed and $s = O(n)$, then by \eqref{eq:1F1_large_negative_x} and 
\eqref{eq:prior_mean_1/g},
\[
\mathbb{E}(1/g) 
 \approx   \frac{B\left(  \frac{a}{2} + 1, \frac{b}{2}-1 \right)
  \Gamma\left( \frac{b}{2} \right)\left(\frac{s}{2}\right)^{\frac{a}{2}}
  }
  {B\left( \frac{a}{2}, \frac{b}{2} \right)
  \Gamma\left( \frac{b}{2} -1 \right) \left(\frac{s}{2}\right)^{\frac{a}{2}+1}
  } 
  \propto \frac{1}{s} = O(1/n).
\]

\end{proof}

\subsection{Proof of Proposition \ref{prop:consistency_shrinkage_factor1}}\label{PROOFprop:consistency_shrinkage_factor1}

\begin{proof}
For the CH prior, 
according to \eqref{eq:upost}, the conditional posterior of $z = 1-u$ is
\begin{equation}\label{eq:z_post}
z \mid \mathbf{Y}, \MM \convd  
  \text{CH}\left(\frac{b}{2}, \frac{a + p_\MM}{2},  -\frac{s + Q_\MM}{2} \right),
\end{equation}
and its characteristic function is
\begin{align}
\label{eq:chf_z} \phi_z(t) & = \mathbb{E}\left( e^{itz} \right)
 = \int \frac{z^{\frac{b}{2}-1} (1-z)^{\frac{a+p_{\MT}}{2}-1} e^{\left(\frac{s+Q_{\MT}}{2} + it\right)z}}{B(\frac{b}{2}, \frac{a+p_{\MT}}{2})\  _{1}F_{1}(\frac{b}{2}, \frac{a+b+p_{\MT}}{2}, \frac{s+Q_{\MT}}{2})} dz \nonumber 
 = \frac{\  _{1}F_{1}(\frac{b}{2}, \frac{a+b+p_{\MT}}{2}, \frac{s+Q_{\MT}}{2}+it)}{\  _{1}F_{1}(\frac{b}{2}, \frac{a+b+p_{\MT}}{2}, \frac{s+Q_{\MT}}{2})} \nonumber
\end{align}
Lemma \ref{lemma:Q} shows that if $\MT \neq \Mnull$, then $s + Q_\MT = O_P(n)$. 
If $b = O(1)$,
then by \eqref{eq:1F1_limit_positive_s}
and the continuous mapping theorem,
 for any $t \in \mathbb{R}$, as $n$ goes in to infinity,
\begin{equation*}
 \phi_z(t)  
\longrightarrow  \frac{ \exp(\frac{s+Q_{\MT}}{2}+it) \cdot (\frac{s+Q_{\MT}}{2}+it)^{-\frac{a+p_{\MT}}{2}} }{\exp(\frac{s+Q_{\MT}}{2}) \cdot (\frac{s+Q_{\MT}}{2})^{-\frac{a+p_{\MT}}{2}} } 
\stackrel{P}{\longrightarrow}  \exp(it).
\end{equation*}
If $b = O(n)$, then using formula \eqref{eq:1F1_limit_large_b_s},
we can obtain the same limit.

For the robust prior, we examine the characteristic function of $u = 1-z$.
Based on \eqref{eq:u_post_robust},
\begin{align*}
\phi_u(t)  & = \mathbb{E}\left( e^{itu} \right)
 = \frac{\displaystyle\int_0^{\frac{p_\MT + 1}{n + 1}}
	u^{\frac{p_\MT+1}{2}-1} e^{\left(it - \frac{Q_{\MT}}{2} \right)u}du}
	{\displaystyle\int_0^{\frac{p_\MT + 1}{n + 1}}
	u^{\frac{p_\MT+1}{2}-1} e^{- \frac{Q_{\MT}u}{2}}du}\\ 
& = \frac{\gamma\left(\frac{p_\MT + 1}{2},  \frac{(Q_{\MT}-2it)(p_\MT + 1)}{2(n + 1)}\right)}
	{\gamma\left(\frac{p_\MT + 1}{2}, \frac{Q_{\MT}(p_\MT + 1)}{2(n + 1)}\right)}
	\cdot\left(\frac{Q_{\MT}-2it}{Q_{\MT}}  \right)^{-\frac{p_\MT+1}{2}}. 
\end{align*}
Since $Q_\MT = O_P(n)$, for any fixed $t\in \mathbb{R}$, 
the ratio of the incomplete Gamma functions goes to $1$, and 
so does the second fraction. Therefore, $\phi_u(t)  \stackrel{\text{P}}{\longrightarrow} 1$,
which is the characteristic function of the degenerate distribution at $0$.

For the intrinsic prior, by \eqref{eq:tCCH_post} and Table \ref{tb:tCCHg_parameters}, 
the conditional posterior of $u$ is
\begin{equation}\label{eq:u_post_instrinsic}
u \mid \mathbf{Y}, \MT \sim 
\text{tCCH}\left( \frac{p_\MT + 1}{2}, \frac{1}{2}, 1, \frac{Q_\MT}{2}, 
\frac{n + p_\MT + 1}{p_\MT + 1}, \frac{n + p_\MT + 1}{n} \right),
\end{equation}
and hence its characteristic function for any $t \in \mathbb{R}$ is
\begin{equation}\label{eq:chfcn_intrinsic}
\phi_u(t)  
= 	\exp\left\{\frac{ it (p_\MT + 1)}{n + p_\MT + 1} \right\}
	\frac{\Phi_1\left(\frac{1}{2}, 1, \frac{p_\MT + 2}{2}, 
	\frac{(Q_\MT - 2it)(p_\MT + 1)}{2(n + p_\MT + 1)}, -\frac{p_\MT + 1}{n}  \right)}
	{\Phi_1\left(\frac{1}{2}, 1, \frac{p_\MT + 2}{2}, 
	\frac{Q_\MT(p_\MT + 1)}{2(n + p_\MT + 1)}, -\frac{p_\MT + 1}{n}  \right)}.
\end{equation}
Since $Q_\MT = O_P(n)$ and
\[
\frac{(Q_\MT - 2it)(p_\MT + 1)}{2(n + p_\MT + 1)} 
- \frac{Q_\MT(p_\MT + 1)}{2(n + p_\MT + 1)} \stackrel{\text{P}}{\longrightarrow} 0,
\]
by continuous mapping theorem, 
the ratio of the two $\Phi_1$ functions in \eqref{eq:chfcn_intrinsic} converges
to one in probability.
Therefore, under the intrinsic prior, 
$\phi_u(t) \stackrel{\text{P}}{\longrightarrow} 1$.
\end{proof}

\subsection{Proof of Theorem \ref{thm:consistency_post_beta}}\label{PROOFthm:consistency_post_beta}

\begin{proof}

For the CH prior, we will prove the BMA estimation consistency in two steps: 
1) $\MT \neq \Mnull$ and 2) $\MT = \Mnull$.
When $\MT \neq \Mnull$, the model selection consistency always holds,
so we just need to show the estimation consistency under the true model $\MT$.
For notation simplicity, we denote $\boldsymbol\Sigma_{n, \MM} 
= \mathcal{J}_n(\hat{\boldsymbol\beta}_{\MM})^{-1}$.
According to  \eqref{eq:betapost} and \eqref{eq:z_post}, the characteristic function
of the posterior distribution $p(\boldsymbol\beta_{\MT} \mid \MT, \mathbf{Y})$ is
\begin{align*}
 \phi_{\boldsymbol\beta_{\MT}}(\mathbf{t})
& = \int e^{i \mathbf{t}^T \boldsymbol\beta_{\MT}}\ 
  p(\boldsymbol\beta_{\MT} | \MT, \mathbf{Y})\ d\boldsymbol\beta_{\MT}\\
& 
= \int e^{i \mathbf{t}^T \boldsymbol\beta_{\MT}} 
  \left\{ \int p(\boldsymbol\beta_{\MT} 
  | z,\MT, \mathbf{Y})\ p(z | \MT, \mathbf{Y}) dz \right\} d\boldsymbol\beta_{\MT}\\
 & = \int \left\{ \int e^{i \mathbf{t}^T \boldsymbol\beta_{\MT}}\ 
  p(\boldsymbol\beta_{\MT} 
  | z, \MT, \mathbf{Y})\ d\boldsymbol\beta_{\MT} \right\} p(z | \MT, \mathbf{Y}) dz \\
&
= \int  e^{ z(i\mathbf{t}^T \hat{\boldsymbol\beta}_{\MT}
   - \frac{1}{2}  \mathbf{t}^T \boldsymbol\Sigma_{n,\MT}  \mathbf{t}) } \ p(z| \MT, \mathbf{Y}) dz
\end{align*}
In the above calculation, the integrand 
$e^{i \mathbf{t}^T \boldsymbol\beta_{\MT}}$ has a bounded modulus,
so according to Fubini's Theorem, the two integrals
(with respect to $z$ and $\boldsymbol\beta_{\MT}$) can be interchanged.
Since $Q_\MT = O_P(n)$ and $\boldsymbol\Sigma_{n,\MT} = O_P(n^{-1})$, 
using methods similar to the proof of Proposition \ref{prop:consistency_shrinkage_factor1}
and asymptotic normality of MLE,
we can show that for any vector $\mathbf{t}$,
\begin{align*}
\phi_{\boldsymbol\beta_{\MT}}(\mathbf{t})
\longrightarrow
  e^{i\mathbf{t}^T \hat{\boldsymbol\beta}_{\MT} 
  - \frac{1}{2}  \mathbf{t}^T \boldsymbol\Sigma_{n,\MT}  \mathbf{t}}
\stackrel{P}{\longrightarrow} e^{i\mathbf{t}^T \boldsymbol\beta_{\MT}^*}.
\end{align*}

On the other hand, when $\MT = \Mnull$,
under the CH prior model selection consistency does not hold if 
both $b$ and $s$ are fixed.
Hence we need to examine the limit of posterior distribution of $\boldsymbol\beta_\MM$ 
under all models.
Under any model $\MM$, the true model is nested in it, so
the MLE of the coefficient $\hat{\boldsymbol\beta}_\MM$ 
converges to the true parameters $\mathbf{0}$ in probability as $n$ goes to infinity.
Since the modulus of $e^{i \mathbf{t}^T \boldsymbol\beta_{\MM}}$ is bounded by a constant $1$, 
which is integrable if regarded as a function of $z$,
so according to the dominated convergence theorem, 
the characteristic function of the posterior distribution $p(\boldsymbol\beta_{\MM} \mid \mathbf{Y}, \MM)$ 
evaluated at any vector $\mathbf{t} \in \mathbb{R}^p$ is
\begin{align*}
\phi_{\boldsymbol\beta_{\MM}}(\mathbf{t}) 
& = \int  e^{ z(i\mathbf{t}^T \hat{\boldsymbol\beta}_{\MM} 
  - \frac{1}{2}  \mathbf{t}^T \boldsymbol\Sigma_{n,\MM}  \mathbf{t}) } \ p(z\mid \MM, \mathbf{Y}) dz\\
& \stackrel{P}{\longrightarrow}\int  \left[ e^{ z(i\mathbf{t}^T \mathbf{0} 
 - \frac{1}{2}  \mathbf{t}^T \mathbf{0} \mathbf{t}) } \right]  p(z\mid \MM, \mathbf{Y}) dz =1.
\end{align*}

For the robust and intrinsic priors, model selection consistency always holds. 
So we just need to consider under $\MT$.
Based on \eqref{eq:u_post_robust} and \eqref{eq:u_post_instrinsic},
proofs similar to the above proof of the CH prior can show that
either $\MT \neq \Mnull$ or $\MT =\Mnull$, 
the characteristic function of $p(\boldsymbol\beta_{\MT} \mid \MT, \mathbf{Y})$ 
converges to $e^{i\mathbf{t}^T \boldsymbol\beta_{\MT}^*}$ or $1$ in probability,
 respectively.
\end{proof}

\section{Test-Based Bayes Factors}\label{section:TBF} 

\subsection{Test-Based Bayes Factor under the $g$-Prior}\label{subsection:TBF}
In Bayesian hypothesis testing, 
while the traditional Bayes factor computes the ratio between marginal likelihoods of data
(referred to as data-based BF, or DBF in short), another type of Bayes factor,
defined as the ratio between marginal likelihoods of a test statistic, has also been 
introduced \citep{Johnson_2005, Johnson_2008}. 
In particular, based on the likelihood ratio statistic, 
the test-based Bayes factor (TBF) has been applied 
in model selection under the $g$-prior
\citep{Hu_Johnson_2009, Held_etal_2015, Held_etal_2016},
where models with high TBFs are preferable. 

To compute the TBF based on the likelihood ratio deviance $z_\MM$ \eqref{eq:z}, 
first, asymptotic theory \citep{Davidson_Lever_1970} suggests that
the limit distribution of $z_\MM$ under the null model $\Mnull$
and under a local alternative model $\MM$ are central and non-central
Chi-squares, respectively,
\[
z_\MM \mid \Mnull \sim \chi^2_{p_\MM}, \quad
z_\MM \mid \MM \sim \chi^2_{p_\MM}(\lambda_\MM), \ \text{where }
\lambda_\MM= \b_\MM^T \mathcal{I}_n(\b_\MM = \zero) \b_\MM.
\]
Then, as $p(z_\MM \mid \MM, \b_\MM)$ depends on $\b_\MM$ through the 
non-centrality parameter $\lambda_\MM$, integrating $\b_\MM$ out under 
its prior density yields the marginal likelihood $p(z_\MM \mid \MM)$.
Last, the TBF is defined as the ratio
\begin{equation}
\text{TBF}_{\MM:{\Mnull}} 
= \frac{p(z_\MM \mid \MM)}{p(z_{\MM} \mid \Mnull)}
= \frac{\int p(z_\MM \mid \b_\MM, \MM) p(\b_\MM \mid \MM) d \b_\MM}{p(z_{\MM} \mid \Mnull)}.
\end{equation}

To conduct model selection in GLMs, \citet{Held_etal_2015} derive the TBF under 
the $g$-prior \eqref{eq:g-prior_SBH2011}, in whose density, $\b_\MM$ appears in the format of $\lambda_\MM$.
Thus the conjugacy permits a tractable marginal likelihood $p(z_\MM \mid \MM)$
as a Gamma distribution.
Therefore, the resulting TBF has a closed form expression as in \eqref{eq:TBF_g}.

\subsection{Comparing Data-Based and Test-Based Bayes Factors}\label{subsection:TBF_DBF}

The TBF \eqref{eq:TBF_g} has a similar expression to the DBF \eqref{eq:DBF_g}.
In fact, the two Bayes factors would be the same
if $z_\MM = Q_\MM$ and $\mathcal{J}_n( \hat{\alpha}_{\Mnull}) 
= \mathcal{J}_n( \hat{\alpha}_{\MM})$. 
Naturally, it is interesting to examine
how different the two Bayes factors are.

We compare DBF \eqref{eq:DBF_g} and TBF \eqref{eq:TBF_g} 
empirically through a logistic regression toy example,
with $g = n$ and a single covariate generated from independent standard normal distributions.
With the intercept set to $\alpha = 0.5$, three scenarios are studied with different coefficients
$\beta = 0, 20/\sqrt{n}, 2$, which correspond to the null, local alternative, and alternative, 
respectively. To study asymptotics, various sample sizes $n = 100, 500, 1000, 5000$ are taken.
For each combination of $\beta$ and $n$, 100 independent datasets are generated.
To obtain an accurate approximation to the DBF, 
in addition to the integrated Laplace approximation (ILA) formula  \eqref{eq:DBF_g},
we also implement importance sampling (IS), which can be viewed as a 
gold standard if the number of samples drawn is  large. 
Here we draw $m=10000$ samples $\alpha^{(t)}, \beta^{(t)}$, independently from 
Student-$t$ distributions with degrees of freedom $4$, with location and scale parameters  
matching those in the corresponding conditional posteriors 
\eqref{eq:betapost}, \eqref{eq:alphapost}.

\begin{figure}
\begin{center}
\includegraphics[width = \textwidth]{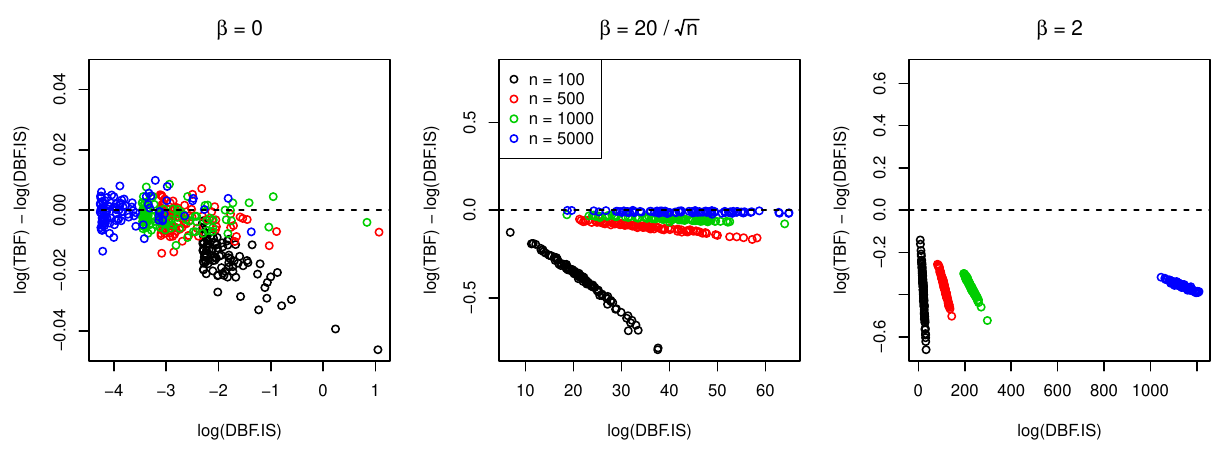}
\includegraphics[width = \textwidth]{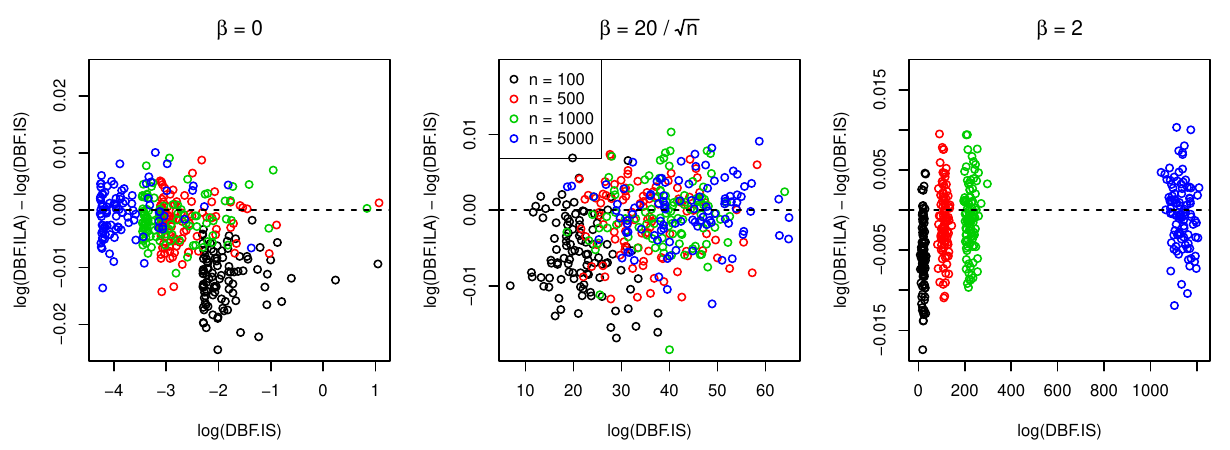}
\includegraphics[width = \textwidth]{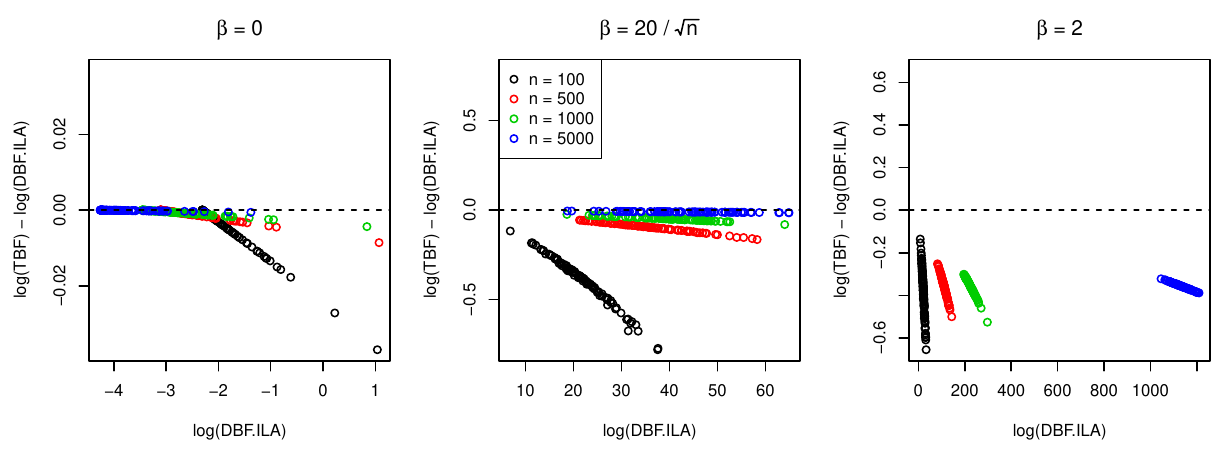}
\end{center}
\caption{From top to bottom: TBF versus DBF approximated by IS, 
DBF approximated by ILA vs DBF approximated by IS, and TBF versus DBF approximated by ILA.
From left to right: the null, local alternative, and alternative hypotheses. 
}
\label{fig:DBF_TBF}
\end{figure}

\begin{figure}
\begin{center}
\includegraphics[width = \textwidth]{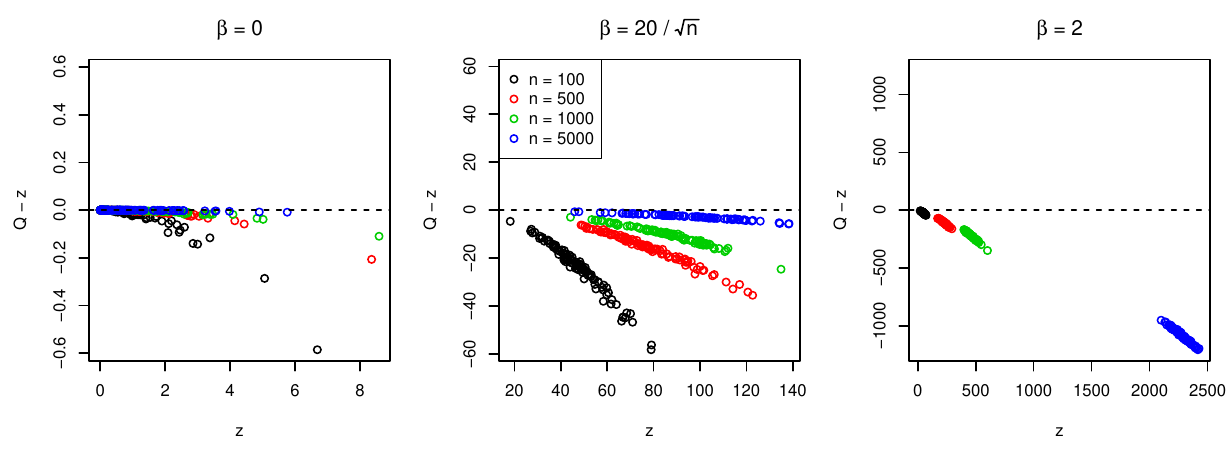}
\end{center}
\caption{Wald statistic $Q_\MM$ versus the deviance $z_\MM$.}
\label{fig:Q_z}
\end{figure}

Figure \ref{fig:DBF_TBF} shows that when the null or the local alternative is true, 
TBF \eqref{eq:TBF_g} is asymptotically the same
as the DBF computed  under either IS or ILA \eqref{eq:DBF_g}. 
In contrast, when the alternative is true, TBF differs from DBF by a relatively small 
but systematic amount.
Comparison between the Wald statistic $Q_\MM$ \eqref{eq:Q} and the deviance $z_\MM$ \eqref{eq:z} 
suggests a similar phenomenon (Figure \ref{fig:Q_z}).
They are asymptotically the same under the null or local alternative, but different under the alternative.

In addition to the similarity between the two Bayes factors under $g$-priors, 
we notice that as a function of $g$, the test-based marginal likelihood 
would have the same kernel 
$p(z_\MM \mid \MM) \propto (1+g)^{-p_\MM / 2} \exp\left( -z_\MM / [2(1+g)] \right)$
as its data-based counterpart \eqref{eq:marlik_fixedg}
if $z_\MM = Q_\MM$.
Therefore, all empirical Bayes and fully Bayes approaches on $g$,
discussed in Section \ref{subsection:EB} and Section \ref{section:mix_g},
can be readily applied to test-based methods with minimal changes.
\citet{Held_etal_2015} apply local empirical Bayes, 
$p(z_\MM \mid \MM) = \max_{g\geq 0} p(z_\MM \mid g, \MM)$,
and fully Bayes, $p(z_\MM \mid \MM) = \int p(z_\MM \mid g, \MM) p(g) dg$
to compute marginal likelihoods for TBFs.  
However, we find that these optimized and integrated versions of TBF may 
no longer be coherent, in the sense that results change with the choice of the baseline model.
Elaborating, when testing nested models $\MM_1 \subset \MM_2$, 
\[
\text{TBF}_{\MM_2:{\MM_1}} \neq \frac{\text{TBF}_{\MM_2:{\Mnull}}}{ \text{TBF}_{\MM_1:{\Mnull}}},
\] 
if one computes the left hand side TBF under baseline $\MM_1$,
but computes the right hand side TBFs under baseline $\Mnull$.
The main reason for this incoherence is that for model $\MM$,
unlike the data-based marginal likelihood which only depends on $\MM$ itself, 
the test statistic $z_\MM$ also depends on the baseline model.
On the other hand, coherence exists for the TBF \eqref{eq:TBF_g} under fixed $g$,
since $z_{\MM_2:\MM_1} = z_{\MM_2:\Mnull} - z_{\MM_1:\Mnull}$ \citep{Johnson_2008}.
Hence, change of baseline models does not affect the results of the TBF under fixed $g$, 
which is also the case with the DBF.

\section{Additional Simulation Examples}\label{section:poisson}

We first include some additional results from the logistic regression simulation example
that are examined in Section \ref{section:simulation} (see Table \ref{tb:selection_logistic_averagesize})
and then introduce a different simulation study on 
Poisson regressions.

\begin{table}[ht]
\centering
\caption[Logistic regression: model selection accuracy under 0-1 loss.]
{
Logistic regression simulation example: 
average size of selected models, out of 100 realizations.
} \label{tb:selection_logistic_averagesize}
\begin{tabular}{| l | rr|rr|rr|rr | rr | rr|}
  \hline
$p$ 			& \multicolumn{8}{c|}{20}	& \multicolumn{4}{c|}{100}	 \\ \hline
$p(\MM)$		&  \multicolumn{8}{c|}{Uniform}	&  \multicolumn{2}{c|}{Uniform}	 &  \multicolumn{2}{c|}{BB$(1,1)$}	\\ \hline
$p_\MT$ 		& \multicolumn{2}{c|}{0}	& \multicolumn{2}{c|}{5}	& \multicolumn{2}{c|}{10}	& \multicolumn{2}{c|}{20}	& \multicolumn{2}{c|}{5}	& \multicolumn{2}{c|}{5} \\ 
$r$ 				& 0		& 0.75	& 0		& 0.75	& 0		& 0.75	& 0		& 0.75	& 0		& 0.75		& 0		& 0.75		 \\ 
  \hline
CH$(a=1/2,b=n)$ & 0 & 0 & 5 & 4 & 10 & 8 & 17 & 13 & 17 & 15 & 5 & 3 \\ 
  CH$(a=1,b=n)$ & 0 & 0 & 5 & 5 & 10 & 8 & 17 & 13 & 18 & 15 & 5 & 3 \\ 
  CH$(a=1/2,b=n/2)$ & 0 & 0 & 6 & 5 & 10 & 9 & 17 & 14 & 25 & 20 & 5 & 3 \\ 
  CH$(a=1,b=n/2)$ & 0 & 0 & 6 & 5 & 10 & 9 & 17 & 14 & 26 & 22 & 5 & 3 \\ 
  Beta-prime & 0 & 0 & 5 & 4 & 10 & 8 & 17 & 13 & 19 & 15 & 5 & 3 \\ 
  ZS adapted & 0 & 0 & 5 & 5 & 10 & 8 & 17 & 13 & 18 & 15 & 5 & 3 \\ 
  Benchmark & 0 & 0 & 6 & 6 & 11 & 10 & 18 & 15 & 27 & 24 & 5 & 3 \\ 
  Robust & 0 & 0 & 6 & 5 & 11 & 9 & 18 & 14 & 34 & 30 & 21 & 10 \\ 
  Intrinsic & 0 & 0 & 6 & 5 & 11 & 9 & 18 & 14 & 32 & 30 & 14 & 5 \\ 
  Hyper-$g/n$ & 0 & 1 & 6 & 5 & 11 & 10 & 18 & 15 & 69 & 56 & 99 & 80 \\ 
  DBF, $g=n$ & 0 & 0 & 5 & 4 & 9 & 7 & 15 & 11 & 7 & 5 & 5 & 3 \\ 
  TBF, $g=n$ & 0 & 0 & 5 & 4 & 9 & 7 & 15 & 11 & 7 & 5 & 5 & 3 \\ 
  Jeffreys & 3 & 3 & 6 & 6 & 11 & 10 & 18 & 15 & 70 & 60 & 99 & 91 \\ 
  Hyper-$g$ & 4 & 4 & 6 & 6 & 11 & 10 & 18 & 15 & 70 & 61 & 100 & 93 \\ 
  Uniform & 4 & 4 & 7 & 6 & 12 & 10 & 18 & 15 & 70 & 61 & 100 & 97 \\ 
  Local EB & 19 & 19 & 6 & 6 & 11 & 10 & 18 & 15 & 71 & 60 & 100 & 96 \\ 
  AIC & 3 & 3 & 8 & 7 & 12 & 11 & 18 & 15 & 34 & 34 & 6 & 4 \\ 
  BIC & 0 & 0 & 5 & 4 & 9 & 7 & 15 & 11 & 7 & 5 & 5 & 3 \\ 
   \hline
\end{tabular}
\end{table}

The simulation setup of the Poisson regression example is similar to that
of the logistic regression in Section \eqref{section:simulation}.
True values of coefficients (including the intercept) are set to one-fifth of those
in the logistic regression, to avoid occasional extremely large values in $\mathbf{Y}$.
Tables \ref{tb:selection_poisson_top}-\ref{tb:selection_poisson_bma}
display model selection and parameter estimation performance. 
Comparison among priors on $\b_\MM$ leads to similar conclusions to 
the logistic regression example. 
For the Poisson regression, overall model selection accuracy is not as high
as the logistic regression when $\MT \neq \Mnull$, which is
likely due to the smaller magnitude of coefficients.

\setlength{\tabcolsep}{5pt}

\begin{table}[ht]
\centering
\caption[Poisson regression: model selection accuracy under 0-1 loss.]
{
Poisson regression simulation example:
number of times the true model are selected out of 100 realizations.
Column-wise maximum is in bold type.} \label{tb:selection_poisson_top}
\begin{tabular}{| l | rr|rr|rr|rr | rr | rr|}
  \hline
$p$ 			& \multicolumn{8}{c|}{20}	& \multicolumn{4}{c|}{100}	 \\ \hline
$p(\MM)$		&  \multicolumn{8}{c|}{Uniform}	&  \multicolumn{2}{c|}{Uniform}	 &  \multicolumn{2}{c|}{BB$(1,1)$}	\\ \hline
$p_\MT$ 		& \multicolumn{2}{c|}{0}	& \multicolumn{2}{c|}{5}	& \multicolumn{2}{c|}{10}	& \multicolumn{2}{c|}{20}	& \multicolumn{2}{c|}{5}	& \multicolumn{2}{c|}{5} \\ 
$r$ 				& 0		& 0.75	& 0		& 0.75	& 0		& 0.75	& 0		& 0.75	& 0		& 0.75		& 0		& 0.75		 \\ 
  \hline
  CH$(a=1/2,b=n)$ & 94 & 92 & 10 & {\bf 2} & 10 & 0 & 0 & 0 & 2 & 0 & 1 & 0 \\ 
  CH$(a=1,b=n)$ & 87 & 89 & 10 & {\bf 2} & 10 & 0 & 0 & 0 & 11 & 1 & 1 & 0 \\ 
  CH$(a=1/2,b=n/2)$ & 91 & 89 & {\bf 11} & {\bf 2} & 10 & 0 & 0 & 0 & 3 & 0 & 1 & 0 \\ 
  CH$(a=1,b=n/2)$ & 82 & 85 & {\bf 11} & {\bf 2} & 9 & 0 & 0 & 0 & 5 & {\bf 2} & 2 & 0 \\ 
  Beta-prime & 94 & 92 & 10 & {\bf 2} & 10 & 0 & 0 & 0 & 7 & 0 & 1 & 0 \\ 
  ZS adapted & 87 & 89 & 10 & {\bf 2} & 11 & 0 & 0 & 0 & 6 & 0 & 1 & 0 \\ 
  Benchmark & {\bf 97} & {\bf 93} & 7 & 0 & 12 & {\bf 1} & 0 & 0 & 4 & 0 & 1 & 0 \\ 
  Robust & 91 & 89 & 9 & {\bf 2} & 11 & 0 & 0 & 0 & 1 & 0 & 3 & 0 \\ 
  Intrinsic & 85 & 88 & 8 & {\bf 2} & 12 & {\bf 1} & 0 & 0 & 1 & 0 & 3 & 0 \\ 
  Hyper-$g/n$ & 84 & 87 & 9 & 0 & 12 & {\bf 1} & 0 & 0 & 1 & 0 & 3 & 0 \\ 
  DBF, $g=n$ & 84 & 88 & 7 & 0 & 8 & 0 & 0 & 0 & 11 & 0 & 1 & 0 \\ 
  TBF, $g=n$ & 84 & 88 & 7 & 0 & 8 & 0 & 0 & 0 & {\bf 14} & 0 & 1 & 0 \\ 
  Jeffreys & 0 & 0 & 7 & 1 & 12 & 1 & 0 & 0 & 0 & 0 & 3 & 0 \\ 
  Hyper-$g$ & 6 & 7 & 7 & 0 & {\bf 13} & {\bf 1} & 0 & 0 & 0 & 0 & 3 & 0 \\ 
  Uniform & 4 & 2 & 7 & 0 & {\bf 13} & {\bf 1} & 0 & 0 & 1 & 1 & 3 & 0 \\ 
  Local EB & 0 & 0 & 7 & 0 & {\bf 13} & {\bf 1} & 0 & 0 & 0 & 0 & 3 & 0 \\ 
  AIC & 4 & 4 & 3 & 0 & 6 & {\bf 1} & {\bf 1} & 0 & 0 & 0 & {\bf 8} & 0 \\ 
  BIC & 84 & 88 & 7 & 0 & 8 & 0 & 0 & 0 & 13 & 1 & 1 & 0 \\ 
   \hline
\end{tabular}
\end{table}

\begin{table}[ht]
\centering
\caption[Poisson regression: model selection accuracy under 0-1 loss.]
{
Poisson regression simulation example: 
average size of selected models, out of 100 realizations.
} \label{tb:selection_poisson_averagesize}
\begin{tabular}{| l | rr|rr|rr|rr | rr | rr|}
  \hline
$p$ 			& \multicolumn{8}{c|}{20}	& \multicolumn{4}{c|}{100}	 \\ \hline
$p(\MM)$		&  \multicolumn{8}{c|}{Uniform}	&  \multicolumn{2}{c|}{Uniform}	 &  \multicolumn{2}{c|}{BB$(1,1)$}	\\ \hline
$p_\MT$ 		& \multicolumn{2}{c|}{0}	& \multicolumn{2}{c|}{5}	& \multicolumn{2}{c|}{10}	& \multicolumn{2}{c|}{20}	& \multicolumn{2}{c|}{5}	& \multicolumn{2}{c|}{5} \\ 
$r$ 				& 0		& 0.75	& 0		& 0.75	& 0		& 0.75	& 0		& 0.75	& 0		& 0.75		& 0		& 0.75		 \\ 
  \hline
  CH$(a=1/2,b=n)$ & 0 & 0 & 4 & 3 & 9 & 5 & 13 & 7 & 12 & 7 & 3 & 2 \\ 
  CH$(a=1,b=n)$ & 0 & 0 & 4 & 3 & 9 & 5 & 13 & 7 & 13 & 8 & 3 & 2 \\ 
  CH$(a=1/2,b=n/2)$ & 0 & 0 & 5 & 3 & 9 & 6 & 13 & 8 & 16 & 10 & 3 & 2 \\ 
  CH$(a=1,b=n/2)$ & 0 & 0 & 5 & 3 & 9 & 6 & 13 & 8 & 17 & 10 & 3 & 2 \\ 
  Beta-prime & 0 & 0 & 4 & 3 & 9 & 5 & 13 & 7 & 13 & 7 & 3 & 2 \\ 
  ZS adapted & 0 & 0 & 4 & 3 & 9 & 5 & 13 & 7 & 13 & 7 & 3 & 2 \\ 
  Benchmark & 0 & 0 & 5 & 4 & 10 & 7 & 14 & 9 & 17 & 7 & 3 & 1 \\ 
  Robust & 0 & 0 & 5 & 3 & 9 & 6 & 14 & 8 & 20 & 14 & 3 & 2 \\ 
  Intrinsic & 0 & 0 & 5 & 3 & 10 & 6 & 14 & 8 & 22 & 13 & 3 & 2 \\ 
  Hyper-$g/n$ & 0 & 0 & 5 & 4 & 9 & 7 & 14 & 9 & 24 & 31 & 3 & 4 \\ 
  DBF, $g=n$ & 0 & 0 & 4 & 2 & 8 & 5 & 12 & 6 & 5 & 3 & 3 & 2 \\ 
  TBF, $g=n$ & 0 & 0 & 4 & 2 & 8 & 5 & 12 & 6 & 6 & 4 & 3 & 2 \\ 
  Jeffreys & 2 & 3 & 5 & 4 & 10 & 7 & 14 & 10 & 29 & 36 & 3 & 18 \\ 
  Hyper-$g$ & 3 & 4 & 5 & 5 & 10 & 7 & 15 & 10 & 30 & 37 & 3 & 24 \\ 
  Uniform & 4 & 4 & 6 & 5 & 10 & 7 & 15 & 10 & 30 & 38 & 3 & 34 \\ 
  Local EB & 19 & 19 & 5 & 5 & 10 & 7 & 15 & 10 & 32 & 74 & 3 & 76 \\ 
  AIC & 3 & 3 & 7 & 6 & 11 & 8 & 16 & 11 & 30 & 28 & 4 & 2 \\ 
  BIC & 0 & 0 & 4 & 2 & 8 & 5 & 12 & 6 & 5 & 3 & 3 & 2 \\ 
   \hline
\end{tabular}
\end{table}

\begin{table}[ht]
\centering
\caption[Poisson regression: model selection accuracy under 0-1 loss.]
{
Poisson regression simulation example: 
$1000$ times the average SSE
$=\sum_{j=1}^p (\tilde{\beta}_j - \beta_{j, \MT}^*)^2$ of 100 realizations.
Column-wise minimum is in bold type.
} \label{tb:selection_poisson_bma}
\begin{tabular}{| l | rr|rr|rr|rr | rr | rr|}
 \hline
$p$ 			& \multicolumn{8}{c|}{20}	& \multicolumn{4}{c|}{100}	 \\ \hline
$p(\MM)$		&  \multicolumn{8}{c|}{Uniform}	&  \multicolumn{2}{c|}{Uniform}	 &  \multicolumn{2}{c|}{BB$(1,1)$}	\\ \hline
$p_\MT$ 		& \multicolumn{2}{c|}{0}	& \multicolumn{2}{c|}{5}	& \multicolumn{2}{c|}{10}	& \multicolumn{2}{c|}{20}	& \multicolumn{2}{c|}{5}	& \multicolumn{2}{c|}{5} \\ 
$r$ 				& 0		& 0.75	& 0		& 0.75	& 0		& 0.75	& 0		& 0.75	& 0		& 0.75		& 0		& 0.75		 \\ 
  \hline
  CH$(a=1/2,b=n)$ & 5 & 8 & {\bf 24} & {\bf 61} & 34 & 120 & 58 & 198 & 66 & 132 & 37 & 103 \\ 
  CH$(a=1,b=n)$ & 6 & 9 & {\bf 24} & {\bf 61} & 34 & 120 & 58 & 197 & 66 & 134 & 37 & 98 \\ 
  CH$(a=1/2,b=n/2)$ & 7 & 11 & {\bf 24} & {\bf 61} & {\bf 33} & 116 & 56 & 188 & 75 & 148 & 36 & 97 \\ 
  CH$(a=1,b=n/2)$ & 7 & 13 & {\bf 24} & {\bf 61} & {\bf 33} & 115 & 55 & 187 & 77 & 135 & 36 & 94 \\ 
  Beta-prime & 5 & 8 & {\bf 24} & {\bf 61} & 34 & 120 & 58 & 197 & 66 & 132 & 37 & 103 \\ 
  ZS adapted & 6 & 9 & {\bf 24} & {\bf 61} & 34 & 119 & 55 & 197 & 66 & 125 & 37 & 99 \\ 
  Benchmark & 8 & 18 & 26 & 65 & {\bf 33} & {\bf 108} & 51 & 170 & 74 & 150 & 36 & 133 \\ 
  Robust & 7 & 13 & 25 & 63 & {\bf 33} & 115 & 51 & 183 & 88 & 182 & 36 & 97 \\ 
  Intrinsic & 8 & 14 & 25 & 63 & {\bf 33} & 115 & 51 & 182 & 90 & 183 & 35 & 94 \\ 
  Hyper-$g/n$ & 5 & 12 & 25 & 65 & {\bf 33} & 109 & 52 & 172 & 84 & 162 & 36 & 97 \\ 
  DBF, $g=n$ & 5 & 6 & 25 & 63 & 37 & 132 & 68 & 231 & {\bf 40} & {\bf 83} & 39 & 101 \\ 
  TBF, $g=n$ & 5 & 6 & 25 & 63 & 37 & 132 & 68 & 231 & {\bf 40} & 84 & 39 & 101 \\ 
  Jeffreys & 4 & 9 & 26 & 67 & {\bf 33} & {\bf 108} & 51 & 169 & 87 & 165 & 35 & 97 \\ 
  Hyper-$g$ & 4 & 7 & 26 & 68 & {\bf 33} & {\bf 108} & 51 & {\bf 168} & 87 & 164 & {\bf 34} & 112 \\ 
  Uniform & 3 & 7 & 26 & 70 & {\bf 33} & {\bf 108} & 51 & {\bf 168} & 87 & 164 & {\bf 34} & 121 \\ 
  Local EB & {\bf 2} & {\bf 4} & 26 & 71 & {\bf 33} & {\bf 108} & 51 & {\bf 168} & 99 & 256 & {\bf 34} & 222 \\ 
  AIC & 17 & 40 & 28 & 74 & 34 & 115 & {\bf 46} & 171 & 120 & 284 & 37 & {\bf 79} \\ 
  BIC & 5 & 6 & 25 & 63 & 37 & 132 & 68 & 231 & {\bf 40} & 84 & 39 & 100 \\ 
   \hline
\end{tabular}
\end{table}

\end{document}